\documentclass[reprint,
nofootinbib,
 amsmath,amssymb,
]{revtex4-2}

\usepackage{graphicx}% Include figure files
\usepackage{bm}% bold math
\usepackage{physics}
\usepackage{bbm}
\usepackage{mathrsfs}
\usepackage{amsthm} %package for theorems env
\usepackage[toc,page]{appendix}
\usepackage{hyperref}
\global\long\def\EE{\mathbb{E}}
\global\long\def\PP{\mathbb{P}}
\global\long\def\FF{\mathbb{F}}
\global\long\def\11{\mathbbm{1}}

\global\long\def\+{\oplus}
\global\long\def\P{\mathsf{P}}

\DeclareMathOperator*{\medcap}{\mathbin{\scalebox{1}{\ensuremath{\bigcap}}}}
\DeclareMathOperator*{\medcup}{\mathbin{\scalebox{1}{\ensuremath{\bigcup}}}}
\usepackage{stackengine}
\def\deq{\mathrel{\ensurestackMath{\stackon[1pt]{=}{\scriptstyle\Delta}}}}

\global\long\def\tensor{\otimes}
 \def\tr{\Tr}
 \def\id{\text{id}}
\global\long\def\bigtensor{\bigotimes}
\allowdisplaybreaks% makes automatic splits of equations

\def\TDeltaN{\mathcal{T}_{\delta}^{(n)}}

\def\PiA{\Pi_{\rho_A}}
\def\PiB{\Pi_{\rho_B}}

\def\PiuA{\Pi_{u^n}^A}
\def\PivB{\Pi_{v^n}^B}
\def\PiwA{\Pi_{w^n}}
\def\lambdauA{\lambda^A_{u^n}}
\def\lambdavB{\lambda^B_{v^n}}
\def\lambdawA{\lambda_{w^n}}
\def\lambdaAB{\lambda^{AB}_{u^n,v^n}}
\def\rhohatuA{\hat{\rho}^A_{u^n}}
\def\rhohatvB{\hat{\rho}^B_{v^n}}
\def\rhohatwA{\hat{\rho}_{w^n}}
\def\rhohatuAvB{\hat{\rho}^{AB}_{u^n,v^n}}

\def\TDeltan{\mathcal{T}_{\delta}^{(n)}}
\def\IndiU1{\mathbbm{1}_{\left\{U^{n,(\mu_1)}(l)=u^n\right\}} }
\def\IndiV1{\mathbbm{1}_{\left\{V^{n,(\mu_2)}(k)=v^n\right\}} }
\def\TDeltaN{\mathcal{T}_{\delta}^{(n)}}

\def\CalR{\mathcal{R}}

\def\sigTilde{\tilde{\sigma}_x}
\newcommand{\eleInSet}[2]{#1\in \mathcal{#2}}
\newcommand{\zeroWithDelta}[1]{#1(\delta) \searrow 0 \text{ as } \delta \searrow 0}
\def\UCodeword{U^{n,(\mu_1)}(a_1,i)}
\def\VCodeword{V^{n,(\mu_2)}(a_2,j)}
\def\UBarCodeword{U^{n,(\bar\mu_1)}(a_1,i)}
\def\VBarCodeword{V^{n,(\bar\mu_2)}(a_2,j)}
\def\WCodeword{W^{n,(\mu)}(a,i)}
\def\WCodewordm{W^{n,(\mu)}(a,m)}
\def\CutOffA{\Pi_A^{\mu_1}}
\def\CutOffB{\Pi_B^{\mu_2}}
\def\CutOffBarA{\Pi_A^{\bar\mu_1}}
\def\CutOffBarB{\Pi_B^{\bar\mu_2}}
\def\CutOff{\Pi^{\mu}}
\def\gammaCoeff{\gamma_{u^n}^{(\mu_1)}}
\def\gammaBarCoeff{\gamma_{u^n}^{(\bar{\mu}_1)}}
\def\gammaWCoeff{\gamma_{w^n}^{(\mu)}}
\def\zetaCoeff{\zeta_{v^n}^{(\mu_2)}}
\def\zetaBarCoeff{\zeta_{v^n}^{(\bar{\mu}_2)}}
\def\rhotilduA{\tilde{\rho}_{u^n}^A}
\def\rhotildvB{\tilde{\rho}_{v^n}^B}
\def\rhotildwA{\tilde{\rho}_{w^n}}
\def\PiTauWn{\Pi_{w^n}^{\theta}}
\def\PiRho{\Pi_{\rho}}
\def\LambdaTildew{\tilde{\lambda}_{w^n}}

\newtheorem{theorem}{Theorem}
\newtheorem{lem}{Lemma}

\newtheorem{proposition}{Proposition}

\newtheorem{definition}{Definition}
\newtheorem{example}{Example}

\newtheorem{remark}{Remark}

\begin{document}

\preprint{APS/123-QED}

\title{\huge Distributed Quantum Faithful Simulation and Function Computation Using Algebraic Structured Measurements}% Force line breaks with \\
\thanks{This work was presented in part at IEEE International Symposium on Information Theory (ISIT) 2021.}

\author{Touheed Anwar Atif and S. Sandeep Pradhan                                             }
%  \altaffiliation[Also at ]{Physics Department, XYZ University.}%Lines break automatically or can be forced with \\
% \author{S. Sandeep Pradhan}%
\affiliation{%
 University of Michigan,  Ann Arbor\\
 touheed@umich.edu, pradhanv@umich.edu
}%

% \collaboration{MUSO Collaboration}%\noaffiliation

%  \homepage{http://www.Second.institution.edu/~Charlie.Author}
% \affiliation{
%  Second institution and/or address\\
%  This line break forced% with \\
% }%
% \affiliation{
%  Third institution, the second for Charlie Author
% \affiliation{%
%  Authors' institution and/or address\\
%  This line break forced with \textbackslash\textbackslash
% }%

% \collaboration{CLEO Collaboration}%\noaffiliation


\begin{abstract}
  In this work, we consider the task of faithfully simulating a quantum measurement, acting on a joint bipartite quantum state, in a distributed manner. In the distributed setup,  the constituent sub-systems of 
  the joint quantum state are measured by two agents, Alice and Bob.
  A third agent, Charlie receives
  the measurement outcomes sent 
  by Alice and Bob.  Charlie  uses local and pairwise shared randomness to compute a bivariate function of the measurement outcomes. 
  The objective of three agents is to faithfully simulate the given distributed quantum measurement acting on the given quantum state
  while minimizing the communication and shared  randomness rates.   We  demonstrate a new achievable quantum information-theoretic rate-region
  that exploits the bivariate function using random structured POVMs based on asymptotically good algebraic codes. The algebraic structure of these codes is matched to that of the bivariate function that models the action of Charlie.   The conventional approach for this class of problems has been to reconstruct individual measurement outcomes corresponding to Alice and Bob, at Charlie, and then compute the bivariate function. This is achieved using mutually independent approximating POVMs based on  random unstructured codes.
  In the present approach, using algebraic structured POVMs,
  the computation is performed on the fly, thus obviating the need to reconstruct individual measurement outcomes at Charlie.
Using this,  we show that a strictly larger rate region can be achieved. The performance limit is characterized using single-letter quantum mutual information quantities. 
  We provide examples to illustrate the information-theoretic gains attained by endowing POVMs with algebraic structure. One of the challenges in analyzing these structured POVMs is that they exhibit only pairwise independence and induce only uniform single-letter distributions. To address this, we use nesting of algebraic codes and develop 
 a covering lemma applicable to pairwise-independent POVM ensembles.
 Combining these techniques, we provide a multi-party distributed faithful simulation and function computation protocol.  
 
%   In this work, we consider the task of faithfully simulating a distributed quantum measurement  and demonstrate a new achievable information-theoretic rate-region for the communication and common. We accomplish this  by developing a technique of randomly generating structured POVMs using algebraic codes. One of the major challenges toward employing these structured POVMs is that they only exhibit pairwise independence. To overcome this
%   we develop a Pruning Trace inequality which is a tighter version of the known operator we develop a covering lemma which is applicable to pairwise-independent operators.  We demonstrate rate gains for this problem and also for the problem of distributed function computation over traditional coding schemes.  Combining these techniques, we provide a multi-party distributed faithful simulation and function computation protocol. 
\end{abstract}

%\keywords{Suggested keywords}%Use showkeys class option if keyword , independent of the operator Chernoff inequality, and
\maketitle

%\tableofcontents

\section{Introduction}\label{sec:introduction}
%\section{Introduction}
% Measurements provide an interface to the intricate quantum world 
%with the conceivable macroscopic classical world 
% by associating classical attributes to quantum states. However, quantum phenomena, such as superposition, entanglement and non-commutativity contribute to uncertainty in measurement outcomes. As a result, these 
% Measurements provide an interface between intricate quantum world and the conceivable macroscopic classical world.
Measurement compression is one of the foremost and fundamental  quantum information processing techniques which form the basis of many quantum protocols \cite{devetak2008resource}.
% Quantum measurements can be treated as special channels, with their input as quantum states and output as classical bits. A key concern from an information-theoretic viewpoint, is to characterize the amount of  ``relevant information" conveyed by a measurement about a quantum state. This is accomplished, in its fullest potential, not just by compressing the classical outputs  obtained from these measurements, but also by eliminating the inherent redundancy present within these quantum measurements.
One of the seminal works in this regard was by Winter \cite{winter}, where he performed a novel information theoretic analysis to compress measurements in an asymptotic sense. 
%Winter's  measurement compression theorem \cite{winter} (also elaborated in \cite{wilde_e}) quantifies the ``relevant information" as the amount of resources needed to simulate the output of a quantum measurement applied on a given state in an asymptotic sense. 
The measurement compression problem  formulated in \cite{winter} is as follows.
Consider an agent (Alice) who performs a measurement $M$ on a quantum state $\rho$, and sends a set of classical bits to another agent  (Bob). Bob intends to \textit{faithfully} recover the outcomes of Alice’s measurements without having access to $\rho$, while preserving the correlation with the post-measured state of Alice's reference. 
The major contribution of this work (as elaborated in \cite{wilde_e}) was in specifying an optimal rate region in terms of classical communication and common randomness needed to \textit{faithfully simulate} the action of repeated independent measurements performed on many independent copies of the given quantum state.

  Wilde et al. \cite{wilde_e} extended  the measurement compression problem by considering additional resources available to each of the participating parties. One such formulation allows Bob to process the information received from Alice using local private randomness. 
%   Recently, this work was extended in \cite{anshu2019convex}
%   In analogy with \cite{bennett2009quantum}, this problem formulation is referred to as the non-feedback measurement simulation, while the former is termed as simulation with feedback. Further, the problem of measurement compression in the presence of quantum side information was also studied in \cite{wilde_e}. 
The authors here also combined the ideas from \cite{winter} and \cite{devetak2003classical}  to simulate a measurement in presence of quantum side information.
% Recently, authors in \cite{anshu2019convex}. 
% provided a protocol for this problem based on convex-split and position based decoding, and bounded rates in terms of smooth max and hypothesis testing relative entropies (defined in \cite{anshu2019convex}).
% we first provide a brief overview of the major components involved in the above mentioned faithful simulation problems. 
In the above problem formulations, authors have derived the results using the prevalent random coding techniques analogous to Shannon's unstructured random codes \cite{shannon} involving mutually independent codewords.
The point-to-point setup \cite{winter,wilde_e} requires  randomly generating approximating POVMs and
% and proving that these POVMs are indeed realizable, i.e., the POVM elements form a positive resolution of the Identity operator, and (ii) 
analyzing the error associated with
% substituting the tensor product POVMs with 
these approximating POVMs, also termed as ``covering error''.  The key analytical tool that facilitates this is the operator Chernoff bound \cite{ahlswede2002strong},  which crucially exploits the mutual independence of codewords, 
yielding the quantum covering lemma \cite[Lemma 17.2.1]{Wilde_book}. 
% The distributed case additionally involves a third component of analysis (iii)  bounding of error induced by randomly and independently binning the POVM elements. Further, these three components are proved as follows. 

% A similar result is proved for the classical setting in \cite{atif2020source}.

The measurement compression problem has been studied extensively. Early works on quantifying the information gain of a measurement include \cite{groenewold1971problem,lindblad1972entropy,ozawa1986information}.
Buscemi et al. \cite{buscemi2008global,luo2010information,shirokov2011entropy} later advocated quantum mutual information with respect to a classical-quantum state  as the measure to characterize the corresponding information gain.
% provided by a measurement when acting on a particular state. 
Berta et al. \cite{berta2014identifying} generalized the Winter’s measurement compression theorem by developing a universal measurement compression theorem for arbitrary inputs, and identified the quantum mutual information of a measurement as the information gained by performing the measurement, independent of the input state on which it is performed. They provide a proof based on new ``classically coherent state merging protocol" - a variation of the quantum state merging protocol \cite{horodecki2005partial,horodecki2007quantum}, and the post-selection technique for quantum channels \cite{christandl2009postselection}.

Anshu et al. \cite{anshu2019ConvexSplit} considered the problem of measurement compression with side information in the one-shot setting. They presented a protocol by proposing a new convex-split lemma for classical-quantum states and employing the position based decoding, and bounded the communication in terms of smooth max and hypothesis testing relative entropies. The original convex-split lemma \cite{anshu2017quantum,anshu2014quantum} demanded sub-optimal shared-randomness rate in the one-shot setting, by requiring large amount of additional quantum states in its statement. The authors addressed this by modifying the lemma to only use pairwise independent random variables. This substantially simplified the derandomization required, leading to an exponential reduction in the randomness cost in comparison to \cite{anshu2017quantum}.
Considering a related problem, Renes and Renner \cite{renes2012one} also studied sending of classical messages in the presence of quantum side information in the one-shot setting.
For more discussion and results pertaining to one-shot quantum information theory, the reader is directed to \cite{tomamichel2015quantum,khatri2020principles}.
% The proposed convex-split lemma countered the problem of sub-optimal rate of shared randomness required by the one-shot protocol by only requiring pairwise independence among the random variables leading to substantial reduction in the randomness cost. 

% Our new statement of the convex-split lemma implies that the codebook needed for this protocol solely requires pairwise independent random variables. This consid- erably simplifies the derandomization task, because the support size of pairwise independent random variables is exponentially smaller than independent random variables. Similar arguments apply to the applications of our techniques given in the work [34], to various network setting in classical information theory.

% (Convex-split and hypothesis testing approach to one-shot quantum measurement compression and randomness extraction)
% (Talk about pairwise independent random variables )
% Anshu, Devabathini, Jain, 
% (quantum measurement compression applications)
% Talk about pairwise independent random variables 

% Tomamichel on finite resources (quantum information processing) 
% Marco Tomamichel

Furthermore, the authors in \cite{atif2019faithful} considered the task of quantifying ``relevant information'' for the quantum measurements performed in a distributed fashion on bipartite entangled states involving three agents.
% , in an asymptotic sense. 
In this multi-terminal setting, a composite bipartite quantum system $AB$ is made available to two agents, Alice and Bob, where they have access to the sub-systems $A$ and $B$, respectively. Two separate measurements, one for each sub-system, are performed in a distributed fashion with no communication taking place between Alice and Bob. A third party, Charlie, is connected to Alice and Bob via two separate classical links. The objective of the three parties is to simulate the action of repeated independent measurements performed on many independent copies of the given composite state. Further, common randomness at rate $C$ is also shared amidst the three parties.
% This work was further extended in \cite{atif2019faithful} to the setting where Eve has access to additional private randomness. Similar to the earlier point-to-point setting, the problem in \cite{atif2019faithful} was also referred ot as the non-feedback version of the distributed problem. 
% The authors  also use additional private randomness resource to incorporate an added stochastic processing at the decoder, and prove that this reduces the required shared randomness constraint. 
This is achieved 
using random unstructured code ensembles while still using the operator Chernoff bound. 

The measurement compression theorem has found its applications in several quantum information processing protocols. Examples include the quantum reverse Shannon theorem \cite{bennett2002entanglement,bennett2009quantum,berta2011quantum}, local purity distillation protocols \cite{horodecki2003local,horodecki2005local,devetak2005distillation,krovi2007local}, and  also in the grandmother protocol \cite{devetak2008resource} which is useful in entanglement distillation from noisy quantum states.

An ubiquitous application of distributed systems in current quantum settings arises due to the inherent vulnerability of the large-scale quantum computation systems to noise. The state-of-art systems exhibits technical difficulties in increasing  the number of low-noise qubits in a single quantum device. A solution to this is cooperative processing of information separately on spatially segregated units. This necessitates the need for distributed compression protocols to compress efficiently and recover the data. In addition, when one is interested in solely reconstructing functions of the distributively stored quantum data, the rate of communication may be further reduced by employing structured coding techniques. For this, we need to impose further structure on these POVMs. This is to ensure that the joint decoder (Charlie) is able to reconstruct a lower dimensional quantum state with minimal use of the classical communication resource. Hence, structure of the POVM is desired to match with the structure of the function being computed.

The traditional random coding techniques using unstructured code ensembles may not always achieve optimality  for distributed multi-terminal settings. For instance, the work by Korner-Marton \cite{korner1979encode} demonstrated this sub-optimality for the problem of classical distributed lossless compression with the objective of computing the sum of  the sources  for the binary symmetric  case using random linear codes.
Traditionally, algebraic-structured codes are used in information coding problems toward achieving computationally efficient (polynomial-time) encoding and decoding algorithms.  However, in multi-terminal communication problems, even if computational complexity is a non-issue, random algebraic 
structured codes outperform random unstructured codes in terms of achieving improved asymptotic rate regions
in many cases
\cite{krithivasan2011distributed,200710TIT_NazGas,200906TIT_PhiZam,201109TITarXiv_JafVis}.
% Later, the authors in \cite{krithivasan2011distributed} also established for a general distributed source coding problem that structured random codes offer rate gains over their unstructured counterparts. 

%by employing the quantum distributed function computation algorithms. 
% Moreover, with the recent discussions on the development of quantum internet, processing and storing would imperatively become distributed, facilitating a dramatic growth in the applicability of these algorithms.

Motivated by this, we consider the quantum distributed faithful measurement simulation problem and present a new achievable rate-region using algebraic structured coding techniques.
However, there are two main challenges in using these algebraic structured codes toward an asymptotic analysis in  quantum information theory. The first challenge is 
to be able to induce arbitrary empirical single-letter distributions. For example, if we were to send codewords from a linear code with uniform probability, then the induced  empirical distribution of codeword symbols
(single-letter distribution on
the symbols of the codewords) is uniform. To address this challenge,
% i.e., to induce arbitrary single-letter distributions for the above quantum simulation problem, 
we use a collection of cosets of a linear code called Unionized Coset Codes (UCCs) \cite{pradhanalgebraic}. The second challenge is that unlike the random unstructured codes, the codewords generated from a random linear code are only pairwise-independent \cite{Gal-ITRC68}. This renders the above technique of operator Chernoff bound, or even the covering lemma, unusable. 
 Since our approach relies on the use of UCCs for generating the approximating POVMs, the binning of these POVM elements is performed in a correlated fashion as governed by these structured codes. This is in contrast to the common technique of  independent binning. Due to the correlated binning, the pairwise-independence issue gets exacerbated.

We address these challenges  using three main ideas summarized as follows:
\begin{itemize}
	\item \textbf{Random structured generation of pruned POVMs} - We generate a collection of algebraic structured approximating POVMs randomly using the above described UCC technique, and then prune them.
% 	using a onto the subspace of the Identity operator.
This pruning ensures that these POVMs form a positive resolution of identity, and thus eliminates any need for the operator Chernoff inequality. However, such pruning comes at the cost of additional  approximating error.
% 	\item \textbf{Pruning Trace Inequality} - 
To bound the approximating error caused by pruning the POVMs, we develop a new Operator Inequality which provides a handle to convert the pruning error in the form of covering error expression (dealt within the next idea). 
% We later prove that this new inequality is indeed a  tightening of the existing Operator Markov Inequality \cite{book_quantum}. 
	\item \textbf{Covering Lemma for Pairwise-Independent Ensemble} - Since the traditional covering lemma is based on the Chernoff inequality, we develop an alternative proof for the aforementioned covering lemma \cite[Lemma 17.2.1]{book_quantum}. This alternative proof is based on the second-order analysis using the operator trace inequalities and hence requires the operators to be only pairwise-independent.
	\item \textbf{Multi-partite Packing Lemma} - We develop a binning technique for performing computation on the fly so as to achieve a low dimensional reconstruction of a function at the location of Charlie. In an effort towards analysing this binning technique, we develop a multi-partite  packing Lemma for the structured POVMs.
	
\end{itemize} 

	Combining these techniques, we provide a multi-party distributed faithful simulation and function computation protocol in a quantum information theoretic setting. We provide a characterization of the asymptotic performance limit of this protocol in terms of a computable single-letter achievable rate-region, which is the main result of the paper (see Theorem \ref{thm:dist POVM appx}).

The organization of the paper is as follows.
 In Section \ref{sec:prelim}, we set the notation, state requisite definitions and also provide related results. In Section \ref{sec:dist} we state our main result on the distributed measurement compression and provide the theorem (Theorem \ref{thm:dist POVM appx}) characterizing the rate-region. In Section \ref{sec:coveringLemma} we provide a new Covering Lemma for pairwise-independent ensembles. Section \ref{sec:novelLemmas} provides useful lemmas. In Section \ref{sec:p2p}, we consider the point-to-point setup and provide a theorem characterizing the rate-region
 using algebraic structured codes.  We prove the main result (Theorem \ref{thm:dist POVM appx}) in Section \ref{sec:proof of thm dist POVM} using the point-to-point result as a building block. Finally, we conclude the paper in Section \ref{sec:conclusion}.

\section{Preliminaries}\label{sec:prelim}
\noindent \textbf{Notation:} Given any natural number $M$, let the finite set $\{1, 2, \cdots, M\}$ be denoted by $[1,M]$. Let $\mathcal{B(H)}$ denote the algebra of all bounded linear operators acting on a finite dimensional Hilbert space $\mathcal{H}$. Further, let $ \mathcal{D(H)} $ denote the set of all unit trace positive operators acting on $ \mathcal{H} $. Let $I$ denote the identity operator. The trace distance between two operators $A$ and $B$ is defined as $\|A-B\|_1\deq \Tr|A-B|$, where for any operator $\Lambda$ we define $|\Lambda|\deq \sqrt{\Lambda^\dagger \Lambda}$. The von Neumann entropy of a density operator $\rho \in \mathcal{D}(\mathcal{H})$ is denoted by $S(\rho)$.  The quantum mutual information for a bipartite density operator $\rho_{AB} \in \mathcal{D}(\mathcal{H}_A \otimes \mathcal{H}_B)$ is defined as 
% The quantum mutual information and conditional entropy for a bipartite density operator $\rho_{AB} \in \mathcal{D}(\mathcal{H}_A \otimes \mathcal{H}_B)$ are defined, respectively, as 
\begin{align*}
I(A;B)_{\rho}&\deq S(\rho_{A})+S(\rho_{B})-S(\rho_{AB}).
% S(A|B)_{\rho}&\deq S(\rho_{AB})-S(\rho_{B}).
\end{align*} 
% Given any ensemble $\{p_i, \rho_i\}_{i\in [1,m]}$, the Holevo information, as in \cite{holevo}, is defined as 
% \begin{align*}
%     \chi \big( \{p_i, \rho_i\}\big) \deq S\Big(\sum_i p_i \rho_i\Big) - \sum_i p_i S(\rho_i).
% \end{align*}
 A positive-operator valued measure (POVM) acting on a Hilbert space $\mathcal{H}$ is a collection $M\deq \{\Lambda_x\}_{x \in \mathcal{X}}$ of positive operators in $\mathcal{B}(\mathcal{H})$ that form a resolution of the identity:
\begin{align*}
\Lambda_x\geq 0, \forall x \in \mathcal{X}, \qquad \sum_{x \in \mathcal{X}} \Lambda_x=I,
\end{align*}
where $\mathcal{X}$ is a finite set.
If instead of the equality above, the inequality $\sum_x \Lambda_x \leq I$ holds, then the collection is said to be a sub-POVM. A sub-POVM $M$ can be completed to form a POVM, denoted by $[M]$, by adding the operator $\Lambda_{0}\deq (I-\sum_x \Lambda_x )$ to the collection.  Let $\Psi^\rho_{RA}$ denote a purification of a density operator $\rho \in D(\mathcal{H}_A)$. Given a POVM $M\deq \{\Lambda^A_x\}_{x \in \mathcal{X}}$ acting on  $\rho$, the post-measurement state of the reference together with the classical outputs is represented by 
\begin{equation}\label{eq:1}
     (\text{id} \tensor M)(\Psi^\rho_{RA})\deq \sum_{x\in \mathcal{X}} \ketbra{x}\tensor \tr_{A}\{(I^R \tensor \Lambda_x^A ) \Psi^\rho_{RA} \}.
\end{equation} 
%where $\Psi^\rho_{RA}$ is a purification of $\rho$.
Consider two POVMs $M_A=\{\Lambda^A_x\}_{x \in \mathcal{X}}$ and $M_B=\{\Lambda^B_y\}_{y \in \mathcal{Y}}$ acting on $\mathcal{H}_A$ and $\mathcal{H}_B$, respectively. Define $M_A\tensor M_B \deq \{\Lambda^A_x\tensor \Lambda^B_y\}_{x\in \mathcal{X},y\in \mathcal{Y}}$
% .as a the collection of all operators of the form $\Lambda^A_x\tensor \Lambda^B_y,  \text{for all } x,y$.
With this definition, $M_A\tensor M_B$ is a POVM acting on $\mathcal{H}_A\tensor \mathcal{H}_B$. By $M^{\tensor n}$ denote the $n$-fold tensor product of the POVM $M$ with itself. 
For a prime $p$, we denote the unique finite field of size $p$ by $\FF_p$, and denote the addition operation over the field by $+$.

\begin{definition}[Faithful simulation \cite{wilde_e}]\label{def:faith-sim}
	Given a POVM ${M}\deq \{\Lambda_x\}_{x\in \mathcal{X}}$ acting on a Hilbert space $\mathcal{H}$ and a density operator $\rho\in \mathcal{D}(\mathcal{H})$, a sub-POVM $\tilde{M}\deq \{\tilde{\Lambda}_{x}\}_{x\in \mathcal{X}}$ acting on $\mathcal{H}$ is said to be $\epsilon$-faithful to $M$ with respect to $\rho$, for $\epsilon > 0$, if the following holds: 
	\begin{equation}\label{eq:faithful-sim-cond-1_2}
	\sum_{x\in \mathcal{X}} \Big\|\sqrt{\rho} (\Lambda_{x}-\tilde{\Lambda}_{x}) \sqrt{\rho}\Big\|_1+\tr\left\{(I-\sum_{x} \tilde{\Lambda}_{x})\rho\right\}   \leq \epsilon.
	\end{equation}
\end{definition}

\begin{lem}\label{lem:Separate}
	Given a density operator $ \rho_{AB} \in \mathcal{D}(\mathcal{H}_{A}\tensor\mathcal{H}_{B}) $, a sub-POVM
	$M_Y \deq \left \{\Lambda_y^B: y \in \mathcal{Y}\right \} $ acting on $ \mathcal{H}_B,  $ for some set $\mathcal{Y}$, 
	and any Hermitian operator $\Gamma^A$ acting on $ \mathcal{H}_A $, we have
	\begin{align}\label{eq:lemSeparate1}
	\sum_{y \in \mathcal{Y}}\left\| \sqrt{\rho_{AB}}\left (\Gamma^A\tensor \Lambda_y^B\right )\sqrt{\rho_{AB}}\right \|_1 \leq \left \| \sqrt{\rho_A} \Gamma^A\sqrt{\rho_{A}} \right \|_1,
	\end{align} with equality if $ \displaystyle \sum_{y \in \mathcal{Y}}\Lambda_y^B = I  $, where $ \rho_A = \Tr_{B}\{\rho_{AB}\}$.
\end{lem}
\begin{proof}
	The proof is provided in Lemma 3 of \cite{atif2019faithful}. 
\end{proof}

\section{Main Results}\label{sec:mainResults}

In this section we present the main results of this paper. 

\subsection{Simulation of Distributed POVMs using Algebraic-Structured POVMs}\label{sec:dist}
%\section{Simulation of Distributed POVMs with Stochastic Processing}\label{sec:nf_dist_POVM}
%In this section, we dCharlielop a stochastic processing  variant of the distributed POVM simulation problem described in Section \ref{sec:Appx_POVM}. 
% Consider a bipartite composite quantum system $(A,B)$ represented by Hilbert Space $\mathcal{H}_A\tensor \mathcal{H}_B$.
Let $\rho_{AB}$ be a density operator acting on a composite Hilbert Space  $\mathcal{H}_A\tensor \mathcal{H}_B$.
Consider two measurements $M_A$ and $M_B$ on sub-systems $A$ and $B$, respectively. 
Imagine again that we have three parties, named Alice, Bob and Charlie, that are trying to collectively simulate the action of a given measurement $M_{AB}$
% = \{\Lambda_z\}_{z\in \mathcal{Z}}$ 
performed on the state $\rho_{AB}$, as shown in Fig. \ref{fig:NF_Distributed}. Charlie additionally has access to unlimited private randomness.
The problem is defined in the following.
		 
\begin{figure}[hbt]
	\begin{center}
	  \includegraphics[scale=0.4]{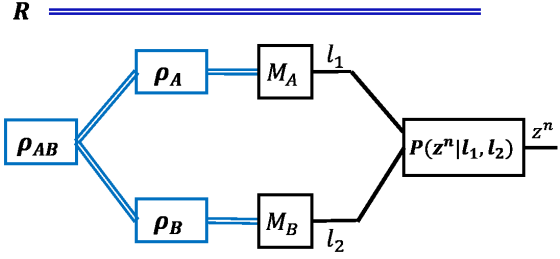}
    \caption{The diagram depicting the distributed POVM simulation problem with stochastic processing. In this setting, Charlie additionally has access to unlimited private randomness.}
    %of a distributed quantum measurement applied to a bipartite quantum system $AB$. A tensor product measurement $M_A\tensor M_B$ is performed on many copies of the observed quantum state. The outcomes of the measurements are given by two classical bits. The receiver functions as a classical-to-quantum channel $\beta$ mapping the classical data to a quantum state.} 
    \label{fig:NF_Distributed}  
	\end{center}
\end{figure}
\begin{definition}\label{def:distProtocol}
	For a given finite set $\mathcal{Z}$, and a Hilbert space $\mathcal{H}_{A}\tensor \mathcal{H}_B$, a distributed protocol with stochastic processing with parameters $(n,\Theta_1,\Theta_2,N_1,N_2)$ 
	is characterized by %a collections of POVM-pairs  $\tilde{M}_A^{(\mu)}$ and $\tilde{M}_B^{(\mu)},\mu \in [1, N]$, each 1) acting on $\mathcal{H}_A^{\tensor n}$ and $\mathcal{H}_B^{\tensor n}$, and 2) having at most $\Theta_1$ and $\Theta_2$ outcomes, respectively.
	\\
		\textbf{1)} a collections of Alice's sub-POVMs  $\tilde{M}_A^{(\mu_1)},\mu_1 \in [1, N_1]$ each acting on $\mathcal{H}_A^{\tensor n}$ and with outcomes in a subset $\mathcal{L}_1$ satisfying $|\mathcal{L}_1|\leq \Theta_1$.\\
		\textbf{2)} a collections of Bob's sub-POVMs  $\tilde{M}_B^{(\mu_2)},\mu_2 \in [1, N_2]$ each acting on $\mathcal{H}_B^{\tensor n}$ and with outcomes in a subset $\mathcal{L}_2$, satisfying $|\mathcal{L}_2|\leq \Theta_2$.\\
		\textbf{3)} a collection of Charlie's classical stochastic maps $P^{(\mu_1,\mu_2)}(z^n|l_1, l_2)$ for all $l_1\in \mathcal{L}_1, l_2\in \mathcal{L}_2, z^n\in \mathcal{Z}^n$, $\mu_1 \in [1,N_1]$ and $\mu_2 \in [1,N_2]$. \\
		The overall  sub-POVM of this distributed protocol, given by $\tilde{M}_{AB}$, is characterized by the following operators: 
		\begin{align*}
		\tilde{\Lambda}_{z^n}\deq \frac{1}{N_1}\frac{1}{N_2}\sum_{\mu_1,\mu_2}&\sum_{l_1, l_2} P^{(\mu_1,\mu_2)} (z^n|l_1,l_2) \\ 
		& \Lambda^{A,(\mu_1)}_{l_1}\tensor \Lambda^{B,(\mu_2)}_{l_2}, \quad \forall z^n \in \mathcal{Z}^n, \nonumber
		\end{align*}
		where $\Lambda^{A,(\mu_1)}_{l_1}$ and  $\Lambda^{B,(\mu_2)}_{l_2}$ are the operators corresponding to the sub-POVMs $\tilde{M}_A^{(\mu_1)}$ and $\tilde{M}_B^{(\mu_2)}$, respectively.
	%
	% $M_A^{(n, \mu)}$ and $M_B^{(n, \mu)}),\mu \in [1, 2^{nC}]$ acting on $\mathcal{H}_A^{\tensor n}$ and $\mathcal{H}_B^{\tensor n}$, respectively.
	%
	%
	%1. a collection of POVMs $M_A^{(n, \mu)},\mu \in [1, 2^{nC}]$ acting on $\mathcal{H}_A^{\tensor n}$ with at most $2^{nR_1}$ outcomes. 
	%2. a collection of POVMs $M_B^{(n, \mu)},\mu \in [1, 2^{nC}]$ acting on $\mathcal{H}_B^{\tensor n}$ with at most $2^{nR_2}$ outcomes.
	%3. a collection of mappings  $\beta^{(\mu)}: [1:2^{nR_1}]\times [1:2^{nR_2}] \mapsto \mathcal{X}^n$
	%
	%
	%a collection of $2^{nC}$ POVM-pairs $(M_A^{(n, \mu)}, M_B^{(n, \mu)}),\mu \in [1, 2^{nC}]$ acting on $\mathcal{H}^{\tensor n}_{AB}$ and with at most $2^{nR_1}$ and $2^{nR_2}$ outcomes. mappings $\beta^{(\mu)}: \mathcal{X}^n\mapsto \mathcal{Y}^n$, and 
\end{definition}
In the above definition,  $(\Theta_1,\Theta_2)$ determines the amount of classical bits communicated from Alice and Bob to Charlie. The amount of pairwise shared randomness is determined by  $N_1$ and $N_2$. The classical stochastic maps $P^{(\mu_1,\mu_2)}(z^n|l_1,l_2)$ represent the action of Charlie on the received classical bits.  
\begin{definition}
	Given a POVM $M_{AB}$ acting on  $\mathcal{H}_{A}\tensor \mathcal{H}_B$, and a density operator $\rho_{AB}\in \mathcal{D}(\mathcal{H}_{A}\tensor \mathcal{H}_B)$, a quadruple $(R_1,R_2,C_1,C_2)$ is said to be achievable, if for all $\epsilon>0$ and for all sufficiently large $n$, there exists a distributed protocol with stochastic processing with parameters  $(n, \Theta_1, \Theta_2, N_1,N_2)$ such that its overall sub-POVM $\tilde{M}_{AB}$  is $\epsilon$-faithful to $M_{AB}^{\tensor n}$ with respect to $\rho_{AB}^{\tensor n}$ (see Definition \ref{def:faith-sim}), and
	\begin{align*}
	\frac{1}{n}\log_2 \Theta_i \leq R_i+\epsilon, \! \quad \mbox{and} \! \quad\!
	\frac{1}{n}\log_2 N_i \leq C_i+\epsilon, \quad i=1,2.
	\end{align*}
	The set of all achievable quadruples $(R_1,R_2,C_1,C_2)$ is called the achievable rate region. 
\end{definition}
\begin{definition}[Joint Measurements]\label{def:Joint Measurements} A POVM $M_{AB} = \{\Lambda^{AB}_z\}_{z \in \mathcal{Z}}$, acting on 
a Hilbert space 
 $\mathcal{H}_{A}\tensor \mathcal{H}_B$, is said to have a separable decomposition with stochastic integration  
 given by $(\bar{M}_A,\bar{M}_B,P_{Z|S,T})$
 if there exist POVMs $\bar{M}_A=\{\bar{\Lambda}^A_{s}\}_{s \in \mathcal{S}}$ and $\bar{M}_B=\{ \bar{\Lambda}^B_{t}\}_{t \in \mathcal{T}}$ and a stochastic mapping $P_{Z|S,T}: \mathcal{S} \times \mathcal{T} \rightarrow\mathcal{Z}$ such that 
	\begin{equation}
	\Lambda^{AB}_{z}=\sum_{s,t} P_{Z|S,T}(z|s,t) \bar{\Lambda}^A_{s}\tensor \bar{\Lambda}^B_{t}, \quad  \forall z \in \mathcal{Z}, \nonumber
	\end{equation}
	where $\mathcal{S}, \mathcal{T}$,  and $\mathcal{Z}$ are finite sets.
% 	Further, if the mapping $P_{Z|S,T}$ is a deterministic function, then the POVM is said to have a separable decomposition with deterministic integration. 
\end{definition}
The following theorem provides an inner bound to the {achievable} rate region, which is proved in Section \ref{sec:proof of thm dist POVM}. This is one of the main results of this paper.

\begin{theorem}\label{thm:dist POVM appx}
		Consider a density operator $\rho_{AB}\in \mathcal{D}(\mathcal{H}_{A}\tensor \mathcal{H}_B)$, and a POVM $M_{AB}=\{\Lambda^{AB}_{z}\}_{z \in \mathcal{Z}}$ acting on  $\mathcal{H}_{A}\tensor \mathcal{H}_B$ having a separable decomposition with stochastic integration  (as in Definition \ref{def:Joint Measurements}), yielding POVMs $\bar{M}_A=\{\bar{\Lambda}^A_{s}\}_{s \in \mathcal{S}}$ and $\bar{M}_B=\{ \bar{\Lambda}^B_{t}\}_{t \in \mathcal{T}}$ and a stochastic map $P_{Z|S,T}:\mathcal{S}\times\mathcal{T} \rightarrow \mathcal{Z}$. Define the auxiliary states
\begin{align*}
    \sigma_{1}^{RSB} &\deq (\emph{id}_R\tensor \bar{M}_{A} \tensor \emph{id}_B) ( \Psi^{\rho_{AB}}_{R A B}), \\
    \sigma_{2}^{RTV} &\deq (\emph{id}_R\tensor \emph{\id}_A \tensor \bar{M}_{B}) ( \Psi^{\rho_{AB}}_{R A B}), \quad \text{and} \nonumber \\
    \sigma_3^{RSTZ}   &\deq \sum_{s,t,z}\sqrt{\rho_{AB}}\left (\bar{\Lambda}^A_{s}\tensor \bar{\Lambda}^B_{t}\right )\sqrt{\rho_{AB}} \\
    & \hspace{30pt}\tensor P_{Z|S,T}(z|s,t) \ketbra{s}\tensor\ketbra{t}\tensor\ketbra{z},
\end{align*}
for some orthonormal sets $\{\ket{s}\}_{s\in\mathcal{S}}, \{\ket{t}\}_{t\in\mathcal{T}}$, and $\{\ket{z}\}_{z\in\mathcal{Z}}$, where $\Psi^{\rho_{AB}}_{R A B}$ is a purification of $\rho_{AB}$.
		A quadruple $(R_1,R_2,C_1,C_2)$ is achievable if there exists a finite field $\mathbb{F}_p$, for a prime $p$, a pair of mappings $f_S:\mathcal{S} \rightarrow \mathbb{F}_p$ and 
$f_T:\mathcal{T} \rightarrow \mathbb{F}_p$, and a stochastic mapping $P_{Z|W}:\mathbb{F}_p \rightarrow \mathcal{Z}$ such that 
\[
P_{Z|S,T}(z|s,t)\!=\!P_{Z|W}(z|f_S(s)+f_T(t)), \forall s\! \in \!\mathcal{S}, t \!\in \!\mathcal{T}, z \!\in \!\mathcal{Z},
\]
yielding $U=f_S(S)$, $V=f_T(T)$, and $W=U+V$, 
and the following inequalities are satisfied:
		\begin{subequations}\label{eq:nf_dist POVm appx rates}
			\begin{align}
%			R_1 &\geq H(U+V) - S(U|RB)_{\sigma_1},\\
%			R_2 &\geq H(U+V) - S(U|RA)_{\sigma_2},\\
%			\textcolor{red}{R_1+R_2} &\geq I(U;RB)_{\sigma_1}+I(V;RA)_{\sigma_2}-I(U;V)_{\sigma_3},\label{eq:nfrate3}\\
%			R_1 + C &\geq H(U+V) - S(U|RVZ)_{\sigma_3},\\
%			R_2 + C & \geq H(U+V)  - S(V|RZ)_{\sigma_3}, \\
%			R_1+R_2+C & \geq 2H(U+V) - S(U|RVZ)_{\sigma_3} -S(V|RZ)_{\sigma_3}, \label{eq:nfrate4}
			&R_1 \geq I(U;R,B)_{\sigma_1}+I(W;V)_{\sigma_3} - I(U;V)_{\sigma_3},\\
			&R_2 \geq I(V;R,A)_{\sigma_2}+I(W;U)_{\sigma_3}-I(U;V)_{\sigma_3},\\
% 			\textcolor{red}{R_1+R_2} &\geq I(U;RB)_{\sigma_1}+I(V;RA)_{\sigma_2}-I(U;V)_{\sigma_3},\label{eq:nfrate3}\\
			&R_1 + C_1 \geq I(U;\!R,Z)_{\sigma_3} +I(W;V)_{\sigma_3}-I(U;V)_{\sigma_3} ,\\
			&R_2 + C_2 \geq I(V;\!R,Z)_{\sigma_3} +I(W;U)_{\sigma_3} -I(U;V)_{\sigma_3}, \\
			&R_1+R_2+C_1+C_2 \geq I(U,V;R,Z)_{\sigma_3}  + I(W;U)_{\sigma_3} \nonumber \\ & \hspace{1.4in}
			 +I(W;V)_{\sigma_3} 
			- I(U;V)_{\sigma_3}. \label{eq:nfrate4}
			\end{align}
		\end{subequations}
% % 		the auxiliary states $\sigma_{1}^{RUB} \deq (\emph{id}_R\tensor \bar{M}_{A} \tensor \emph{id}_B) ( \Psi^{\rho_{AB}}_{R A B}), \sigma_{2}^{RAV} \deq (\emph{id}_R\tensor \emph{\id}_A \tensor \bar{M}_{B}) ( \Psi^{\rho_{AB}}_{R A B}),$  
% and	$\sigma_3^{RUVZ}   \deq \sum_{u,v,z}\sqrt{\rho_{AB}}\left (\bar{\Lambda}^A_{u}\tensor \bar{\Lambda}^B_{v}\right )\sqrt{\rho_{AB}} \tensor P_{Z|W}(z|u+ v) \ketbra{u}\tensor\ketbra{v}\tensor\ketbra{z},$ for some orthonormal sets $\{\ket{u}\}_{u\in\mathcal{U}}, \{\ket{v}\}_{v\in\mathcal{V}}$, and $\{\ket{z}\}_{z\in\mathcal{Z}}$. 
\end{theorem}

\begin{proof}
A proof is provided in Section \ref{sec:proof of thm dist POVM}.
\end{proof}

\begin{remark}
Note that the rate-region obtained in Theorem 6 of  \cite{atif2019faithful} using unstructured random code ensembles, contains the constraint $R_1 + R_2 + C_1 + C_2 \geq I(U,V;R,Z)_{\sigma_3}$. Hence when 
\begin{align*}
I(W;U)_{\sigma_3} &+I(W;V)_{\sigma_3}- I(U;V)_{\sigma_3} \\ &=2S(U+V)_{\sigma_3} - S(U,V)_{\sigma_3}<0,    
\end{align*}

the above theorem gives a lower sum rate constraint. As a result, the rate-region above contains points that are not contained within the rate-region provided in \cite{atif2019faithful}. To illustrate this fact further, consider the following example.
\end{remark}

\begin{remark}
In the above theorem, we restrict our attention to prime finite fields for ease of exposition. The results can be generalized to arbitrary finite fields in a straight-forward manner.
\end{remark}

\begin{example}

Suppose the composite state $\rho_{AB} $ is described using one of the Bell states on $\mathcal{H}_A\tensor\mathcal{H}_B$ as
% Let us assume Alice and Bob each receive a quantum particle from an entangled pair of two quantum particles.
% Let the 
% % pair of entangled quantum particles Consider a quantum bipartite particle (say a pair of photons) with its 
% joint state of this entangled pair be described using the Bell state $\Phi^{AB}$ on the Hilbert space $\mathcal{H}_A\tensor\mathcal{H}_B$, where $\Phi^{AB}$ is  given as
\begin{align*}
    \rho_{AB} = \cfrac{1}{2}\left(\ket{00}_{AB} + \ket{11}_{AB} \right)\left( \bra{00}_{AB} + \bra{11}_{AB}\right).
\end{align*}
% Alice and Bob each have access to the partial sub-systems (A and B, respectively) of the state $\Phi^{AB}$ and hence perceive the 
Since $\pi^A = \Tr_B{\rho^{AB}}$ and $\pi^B = \Tr_A{\rho^{AB}}$, Alice and Bob would perceive each of their particles in maximally mixed states $\pi^A = \frac{I_A}{2}$ and $\pi^B = \frac{I_B}{2}$, respectively. Upon receiving the quantum state, the two parties wish to independently measure their states, using identical POVMs $\bar{M}_A$ and $\bar{M}_B$, given by $\bar{M}_A \deq \left\{ \bar{\Lambda}_s^A \right\}_{s\in \mathcal{S}}, \bar{M}_B \deq \left\{\bar{\Lambda}_v^B\right\}_{t\in\mathcal{T}}$, where 
$\mathcal{S}=\mathcal{T}=\{0,1\}$, and 
\begin{align}
    \Lambda_0^A &= \Lambda_0^B \deq \begin{bmatrix}
                    0.9501 & 0.0826 +i0.1089 \\
                    0.0826 - i0.1089 & 0.0615
                \end{bmatrix}, \nonumber\\
    \Lambda_1^A &= \Lambda_1^B \deq \begin{bmatrix}
                    0.0499 & -0.0826 -i0.1089 \\
                    -0.0826 + i0.1089 & 0.9385
                \end{bmatrix}. \nonumber
\end{align}
Alice and Bob together with Charlie are trying to simulate the action of $M_{AB} \deq \left\{\Gamma_z^{AB}\right\}_{z\in
\mathcal{Z}}$,  using the classical communication and common randomness as the resources available to them, where $\mathcal{Z}=\{0,1\}$, and 
\begin{align}
    \Gamma_z^{AB} \deq \sum_{s\in\{0,1\}}\sum_{t\in\{0,1\}}P_{Z|S,T}(z|s,t)\left(\Lambda_s^A\tensor\Lambda_t^B \right),
\end{align} 
for $z \in \{0,1\},$ and $P_{Z|S,T}(0|0,0) =P_{Z|S,T}(0|1,1)= 1-P_{Z|S,T}(0|0,1)=
1-P_{Z|S,T}(0|1,0)= \lambda$, with 
$\lambda \in (0,1)$. 
% \begin{align}
%     \Gamma_z^{AB} \deq \lambda_z\left(\Lambda_0^A\tensor\Lambda_0^B + \Lambda_1^A\tensor\Lambda_1^B \right) + (1-\lambda_z) \left(\Lambda_1^A\tensor\Lambda_0^B + \Lambda_0^A\tensor\Lambda_1^B \right), \text{ for } z \in \mathcal{Z}, \quad \text{and} \quad \lambda_z \in (0,1).
% \end{align} 
Note that the above POVM $M_{AB}$ admits a separable decomposition as defined in the statement of Theorem \ref{thm:dist POVM appx}
% as given in Definition \ref{def:Joint Measurements} 
with respect to the prime finite field $\FF_2$, with $U=S$ and $V=T$, and 
\[
P_{Z|W}(0|0)=1-P_{Z|W}(0|1)=\lambda.
\]
Hence the above theorem can be employed. This gives
\begin{align}
    S(U+V)_{\sigma_3} &= 0.5155, \quad S(U)_{\sigma_3} = S(V)_{\sigma_3} = 0.9999, \nonumber \\ S(U,V)_{\sigma_3} &= 1.5154, \quad I(U,V)_{\sigma_3} = 0.4844,\nonumber
\end{align}
where $\sigma_3$ is as defined in the statement of Theorem \ref{thm:dist POVM appx}.
Since $S(U)_{\sigma_3}-S(U+V)_{\sigma_3} = S(V)_{\sigma_3}-S(U+V)_{\sigma_3} = I(U,V)_{\sigma_3}, $
 the constraints on $R_1$, $R_2$, $R_1+C$ and $R_2+C$ are the same as obtained in Theorem 6 of \cite{atif2019faithful}. However, with $2S(U+V)_{\sigma_3} - S(U,V)_{\sigma_3} = -0.4844 < 0$, the constraint on $R_1+R_2+C_1 + C_2$ in the above theorem \eqref{eq:nfrate4} is strictly weaker than the constraint obtained using random unstructured codes in Theorem 6 of \cite{atif2019faithful}. Therefore, the rate-region obtained above
 using random structured codes in Theorem \ref{thm:dist POVM appx} is strictly larger than the rate-region in Theorem 6 of \cite{atif2019faithful}. 
%  For further discussion on the comparison of the two rate-regions see \cite{krithivasan2011distributed}.
% Consider a measurement $M_{AB}$ defined as $M_{AB}\deq M_A\tensor M_B$ which acts on a state $\rho_{AB}\in\mathcal{D}(\mathcal{H}_A\tensor\mathcal{H}_B)$. Let $M_A \deq \{\Lambda_u^A\}_{u\in\FF_2}$ and  $M_B \deq \{\Lambda_v^B\}_{v\in\FF_2}$, where $\FF_2$ denotes a binary field with elements $0$ and $1$. Further, let $M_{AB},$ and $\rho_{AB}$ be such that we obtain the distribution $P_{UV}$ as $P_{UV}(0,0) = 0.003920, P_{UV}(0,1) = 0.976080$ and $P_{UV}(1,0) = 0.019920$. For this, we can compute $S(U+V) = 0.0376223 < \frac{1}{2}S(U,V) = 0.08953145.$

\end{example}

\begin{example}
For the same state $\rho_{AB}$ as in the above example, consider the following identical POVMs $M_A \deq \left\{ \bar{\Lambda}_s^A \right\}_{s\in \mathcal{S}}$ and  $M_B \deq \left\{\bar{\Lambda}_t^B\right\}_{t \in \mathcal{T}}$, where $\mathcal{S}=\mathcal{T}=\{0,1\}$, and
\begin{align}
    \Lambda_0^A &= \Lambda_0^B = \begin{bmatrix}
                    0.4974 & 0.0471 + i0.4975 \\
                    0.0471 - i0.4975 & 0.5026
                \end{bmatrix}, \nonumber\\
    \Lambda_1^A &= \Lambda_1^B = \begin{bmatrix}
                    0.5026 & -0.0471 - i0.4975 \\
                    -0.0471 + i0.4975 &  0.4974
                \end{bmatrix}. \nonumber
\end{align}
Let the joint measurement that Alice and Bob are trying to simulate be given by
\begin{align}
    \Gamma_z^{AB} \deq \sum_{s\in\{0,1\}}\sum_{t\in\{0,1\}}P_{Z|S,T}(z|s,t)\left(\Lambda_s^A\tensor\Lambda_t^B \right), 
\end{align} 
for $z \in \{0,1\}$ where $P_{Z|S,T}:  \{0,1\}\rightarrow [0,1] $ is a conditional PMF on $\mathcal{Z}\times\mathcal{S}\times\mathcal{T}$ with $  P_{Z|S,T}(0|0,0) = \delta_0 \in (0,1)$ and $P_{Z|S,T}(0|0,1) = P_{Z|S,T}(0|1,0) = P_{Z|S,T}(0|1,1) = \delta_1 \in (0,1)$. Note that $P_{Z|S,T}$ depends on 
the variables  $(s,t)$ only through $s \lor t$, the logical OR function. 
% for $h:\{0,1\} \rightarrow [0,1]$.
% and $\lor$ denotes the logical OR operation.
Now, we define the random variables $U$ and $V$ on the prime finite field $\FF_3$ with the identity mappings $U = S$ and $V = T$, while noting that $U$ and $V$ take values in $\FF_3$ with $P(U=2)=P(V=2)=0$. Now with $W = U+V$, we identify the mapping $P_{Z|W}$ as 
\begin{align}
    P_{Z|W}(0|0) =\delta_0 ,\quad P_{Z|W}(0|1) =  P_{Z|W}(0|2) =\delta_1.
\end{align}
% where $h(0,0) = 0, h(0,1) = h(1,0) = 1$ and $h(1,1) = 2$.
For this identification, we obtain $2S(U+V) - S(U,V) = -0.9039 < 0$, which gives the constraint on $R_1+R_2+C_1+C_2 $ in the above theorem \eqref{eq:nfrate4} strictly weaker than the corresponding constraint obtained using random unstructured codes in Theorem 6 of \cite{atif2019faithful}. Since this is a biting constraint, the above rate-region is strictly larger than the former for this example.
\label{ex:ex2}
\end{example}

\begin{example}
Building upon Example \ref{ex:ex2}, we explore more points in the POVM space such that the above theorem provides constraints \eqref{eq:nfrate4} that are strictly weaker than the corresponding constraint obtained in Theorem 6 of \cite{atif2019faithful}. For this, we consider the same state $\rho_{AB},$ as above and the following identical POVMs $M_A \deq \left\{ \bar{\Lambda}_s^A \right\}_{s\in \mathcal{S}}$ and  $M_B \deq \left\{\bar{\Lambda}_t^B\right\}_{t \in \mathcal{T}}$, where $\mathcal{S}=\mathcal{T}=\{0,1\}$, and
\begin{align}
    \Lambda_0^A = \Lambda_0^B = \begin{bmatrix}
                    \theta_1 & \theta_2 + i\theta_3 \\
                    \theta_2 - i\theta_3 & 1-\theta_1
                \end{bmatrix}, \quad
    \Lambda_1^A = \Lambda_1^B = I - \Lambda_0^A \nonumber
\end{align}
for $\theta_i \in [-1,1]$\footnote{The above parametrization is only for illustrative purposes and do not constitute all the two dimensional POVMs.}.
Figure \ref{fig:exampleSurface} illustrates the surface where $2S(U+V) = S(U,V)$ and therefore the region inside the surface has $2S(U+V) - S(U,V) < 0$, where the POVMs obtained provides the constraint on $R_1+R_2+C_1+C_2 $ in the above theorem \eqref{eq:nfrate4} strictly weaker than the corresponding constraint obtained in Theorem 6 of \cite{atif2019faithful}.
\begin{figure}[hbt]
	\begin{center}
	  \includegraphics[scale=0.4]{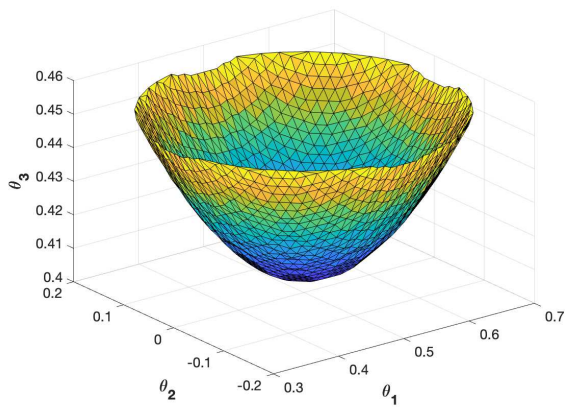}
    \caption{ Shown above is a $(\theta_1,\theta_2,\theta_3)$-surface with POVMs satisfying $2S(U+V) = S(U,V)$. Although the surface is symmetric in $\theta_3$, but for the ease of illustration only the upper half of the surface is shown.}
    %of a distributed quantum measurement applied to a bipartite quantum system $AB$. A tensor product measurement $M_A\tensor M_B$ is performed on many copies of the observed quantum state. The outcomes of the measurements are given by two classical bits. The receiver functions as a classical-to-quantum channel $\beta$ mapping the classical data to a quantum state.} 
    \label{fig:exampleSurface}  
	\end{center}
\end{figure}
\label{ex:ex3}
\end{example}

\begin{remark}
Note that for POVMs contained in the above $(\theta_1,\theta_2,\theta_3)$-surface of Example \ref{ex:ex3}, the sum rate constraint  $R_1+R_2+C_1+C_2 $ is strictly weaker than the corresponding constraint in \cite[Theorem 6]{atif2019faithful}, and vice-versa outside. One can employ a strategy based on superposition and successive encoding that combines the two coding techniques to yield a unified rate-region.
\end{remark}

\subsection{Covering Lemma with Change of Measure for Pairwise-Independent Ensemble }
\label{sec:coveringLemma}
The proof of the theorem is based on a construction of algebraic-structured POVM 
ensemble where the elements are only pairwise independent and not mutually independent.
% To analyze these POVMs we retreat back to first principles and develop a new one-shot Covering Lemma based on a change of measure technique and a second order analysis. 
% This lemma, which can be of independent interest, is one of the main contributions of this work.
To analyze these POVMs we retreat back to first principles and develop a new one-shot Covering Lemma based on a change of measure technique and a second order analysis. 
This lemma, which can be of independent interest, is one of the main contributions of this work.

% Here we present a covering lemma for any pairwise independent ensemble with a change of measure, based on which the subsequent information-theoretic inner bounds are derived. 
\begin{lem}[Covering Lemma]\label{lem:Change Measure Soft Covering Variance Based}
	Let $\{\lambda_x,\sigma_x\}_{x\in \mathcal{X}}$ be an ensemble, with  $\sigma_x \in \mathcal{D}(\mathcal{H})$ for all $x \in \mathcal{X}$, 	$\mathcal{X}$ being a finite set, and $\sigma = \sum_{x \in \mathcal{X}}\lambda_x\sigma_x$. Further, suppose we are given a total subspace projector $\Pi$ and a collection of codeword subspace projectors $\{\Pi_{x}\}_{x \in \mathcal{X}}$ which satisfy the following hypotheses
	\begin{subequations}\label{Soft_Covering-constraints}
	\begin{align}
		 \Tr{\Pi\sigma_x} & \geq 1-\epsilon, \label{Soft_Covering-constraints1} \\
		 \Tr{\Pi_x\sigma_x} & \geq 1-\epsilon, \label{Soft_Covering-constraints2} \\
		 \|\Pi\sqrt{\sigma}\|^2_1 & \leq D, \label{Soft_Covering-constraints3} \\
		 \Pi_x\sigma_x\Pi_x & \leq \frac{1}{d}\Pi_x, \quad \text{and} \label{Soft_Covering-constraints4} \\
		 \Pi_x\sigma_x\Pi_x & \leq \sigma_x. \label{Soft_Covering-constraints5}
	\end{align}
  	\end{subequations}
	for some $\epsilon \in (0,1)$ and $d < D$.
	Let $M $ be a finite non-negative integer. 
	Additionally, assume that there exists some set $\bar{\mathcal{X}} $ containing $\mathcal{X}$, with $\sigma_x \deq 0 $ (null operator) and $\lambda_x \deq 0$ for $x \in \bar{\mathcal{X}}\backslash \mathcal{X}$. Suppose $\{\mu_{\bar{x}}\}_{\bar{x} \in \bar{\mathcal{X}}}$ be any distribution on the set $\bar{\mathcal{X}}$ such that the distribution is $\{\lambda_x\}_{{x}\in \mathcal{X}}$ is absolutely continuous with respect to the distribution $\{\mu_{\bar{x}}\}_{\bar{x} \in \bar{\mathcal{X}}}$. Further, assume that $\lambda_x/\mu_x \leq \kappa$ for all $\eleInSet{x}{X}.$ Let a random covering code $\mathbbm{C}\deq \{C_m\}_{m \in [1,M]}$ be defined as a collection of codewords $C_m$ that are chosen pairwise independently according to the distribution $\{\mu_{\bar{x}}\}_{\bar{x} \in \bar{\mathcal{X}}}$.
% 	Let a random covering code $\mathbbm{C}\deq \{C_m\}_{m \in [1,M]}$ consists of codewords $C_m$ that are chosen pairwise independently according to the distribution $\lambda_x$. Let $\{\frac{1}{M},\sigma_{C_m}\}_{m \in [1,M]}$ be an ensemble obtained by uniformly choosing the states corresponding to the codewords $\{C_m\}_{m \in [1,M]}$. 
	Then we have
	\begin{align}
        \EE_{\mathbbm{C}}\left[\Big\|\sum_{x\in\bar{\mathcal{X}}}\lambda_x\sigma_x - \frac{1}{M}\sum_{m=1}^M \frac{\lambda_{C_m}}{\mu_{C_m}}\sigma_{C_m}\Big\|_1\right] \!\leq \!\sqrt{\frac{\kappa D}{Md}} +2\delta(\epsilon),
	\end{align}
	where $\delta(\epsilon) = 4\sqrt{\epsilon}$. Futhermore, for $\Tilde{\sigma}_x$
 defined as $\Tilde{\sigma}_x \deq \Pi\Pi_x{\sigma}_x\Pi_x\Pi$, we have
 \begin{align}
        \EE_{\mathbbm{C}}\left[\Big\|\sum_{x\in\bar{\mathcal{X}}}\lambda_x\tilde{\sigma}_x - \frac{1}{M}\sum_{m=1}^M \frac{\lambda_{C_m}}{\mu_{C_m}}\tilde{\sigma}_{C_m}\Big\|_1\right] \leq \sqrt{\frac{\kappa D}{Md}}. \label{lem:ChangeMeasure:ineq2}
	\end{align}
 \end{lem}
\begin{proof}
The proof is provided in Appendix \ref{appx:ProofUniformSoftCovering}
\end{proof}

%     Consider the above ensemble $\{\lambda_x,\sigma_x\}_{x\in \mathcal{X}}$ and the operators $\Pi$ and $\{\Pi_x\}_{\eleInSet{x}{X}}$ satisfying the above hypotheses. Additionally, suppose that there exists some subset $\mathcal{T} \subset \mathcal{X},$ such that $max_{\eleInSet{x}{T}}{\lambda_x} \leq c$ and $\sum_{x \notin \mathcal{T}}\lambda_x \leq \epsilon_{\lambda}$ for some $c \in (0,1)$ and $\epsilon_{\lambda} > 0$. Let a random covering code $\mathbbm{C}\deq \{C_m\}_{m \in [1,M]}$ be instead defined as a collection of codewords $C_m$ that are chosen pairwise independently according to the uniform distribution $\frac{1}{|\mathcal{X}|}$. Then we have
%     \begin{align}
%         \EE_{\mathbbm{C}}\left[\Big\|\sigma - \frac{|\mathcal{X}|}{M}\sum_{m=1}^M \lambda_{C_m}\sigma_{C_m}\Big\|_1\right] \leq \sqrt{\frac{c|\mathcal{X}|D}{Md}} +2\big(\epsilon_{\lambda} +\delta(\epsilon)\big),
% 	\end{align}
% 	where $\delta(\epsilon) = 4\sqrt{\epsilon}$.

\section{Useful Lemmas}\label{sec:novelLemmas}
In this section we present a few lemmas which will be used extensively in the sequel. 
\begin{definition}[Pruning Operators] \label{def:pruningOperator} Consider an operator $ A \geq 0 $ acting on Hilbert space $ \mathcal{H}_A. $ We say that a projector $ P $ prunes $ A $ with respect to Identity $ I_A $ on $ \mathcal{H}_A ,$ if $ P $ is a projector on to the non-negative eigenspace of $ I_A-A $.
\end{definition}

\subsection{Pruning Trace Inequality}
\begin{lem} \label{lem:newMarkov0} Consider a random  operator $ X \geq 0 $ acting on a Hilbert space $ \mathcal{H}_A. $ Let $ P $ be a pruning operator for $ X $ with respect to $ I_A $, as in Definition \ref{def:pruningOperator}. Then we have
$$\EE[\Tr{I_A - P}] \leq \EE[\Tr{X}]. $$
\end{lem}
\begin{proof}
The proof follows by noting that $\Tr{I_A - P} \leq \Tr{X}$.
\end{proof}

\begin{remark}
	To demonstrate the significance of this inequality,  we compare it with the popular Operator Markov Inequality \cite{book_quantum}. 
%	Let $ X = \sum_{\eleInSet{x}{X}}\mu_x\ketbra{x} $ for some orthogonal basis $ \{\ketbra{x}\}_{\eleInSet{x}{X}}$. 
We know from Operator Markov inequality
	\begin{align}
	\PP\left(X \nleq I_A \right) \leq \EE[\Tr{X}]. \nonumber
	\end{align}
	One can observe that $\11_{\{X \nleq I_A\}} \leq \Tr{I_A - P} $. Taking expectation, we obtain
	\begin{align}
		 \PP\left(X \nleq I_A \right)  \leq \EE[\Tr{I_A - P}].\nonumber
	\end{align}
	Moreover, one can also note that $\Tr{I_A - P} \leq \Tr{X}$, and expectation gives
	 \begin{align}
	\EE[\Tr{I_A - P}] \leq \EE[\Tr{X}].\nonumber %\label{eq:newMarkovBound}
  	\end{align}  %\eqref{eq:markovTerm} and \eqref{eq:newMarkovBound} both have the same upper bound, however the lower bound in \eqref{eq:newMarkovBound} is a upper bound to the lower bound in \eqref{eq:markovTerm}, and
Hence we conclude that the new inequality is indeed tighter than the operator Markov inequality.
\end{remark}

\begin{lem}{(Pruning Trace Inequality)} \label{lem:newMarkov}Consider the above random  operator $ X \geq 0 $ acting on a Hilbert space $ \mathcal{H}_A. $ Further, suppose $\EE[X] \leq (1-\eta){I_A} $ for $ \eta \in (0,1) $. Let $ P $ be a pruning operator for $ X $ with respect to $ I_A $, as in Definition \ref{def:pruningOperator}. Then, 
% for $ \eta \in (0,\frac{1}{2})$
we have
\begin{align}
\EE[\Tr{I_A - P}]  \leq \frac{1}{\eta}\EE\left[\|X - \EE[X]\|_1\right].\nonumber
\end{align}
\end{lem}
\begin{proof}
The proof is provided in Appendix \ref{appx:newMarkov}
\end{proof}

% \begin{lem} (Trace Expectation Inequality) \label{lem:trace_expec}
% 	If $A$ is a random operator on the Hilbert space $\mathcal{H}$, then
% 	\begin{align}
% 	\EE\left[\|A\|_1\right] = \EE\left[\Tr{\sqrt{A^{\dagger}A}}\right] \leq \Tr{\sqrt{\EE[A^{\dagger}A]}}. \nonumber
% 	\end{align}
% \end{lem}
% \begin{proof}
% Follows from the concavity of operator square-root function.
% \end{proof}

%     Consider the above ensemble $\{\lambda_x,\sigma_x\}_{x\in \mathcal{X}}$ and the operators $\Pi$ and $\{\Pi_x\}_{\eleInSet{x}{X}}$ satisfying the above hypotheses. Additionally, suppose that there exists some subset $\mathcal{T} \subset \mathcal{X},$ such that $max_{\eleInSet{x}{T}}{\lambda_x} \leq c$ and $\sum_{x \notin \mathcal{T}}\lambda_x \leq \epsilon_{\lambda}$ for some $c \in (0,1)$ and $\epsilon_{\lambda} > 0$. Let a random covering code $\mathbbm{C}\deq \{C_m\}_{m \in [1,M]}$ be instead defined as a collection of codewords $C_m$ that are chosen pairwise independently according to the uniform distribution $\frac{1}{|\mathcal{X}|}$. Then we have
%     \begin{align}
%         \EE_{\mathbbm{C}}\left[\Big\|\sigma - \frac{|\mathcal{X}|}{M}\sum_{m=1}^M \lambda_{C_m}\sigma_{C_m}\Big\|_1\right] \leq \sqrt{\frac{c|\mathcal{X}|D}{Md}} +2\big(\epsilon_{\lambda} +\delta(\epsilon)\big),
% 	\end{align}
% 	where $\delta(\epsilon) = 4\sqrt{\epsilon}$.

 \section{Point-to-point Measurement Compression using Structured Random POVMs}\label{sec:p2p}
 % \section{Simulation of POVMs with Stochastic Processing} \label{sec:nf_faithfulsimulation}
Before presenting the proof of Theorem \ref{thm:dist POVM appx}, as a pedagogical first step, we consider the measurement compression problem in the point-to-point setup.
% and achieve the optimal rate-region using structured random codes. 
This problem was addressed in \cite{winter}, where the performance limits were derived  using unstructured random POVM ensembles. Here, we redrive the performance limit using random algebraic structured POVM ensembles. Since the algebraic structured codes can only induce a uniform distribution, we consider a collection of cosets of a random linear code for this task. The problem setup is described as follows. An agent (Alice) performs a measurement $M$ on a quantum state $\rho$, and sends a set of classical bits to a receiver (Bob). Bob has access to additional private randomness, and he 
% The receiver
is allowed to use this additional resource to perform any stochastic mapping of the received classical bits.
% it receives from the sender
The overall effect on the quantum state can be assumed to be a measurement which is a concatenation of the POVM Alice performs and the stochastic map Bob implements. This problem serves as a building block toward the proof of Theorem \ref{thm:dist POVM appx}. Formally, the problem is stated as follows.

\subsection{Problem Formulation and Main Result}
\begin{definition}
For a given finite set $\mathcal{Z}$, and a Hilbert space $\mathcal{H}$, a measurement simulation protocol with parameters
$(n,\Theta,N)$ is characterized by   \\ 
\textbf{1)} a collection of codes $\mathcal{C}^{(\mu)}\subseteq \mathcal{W}^n$, for 
$\mu \in[1,N]$, such that $|\mathcal{C}^{(\mu)}| \leq \Theta$, and $\mathcal{W}$, a finite set, is called the code alphabet, \\
\textbf{2)} a collection of Alice's sub-POVMs  $\tilde{M}^{(\mu)},\mu \in [1,N]$ each acting on $\mathcal{H}_{}^{\tensor n}$ and with outcomes in $\mathcal{C}^{(\mu)}$.\\
	\textbf{3)} a collection of Bob's classical stochastic maps $P^{(\mu)}(z^n|w^n)$ for all $w^n\in \mathcal{C}^{(\mu)}$, $ z^n\in \mathcal{Z}^n$ and $\mu \in [1,N]$. \\
	The overall  sub-POVM of this protocol, given by $\tilde{M}$, is characterized by the following operators: 
	\begin{equation}\label{eq:p2p overall POVM opt}
	\tilde{\Lambda}_{z^n}\deq \frac{1}{N}\sum_{\mu=1}^N \sum_{w^n \in \mathcal{C}^{(\mu)}}
	P^{(\mu)}(z^n|w^n)~ \Lambda^{(\mu)}_{w^n}, \quad \forall z^n \in \mathcal{Z}^n,
	\end{equation}
	where $\{\Lambda^{(\mu)}_{w^n}: w^n \in \mathcal{C}^{(\mu)}\}$ is the set of operators corresponding to the sub-POVM $\tilde{M}^{(\mu)}$. Let $\mathcal{C}^{(\mu)}(i)$ denote the $i$th codeword of $\mathcal{C}^{(\mu)}$. 
\end{definition}
In the above definition,  $\Theta$ characterizes the amount of classical bits communicated from Alice to Bob, and the amount of common randomness is determined by  $N$, with $\mu$ being the common randomness bits distributed among the parties. The classical stochastic mappings induced by  $P^{(\mu)}$ represents the action of Bob on the received classical bits.  
In building the code, we use the Unionized Coset Code (UCC)  \cite{pradhanalgebraic} defined below. These codes involve two layers of codes (i) a coarse code and (ii) a fine code. The coarse code is a coset of the linear code and the fine code is the union of several cosets of the linear code.

For a fixed $k \times n$ matrix $G \in \mathbb{F}_p^{k \times n}$ with 
$k \leq n$, and $p$ being a prime number, and a $1 \times n$
vector $B \in \mathbb{F}_p^n$, define the coset code as 
\begin{align}\label{eq:UCC_Code}
    \mathbb{C}(G,B)\deq  \{ x^n: x^n = a^{k} G+B, \mbox{ for some } a^{k} \in \mathbb{F}_p^{k}
\}.
\end{align}
In other words, $\mathbb{C}(G,B)$ is a shift of the row space of
the matrix $G$. The row space of $G$ is a linear code. 
If the rank of $G$ is $k$, then there are $p^{k}$ codewords in the coset
code.

\begin{definition}\label{def:UCC}
An $(n,k,l,p)$ UCC is characterized by a pair $(G,h)$ consisting of a $k\times n$ matrix $G \in \mathbb{F}_p^{k \times n}$, and a mapping $h:\mathbb{F}_p^{l}
\rightarrow \mathbb{F}_p^n$, and the code is the following union:
$ \bigcup_{m\in \FF_p^l}\mathbb{C}(G,h(m))$, where $\mathbb{C}(\cdot,\cdot)$ is defined in \eqref{eq:UCC_Code}.
% In the context of UCC, define the composite code as $\mathbb{C}\deq\bigcup_{i \in \mathbb{F}_p^{l}} \mathbb{C}(G,h(i))$.
\end{definition}

\begin{definition}
    Given a finite set $\mathcal{Z}$, and a Hilbert space $\mathcal{H}$, an 
    $(n,\Theta,\kappa,N,p)$ UCC-based measurement simulation protocol is a pair of 
    $(n,\Theta,N)$ measurement simulation protocol and a collection of $N$ UCCs with parameters $(n,k,l,p)$  characterized by $\{(G,h^{(\mu)})\}_{\mu \in [1,N]}$
    such that (i) the code alphabet of the protocol $\mathcal{W} \subseteq \mathbb{F}_p$ (with suitable relabeling), (ii) $\kappa=p^{k}$, $\Theta=p^l$, and 
    (iii) for all $m \in \mathbb{F}_p^l$, we have $\mathcal{C}^{(\mu)}(m) \in \{a^{k}G +h^{(\mu)}(m):a^{k} \in \mathbb{F}_p^{k}\}$. 
    \end{definition}

\begin{definition}
    The UCC grand ensemble  is the ensemble of 
    $N$ UCCs where $G$, and $\{h^{(\mu)}\}_{\mu \in [1,N]}$ are chosen randomly, independently and uniformly, where the latter is chosen from the set of all mappings with replacement.  
\end{definition}

\begin{definition}
	Given a POVM $M$ acting on  $\mathcal{H}_{}$, and a density operator $\rho \in \mathcal{D}(\mathcal{H}_{})$, a tuple $(R,R_1,C,p)$ is said to be achievable using the grand UCC ensemble, if for all $\epsilon>0$ and for all sufficiently large $n$, there exists an ensemble of UCC-based measurement simulation protocols with parameters  $(n, \Theta, \kappa,N,p)$ (based on the UCC grand ensemble) such that 
	their overall sub-POVM $\tilde{M}$  is $\epsilon$-faithful to $M^{\tensor n}$ with respect to $\rho^{\tensor n}$ in the expected sense:
	\begin{align*}
	    \mathbb{E}\! &\left[ \!\sum_{z^n}\! \left\|\!
	\sqrt{\rho^{\tensor n}} (\Lambda_{z^n}\! -\!\tilde{\Lambda}_{z^n} ) \sqrt{\rho^{\tensor n}} \right\|\!+\!\tr \{I-\sum_{z^n} \tilde{\Lambda}_{z^n} \} \right] \leq \epsilon,
	\end{align*}
	where the expectation is with respect to the ensemble,
and
	\begin{align*}
	\frac{1}{n}\log_2 \Theta \leq R+\epsilon, \;
	\left|\frac{1}{n} \log \kappa-R_1 \right| \leq \epsilon,;\;
	\frac{1}{n}\log_2 N \leq C+\epsilon.
	\end{align*}
Define $\mathscr{R}_{\mbox{UCC}} $ as $
    \mathscr{R}_{\mbox{UCC}} \deq \{(R,R_1,C,p): (R,R_1,C,p) $  is achievable using the UCC grand ensemble\}.

\end{definition}

\begin{remark}
The appearance of the modulus in the second constraint needs justification. Note that $R$ is the rate of transmission of information from Alice to Bob and $C$ is the rate of the common information shared between them. So if $(R,R_1,C,p)$ is achievable, then it is clear that any $(\tilde{R},\tilde{C})$ is also achievable if $\tilde{R} \geq R$ and $\tilde{C} \geq C$. However 
$R_1$ is a parameter of the UCC grand ensemble, and there is no natural order on $R_1$, i.e., it does not naturally follows that 
$(R,\tilde{R}_1,C,p)$ is achievable for all $\tilde{R}_1 \geq R_1$. 
\end{remark}

The following theorem characterizes the achievable rate region which characterizes the asymptotic performance of the UCC grand ensemble.

\begin{theorem}\label{thm:p2pTheorem}
For any density operator $\rho\in \mathcal{D}(\mathcal{H}_{})$ and any POVM ${M} \deq \{{\Lambda}_z\}_{z \in \mathcal{Z}}$ acting on the Hilbert space $\mathcal{H}_{}$, a tuple $(R,R_1,C,p)$ is achievable using the UCC grand ensemble,
i.e., $(R,R_1,C,p) \in \mathscr{R}_{\mbox{UCC}}$
if there exist a POVM $\bar{M}_{} \deq \{\bar{\Lambda}_w^{}\}_{w \in \mathcal{W}} $, with $|\mathcal{W}|\leq p$, and a stochastic map $P_{Z|W}:\mathcal{W} \rightarrow\mathcal{Z}$ such that
\[
\Lambda_z = \sum_{w\in\mathcal{W}}P_{Z|W}(z|w)
\bar{\Lambda}_w, \quad \forall z \in \mathcal{Z},
\]
and 
\begin{align}\label{eq:ratep2pTheorem}
 	R_1+R &\geq I(W;R)_{\sigma} - S(W)_{\sigma} + \log{p}  ,\\
	R_1+R + C  &\geq I(W; RZ)_{\sigma} - S(W)_{\sigma} + \log{p} ,\\
 	0 \leq R_1  &\leq \log{p} - S(W)_{\sigma},\\
   C &\geq 0, 
\end{align}
where  $\sigma^{RWZ} \deq \sum_{w,z}\sqrt{\rho}\bar{\Lambda}^{}_{w}\sqrt{\rho} \tensor P_{Z|W}(z|w)\ketbra{w}\tensor\ketbra{z}, $ for some orthogonal sets $ \{\ket{w}\}_{\eleInSet{w}{W}}$ and $ \{\ket{z}\}_{\eleInSet{z}{Z}}.$ 
\end{theorem}

\begin{remark}
By choosing $R_1 = \log{p} - S(W)_{\sigma}$, we recover the rate region of Wilde et. al \cite[Theorem 9]{wilde_e}.
\end{remark}

\subsection{Proof of Theorem \ref{thm:p2pTheorem} Using UCC Code Ensemble}\label{sec:p2pProof}

% We provide a proof of achievability using  the Unionized Coset Codes. 
% Similar proofs have been discussed in \cite{wilde_e,atif2019faithful} using unstructured random codes, however, we proof here (the achievability result)  using algebraic structured codes.
% Note that the theorem here is similar to the Theorem V in \cite{atif2019faithful}, the proof technique here  uses the ideas developed in \cite{atif2019faithful}. However the coding strategy here is completely different. Theorem V in \cite{atif2019faithful} was based on constructing a faithful simulation protocol using random unstructured codes, but here we employ one of the structured coding techniques, namely Unionized Coset Codes \cite{pradhanalgebraic}. 
As stated earlier, the main objective of proving this theorem is to build a framework for the main theorem of the paper (Theorem \ref{thm:dist POVM appx}). In doing so, we observe that the structured POVMs constructed below are only pairwise independent. Since the results in \cite{atif2019faithful} are based on the assumption that approximating POVMs are all mutually independent,  the proof below becomes significantly different from \cite{atif2019faithful}. 

Suppose there exist a POVM $\bar{M}\deq \{\bar{\Lambda}_w\}_{w\in \mathcal{W}}$ and a stochastic map $ P_{Z|W}:\mathcal{W} \rightarrow \mathcal{Z} $, such that $ M \deq \{\Lambda_z\}_{\eleInSet{z}{Z}}$ can be decomposed as 
\begin{align}\label{eq:POVMDecomp}
\Lambda_z \deq \sum_{\eleInSet{w}{W}}P_{Z|W}(z|w)\bar{\Lambda}^{}_w , \quad \forall\eleInSet{z}{Z}.
\end{align}
We generate the canonical ensemble corresponding to $\bar{M}_{}$ as
\begin{align}\label{eq:canonicalEnsemble}
\lambda^{}_w &\deq \tr\{\bar{\Lambda}^{}_w \rho\}, \quad 
\hat{\rho}^{}_w \deq \frac{1}{\lambda^{}_w}\sqrt{\rho}  \bar{\Lambda}^{}_w \sqrt{\rho}. \end{align}
Let $\TDeltaN(W)$ denote a $\delta$-typical set associated with the probability distribution induced by $ \{\lambda_w^{}\}_{\eleInSet{w}{W}}, $ corresponding to a random variable $W$.
Let $\Pi_{\rho}$ denote the $\delta$-typical projector (as in \cite[Def. 15.1.3]{Wilde_book}) corresponding to the density operator $\rho \deq \sum_{\eleInSet{w}{W}}\lambda_w^{} \hat{\rho}^{}_w$, and $ \Pi_{w^n}$ denote the strong conditional typical projector (as in \cite[Def. 15.2.4]{Wilde_book}) corresponding to the canonical ensemble $\{\lambda_w, \hat{\rho}_w\}_{w\in \mathcal{W}}$.
% Let $\PiRho$ and $\PiwA$ denote the $\delta$-typical and the conditional typical projectors (as in \cite{holevo}) 
% for marginal density operator  and $\hat{\rho}^{}_{w^n}$ for $ w^n \in \TDeltaN(W), $ respectively.
 For each $w^n\in \TDeltan(W)$, define 
\begin{equation}
\rhotildwA \deq \PiRho \PiwA \rhohatwA \PiwA \PiRho,\nonumber
\end{equation}
and $\rhotildwA = 0,$ for $w^n \notin \TDeltan(W)$, with $\rhohatwA \deq \bigotimes_{i} \hat{\rho}^{}_{w_i}$.
% \footnote{Note that $\Lambda^{A'}_{w^n}$ and $\Lambda^{B'}_{v^n}$ are not tensor products operators.}. 

\subsubsection{Construction of Structured POVMs}
We now construct random structured POVM elements. Fix a block length $n>0$, a positive integer $ N, $ and a finite field $ \FF_p $ with $p \geq |\mathcal{W}|$. Without loss of generality, we assume $\mathcal{W} \deq \{0,1,\cdots ,|\mathcal{W}|-1\}$. Furthermore, we assume $\lambda_w =0 $ for all $ |\mathcal{W}|-1 < w < p$. From now on, we assume  that $W$ takes values in $\FF_p$ with this distribution. 
% and an arbitrary injective mapping from $\mathcal{W}$ into $\FF_p$. 
% With this mapping, we identify every $w \in \mathcal{W}$ is associated unique value in $\FF_p$. 
Let $ \mu \in [1,N]$ denote the common randomness shared between the encoder and decoder. 
% Further, let $W$ be random variable defined on the alphabet $ \mathcal{W} =\FF_p $. 
In building the code, we use the UCCs  \cite{pradhanalgebraic} as defined in Definition \ref{def:UCC} . 
% These codes involve two layers of codes (i) a coarse code and (ii) a fine code. The coarse code is a coset of the linear code and the fine code is the union of several cosets of the linear code.

% For a fixed $k \times n$ matrix $G \in \mathbb{F}_p^{k \times n}$ with 
% $k \leq n$, and a $1 \times n$
% vector $B \in \mathbb{F}_p^n$, define the coset code as 
% \[
% \mathbb{C}(G,B)\deq  \{ x^n: x^n = a^{k} G+B, \mbox{ for some } a^{k} \in \mathbb{F}_p^{k}
% \}.
% \]
% In other words, $\mathbb{C}(G,B)$ is a shift of the row space of
% the matrix $G$. The row space of $G$ is a linear code. 
% If the rank of $G$ is $k$, then there are $p^{k}$ codewords in the coset
% code. 

% \begin{definition}\label{def:UCC}
% An $(n,k,l,p)$ UCC is a pair $(G,h)$ consisting of a $k\times n$ matrix $G \in \mathbb{F}_p^{k \times n}$, and a mapping $h:\mathbb{F}_p^{l}
% \rightarrow \mathbb{F}_p^n$. In the context of UCC, define the composite code as
% $\mathbb{C}=\bigcup_{i \in \mathbb{F}_p^{l}} \mathbb{C}(G,h(i))$.
% \end{definition}

% \noindent \textbf{UCC Ensemble:} We construct a UCC ensemble by choosing $G$ uniformly from $\mathbb{F}_p^{k \times n}$, and $i(\cdot)$ uniformly from the set of all functions from $\mathbb{F}_p^{k}$ to $\mathbb{F}_p^{n}$. The two objects are chosen  independently. 

 For every $ \mu \in [1,N] $, consider a UCC $ (G^{},h^{(\mu)}) $ with parameters $ (n,k,l,p) $.
For each $ \mu$, the generator matrix $ G^{} $ along with the function $ h^{(\mu)} $ generates $ p^{k + l} $ codewords. Each of these codewords are characterized by a triple $ (a,i,\mu)$, where $ a \in \FF^{k}_p $ and $ i \in \FF^{l}_p$  correspond to the coarse code and the coset indices, respectively. Let $ \WCodeword $ denote the codewords associated with the encoder (Alice),  generated using the above procedure, where 
\begin{align}
\WCodeword = aG^{} + h^{(\mu)}(i). \label{def:codewords}
\end{align}
Now, construct the operators
\begin{align}\label{eq:A_w}
\bar{A}^{(\mu)}_{w^n} &\deq  \alpha_{w^n} \bigg(\sqrt{\rho^{\tensor n}}^{-1}\rhotildwA\sqrt{\rho^{\tensor n}}^{-1}\bigg)\quad   \nonumber \\ \quad \alpha_{w^n}& \deq  \frac{1}{(1+\eta)}\frac{p^n\lambda_{w^n}}{p^{k+l}},
\end{align}
 with $\eta \in (0,1)$ being a parameter to be determined. Note that, following the definition of $\rhotildwA$, we have $ \bar{A}^{(\mu)}_{w^n} = 0$ for $w^n \notin \TDeltan(W).$ Having constructed the operators $ \bar{A}^{(\mu)}_{w^n} $, we normalize these operators, so that they constitute a valid sub-POVM. To do so, we  define
\begin{align}
\Sigma^{(\mu)} \deq \sum_{w^n}\gammaWCoeff\bar{A}^{(\mu)}_{w^n}, \;  \gammaWCoeff & \deq |\{(a,i): \WCodeword = w^n\}|.\nonumber
\end{align} 
Now, we define $ \CutOff $ as the pruning operator for $\Sigma^{(\mu)} $ with respect to $\PiRho$ using Definition \ref{def:pruningOperator}.
Note that, the pruning operator $\CutOff$ depends on the pair $(G^{},h^{(\mu)})$. 
% the cut-off projectors onto the subspace spanned by the eigen-states of $ \Sigma^{(\mu)} $ and $ \Sigma_B^{(\mu_2)}, $ with eigen-values less than 1, respectively. 
For ease of analysis, the subspace of $ \CutOff $ is restricted to $ \PiRho $ and hence $\CutOff$ is a projector onto a subspace of $\Pi_\rho$. Using these pruning operators, for each $\mu \in [1,N]$, construct the sub-POVM $\tilde{M}^{( n, \mu)}$ as  \begin{align}
\tilde{M}^{( n, \mu)} & \deq \{\gammaWCoeff A^{(\mu)}_{w^n}\}_{w^n \in \mathcal{W}^n},\quad \label{eq:p2pPOVM-14}
\end{align}   
where $ A^{(\mu)}_{w^n} \deq \CutOff\bar{A}^{(\mu)}_{w^n}\CutOff$. 
Further, using $ \CutOff $  we have $ \sum_{w^n}\gammaWCoeff A^{(\mu)}_{w^n} = \CutOff\Sigma^{(\mu)}\CutOff \leq \PiRho\leq I, $ and thus $ \tilde{M}^{( n, \mu)}$ is a valid sub-POVM for all $ \mu \in[1,N] $. Moreover, the collection $ \tilde{M}^{( n, \mu)}$ is completed using the operators $I - \sum_{w^n \in \mathcal{W}^n}\gammaWCoeff A^{(\mu)}_{w^n}$.
% , and these operators are associated with sequences $u^n_0$ and $v^n_0$.
% , \checkComm{ which are additionally introduced in $ \mathcal{W}^n $ and $ \mathcal{V}^n $, respectively.}

\subsubsection{Binning of POVMs}\label{sec:p2pPOVM binning}
The next step is to bin the above constructed sub-POVMs. Since, UCC is a union of several cosets, we associate a bin to each coset, and hence place all the codewords of a coset in the same bin. For each $i\in \FF_p^{l}$, let $\mathcal{B}^{(\mu)}(i) \deq \mathbb{C}(G^{},h^{(\mu)}(i))$ denote the $i${th} bin.
% More precisely, $\mathbb{C}(G^{},h^{(\mu)})$ \deq \{a_1G^{}+h^{(\mu)}(i):a_1 \in \FF_p^{k_1-l_1}\}$ and $\mathcal{B}^{(\mu_2)}_2(j) \deq \{a_2G^{}+h_2^{(\mu_2)}(j):a_2 \in \FF_p^{k_2-l_2}\}$. 
Further, for all  $i\in \FF_p^{l}$, we define 
\begin{align*}
\Gamma^{A, ( \mu)}_i &\deq  \sum_{w^n \in \mathcal{W}^n}\sum_{a \in \FF_p^{k}} A^{(\mu)}_{w^n}\mathbbm{1}_{\{aG^{}+h^{(\mu)}(i) = w^n\}}.
\end{align*} 
Using these operators, we form the following collection:
% The above operators generate the following POVMs 
\begin{align}
M^{( n, \mu)} \deq  \{\Gamma^{A,  (\mu)}_i\}_{ i \in \FF_p^{l} }. \nonumber
\end{align}
Note that if the collection $\tilde{M}^{( n, \mu)}$ is a sub-POVM for each $\mu \in [1,N]$, then so is the collection $M^{( n, \mu)}$,  which is due to the  relation $ \sum_{i\in \FF_p^{l}} \Gamma^{A, ( \mu)}_i=\sum_{w^n \in \mathcal{W}^n}\gammaWCoeff A^{(\mu)}_{w^n} \leq I. $
To complete $ M^{( n, \mu)}$, we define $\Gamma^{A, ( \mu)}_0$ as $\Gamma^{A, ( \mu)}_0=I-\sum_i \Gamma^{A, ( \mu)}_i$ \footnote{{Note that $\Gamma^{A, ( \mu)}_0=I-\sum_i \Gamma^{A, ( \mu)}_i = I - \sum_{w^n \in \TDeltaN(W)}\gammaWCoeff A^{(\mu)}_{w^n}$.}}. Now, we intend to use the completions $[M^{( n, \mu)}]$ as the POVM for the encoder.
%We use the completion $[M^{( n, \mu)}]$ and $[M_B^{( n, \mu_2)}]$ as the POVMs for each encoder. Note that the index $i=0$ and $j=0$ are used, respectively, for $\Gamma^{A, ( \mu)}_0=I-\sum_i \Gamma^{A, ( \mu)}_i$ and $\Gamma^{B, ( \mu_2)}_0=I-\sum_j \Gamma^{B, ( \mu_2)}_j$.
%Also, note that the effect of the binning is in reducing the communication rates from $(S, S_2)$ to $(R_1,R_2)$, where $ R = \frac{l}{n}\log{p}$. Now, we move on to describe the decoder.

\subsubsection{Decoder mapping }
We create a decoder which, on receiving the classical bits from the encoder, generates a sequence $W^n \in \FF^n_p$ as follows. %according to the mapping $F^{(\mu)}$ defined as
The decoder first creates a set $D^{(\mu)}_{i}$ and a function $F^{(\mu)}$  defined as
% \begin{align}
% D^{(\mu)}_{i}  \deq \big\{ w^n \in \FF_p^n:& w^n \in \TDeltaN(W), w^n =
%  \tilde{a}G^{}+ h^{(\mu)}(i) \text{ for some } \tilde{a} \in \FF_p^{k} \big\}  \quad \text{ and  }\nonumber\\
% F^{(\mu)}(i)& \deq
%  \begin{cases}
% 	w^n &\quad \text{ if  } D^{(\mu)}_{i} \equiv \{w^n\} \\
% 	w^n_0 &\quad \text{ otherwise  } \label{def:Fmu}
% \end{cases}
% \end{align}
\begin{align}
D^{(\mu)}_{i}  &\deq \big\{ \tilde{a} \in \FF_p^{k}: \tilde{a}G^{}+ h^{(\mu)}(i)  \in \TDeltaN(W)
  \big\}  \quad \text{ and  }\nonumber\\
F^{(\mu)}(i)& \deq
 \begin{cases}
	 \tilde{a}G^{}+ h^{(\mu)}(i)  &\quad \text{ if  } D^{(\mu)}_{i} \equiv \{\tilde{a}\} \\
	w^n_0 &\quad \text{ otherwise  }, \label{def:Fmu}
\end{cases}
\end{align}
{where $w_0^n$ is an arbitrary sequence in  $ \FF_p^n \backslash \TDeltaN(W)  $}.
Further, $F^{(\mu)}(i)=w_0^n$ for $i=0$. Given this and the stochastic processing $ P_{Z|W} $, we obtain the approximating sub-POVM $ \hat{M}^{(n)}_{} $
 with the following operators. 
\begin{align*}
\hat{\Lambda}_{z^n}^{} \!\deq\! \frac{1}{N}\hspace{-3pt}\sum_{\mu=1}^{N}\sum_{w^n \in \FF_p^n}\sum_{i:F^{(\mu)}(i)=w^n}
\hspace{-16pt}\Gamma^{A, ( \mu)}_i P^n_{Z|W}(z^n|w^n), ~ \forall z^n\in \mathcal{Z}^n. 
% \mathcal{W}^n \times \mathcal{V}^n.   
\end{align*}
%Now, we use the stochastic mapping $ \P_{Z|W} $ to define the  sub-POVM $\hat{M}^{(n)}_{AB} \deq \{\hat{\Lambda}_{z^n}\}$ as
%\begin{align*}
%\hat{\Lambda}^{AB}_{z^n}=\sum_{w^n}  \tilde{\Lambda}_{w^n}^{AB}P^n_{Z|W}(z^n|w^n), ~ \forall z^n\in \mathcal{Z}^n.
%%, \quad \forall (u^n,v^n) \in \mathcal{W}^n\times \mathcal{V}^n.
%\end{align*}
% {Note that for $\tilde{\Lambda}_{w^n}^{} = 0$ for $w^n \notin \TDeltan(W)\medcup\{w^n_0\}.$}
The generator matrix $G^{}$ and the function $h^{(\mu)}$ are chosen randomly uniformly and independently.

\subsubsection{Trace Distance}
In what follows, we show that $\hat{M}_{}^{(n)}$ is $\epsilon$-faithful to ${M}_{}^{\tensor n}$ with respect to $\rho^{\tensor n}$ (according to Definition \ref{def:faith-sim}), where $\epsilon>0$ can be made arbitrarily small. More precisely, using \eqref{eq:POVMDecomp}, we show that, $ \EE[K] \leq \epsilon, $ where 
%\error{with probability sufficiently close to 1} , the following inequality holds
%\begin{equation}\label{eq:faithful M hat}
%\sum_{z^n} \|\sqrt{\rho_{AB}^{\tensor n}} \left(\Lambda_{z^n}-\hat{\Lambda}_{z^n}\right) \sqrt{\rho_{AB}^{\tensor n}}\|_1\leq \epsilon.
%\end{equation}
%According to the decomposition of $\Lambda_{z}$, given in \eqref{eq:LambdaAB decompos}, the above inequality is equivalent to 
\begin{align}\label{eq:actual trace dist}
{K} & \deq \sum_{z^n} \left\| \sum_{w^n}\sqrt{\rho_{}^{\tensor n}} \bar{\Lambda}^{}_{w^n}\sqrt{\rho_{}^{\tensor n}} P^n_{Z|W}(z^n|w^n) \right. \nonumber\\ & \hspace{1.6in}
\left. - \sqrt{\rho_{}^{\tensor n}}\hat{\Lambda}_{z^n}^{}\sqrt{\rho_{}^{\tensor n}}\right\|_1,
%\sum_{z^n} \norm{ \sum_{w^n}\sqrt{\rho_{AB}^{\tensor n}} \bar{\Lambda}^{}_{u^n}\tensor \bar{\Lambda}^B_{v^n}\sqrt{\rho_{AB}^{\tensor n}} P^n_{Z|W}(z^n|u^n+ v^n) - \sum_{w^n}\sqrt{\rho_{AB}^{\tensor n}}\tilde{\Lambda}_{w^n}^{AB}\sqrt{\rho_{AB}^{\tensor n}}P^n_{Z|W}(z^n|w^n)}_1
%\leq \epsilon.
\end{align}
where the expectation is with respect to the codebook generation.

\noindent \textbf{Step 1: Isolating the effect of error induced by not covering}\\ Consider the second term within $ {K} $, which can be written as
\begin{align}
\sqrt{\rho_{}^{\tensor n}}\hat{\Lambda}_{z^n}^{}\sqrt{\rho_{}^{\tensor n}} &= \frac{1}{N}\sum_{\mu}\sum_{i} \sqrt{\rho_{}^{\tensor n}}\Gamma^{A, ( \mu)}_i\sqrt{\rho_{}^{\tensor n}} \nonumber\\
& \hspace{0.2in} \times P^n_{Z|W}(z^n|F^{(\mu)}(i))\underbrace{\sum_{w^n}\!\mathbbm{1}_{\{F^{(\mu)}(i) = w^n\}}}_{=1} \nonumber \\\vspace{-20pt}
& = T + \widetilde{T}, \nonumber
\end{align}
where \begin{align}
 T \deq & \frac{1}{N}\sum_{\mu}\sum_{i>0} \sqrt{\rho_{}^{\tensor n}}\Gamma^{A, ( \mu)}_i \sqrt{\rho_{}^{\tensor n}} P^n_{Z|W}(z^n|F^{(\mu)}(i)), \nonumber \\
 \widetilde{T} \deq &\frac{1}{N}\sum_{\mu} \sqrt{\rho_{}^{\tensor n}}\Gamma^{A, ( \mu)}_0\sqrt{\rho_{}^{\tensor n}} P^n_{Z|W}(z^n|w^{n}_{0}). \nonumber 
\end{align}
Hence, we have $K \leq S+ \widetilde{S},$
% \begin{align}
%     K \leq S+ \widetilde{S}, \label{eq:p2pG_separation_1}
% \end{align}
where
\begin{align}
    S \deq \sum_{z^n} \norm{ \sum_{w^n}\sqrt{\rho_{}^{\tensor n}} \bar{\Lambda}^{}_{w^n} \sqrt{\rho_{}^{\tensor n}}P^n_{Z|W}(z^n|w^n) - T }_1, \label{eq:p2pdef_S}
\end{align}
and $\widetilde{S} \deq \sum_{z^n}\|\widetilde{T}\|_1$. Note that $\widetilde{S}$ captures the error induced by not covering the state $\rho_{}^{\tensor n}.$
% \noindent{\bf Analysis of  $ S_{2}, S_3 $ and $S_4$:}
We further bound $\widetilde{S}$ as 
\begin{align}
\widetilde{S} & \leq  \frac{1}{N}\sum_{\mu} \sum_{z^n}P^n_{Z|W}(z^n|w^{n}_{0})\left\|  \sqrt{\rho^{\tensor n}}\Gamma^{A, ( \mu)}_0\sqrt{\rho^{\tensor n}}  \right\|_1 \nonumber \\
& \leq  \frac{1}{N}\sum_{\mu}\left\| \sqrt{\rho^{\tensor n}} (I - \sum_{w^n} \gammaWCoeff A_{w^n}^{(\mu)}) \sqrt{\rho^{\tensor n}} \right\|_1 \nonumber \\
& \leq \frac{1}{N}\sum_{\mu}\left\| \sum_{w^n}\lambdawA\rhohatwA - \sum_{w^n}\sqrt{\rho^{\tensor n}} \gammaWCoeff \bar{A}^{(\mu)}_{w^n}\sqrt{\rho^{\tensor n}} \right\|_1 \nonumber\\
& \hspace{25pt} + \frac{1}{N}\sum_{\mu}\left\| \sum_{w^n}\sqrt{\rho^{\tensor n}} \gammaWCoeff \left(\bar{A}^{(\mu)}_{w^n} - A^{(\mu)}_{w^n}\right)\sqrt{\rho^{\tensor n}} \right\|_1 \nonumber \\
& \leq \widetilde{S}_1 +  \widetilde{S}_2,  \nonumber
\end{align}
where
\begin{align}
    \widetilde{S}_1 &\deq \frac{1}{N}\sum_{\mu}\left\| \sum_{w^n}\lambdawA\rhohatwA - \sum_{w^n}\sqrt{\rho^{\tensor n}} \gammaWCoeff \bar{A}^{(\mu)}_{w^n}\sqrt{\rho^{\tensor n}} \right\|_1, \nonumber\\
    \widetilde{S}_2 &\deq  \frac{1}{N}\sum_{\mu}\sum_{w^n}\left\| \sqrt{\rho^{\tensor n}} \gammaWCoeff \left(\bar{A}^{(\mu)}_{w^n} - A^{(\mu)}_{w^n}\right)\sqrt{\rho^{\tensor n}} \right\|_1. \nonumber
\end{align}
To provide a bound for the term $\widetilde{S}_1$, we
(i) develop a n-letter version of Lemma \ref{lem:Change Measure Soft Covering Variance Based} and (ii) provide a proposition bounding the term corresponding to $\widetilde{S}_1$, using this n-letter lemma.
\begin{lem}\label{lem:nf_SoftCovering} Let $\{\lambda_{w},\theta_{w}\}_{w\in \mathcal{W}}$ be an ensemble, with  $\theta_{w} \in \mathcal{D}(\mathcal{H})$ for all $w \in \mathcal{W}$, 	$\mathcal{W} \subseteq \FF_p$ for some finite prime $p$. 
Then, for any $\epsilon_c \in (0,1)$, and for any $\eta, \delta \in (0,1)$ sufficiently small, and  any $n$ sufficiently large, we have 
	\begin{align}
	\EE&\bigg[\bigg\|\sum_{w^n}\lambda_{w^n}^{}\theta_{w^n} - \frac{p^n}{p^{k+l}}\frac{1}{N'}\sum_{\mu=1}^{N'}\sum_{w^n}\sum_{a,m}
	 \frac{\lambda_{w^n}}{(1+\eta)}^{}  \nonumber \\  & \hspace{1.15in}\times  \theta_{w^n}\mathbbm{1}_{\{W^{n,(\mu)}(a,m) = w^n\}}\bigg\|_1\bigg]\leq \epsilon_c, \label{eq:nf_SoftCovering_Term}
	\end{align}  
	if  $\left(\frac{k+l}{n}\right)\log{p}+\frac{1}{n}\log{N'} > I(W;R)_{\sigma_{\theta}} - S(W)_{\sigma_{\theta}} + \log{p}$, where $ \theta_{w^n} \deq \bigtensor_{i=1}^n\theta_{w_i}  $ and $ \lambda_{w^n} \deq \Pi_{i=1}^n\lambda_{w_i} $, $ \sigma_{\theta}^{RW} \deq \sum_{w \in \mathcal{W}}\lambda_{w}^{}\theta_w\tensor \ketbra{w}$, for some orthogonal set $ \{\ket{w}\}_{\eleInSet{w}{W}}, $ and
	$\{W^{n,(\mu)}(a,m) : a \in \FF_p^{k}, m \in \FF_p^{l},\mu \in [2^{nC}]\}$ are as defined in \eqref{def:codewords}, with $G^{}$ and $h^{(\mu)}$ generated randomly uniformly and independently.
% 	pairwise independent random vectors generated according to a Unionized Coset Code with parameters $ (n,k,l,p) $ (as defined in Definition \ref{def:UCC})
\end{lem}
\begin{proof}
	The proof of the lemma is provided in Appendix \ref{appx:nf_SoftCovering_Proof}
\end{proof}    
Now we provide the following proposition.
% \begin{proposition}\label{prop:Lemma for p2p:Tilde_S1}
% There exist functions  $\epsilon_{\widetilde{S}_1}(\delta), $ and $\delta_{\widetilde{S}_1}(\delta)$, 
% such that for  all sufficiently small $\delta$ and sufficiently large $n$, we have $\EE[\widetilde{S}_1] \leq\epsilon_{\widetilde{S}_1}(\delta) $, if  $S >  I(W;R)_{\sigma} - S(W)_{\sigma}  +\log{p} + \delta_{\widetilde{S}_1}, $ where $\sigma$ is the auxiliary state defined in the theorem and $\epsilon_{\widetilde{S}_1}\searrow 0, \delta_{\widetilde{S}_1} \searrow 0$ as $\delta \searrow 0$. 
% \end{proposition}
\begin{proposition}\label{prop:Lemma for p2p:Tilde_S1}
For any $\epsilon \in (0,1)$, any $\eta, \delta \in (0,1)$ sufficiently small, and any $n$ sufficiently large, we have we have $\EE[\widetilde{S}_1] \leq\epsilon$, if  $ \frac{k+l}{n}\log{p} >  I(W;R)_{\sigma} - S(W)_{\sigma}  +\log{p}, $ where $\sigma$ is the auxiliary state defined in the theorem.
\end{proposition}
\begin{proof}
The proof is provided in Appendix \ref{appx:proof of p2p:Tilde_S1}.
\end{proof}
Now we provide a bound for $\widetilde{S}_{2}.$ For that, we first develop another n-letter lemma as follows.
% \begin{align}
%     \widetilde{S}_2 \leq \frac{1}{N}\sum_{\mu} \sum_{w^n}\gammaWCoeff\left\|\sqrt{\rho_{}^{\tensor n}}  \left (\bar{A}_{w^n}^{(\mu)} -	A_{w^n}^{(\mu)} \right )\sqrt{\rho_{}^{\tensor n}}\right\|_1. \label{eq:S2tilde}
% \end{align}
\begin{lem}\label{lem:LemAandABar}
	For $ \gammaWCoeff, \bar{A}_{w^n}^{(\mu)}, $and $ A_{w^n}^{(\mu)} $   as defined above,
	% 	$\Omega_{u^n,v^n}\deq \tr\Big\{ \sqrt{\rho^{\tensor n}\tensor \rho^{\tensor n}_B}^{-1}(\lambdawA\tensor\LambdavB)  \sqrt{\rho^{\tensor n}\tensor\rho^{\tensor n}_B}^{-1} \rho^{\tensor n}\Big\}$, 
	we have 
	\begin{align*}
	    \sum_{w^n}\gammaWCoeff&\left\|\sqrt{\rho_{}^{\tensor n}}  \left (\bar{A}_{w^n}^{(\mu)} -	A_{w^n}^{(\mu)} \right )\sqrt{\rho_{}^{\tensor n}}\right\|_1  \\
	    & \leq {2} \;{2^{3n\delta_\rho}}
    \left(H_0 + \frac{\sqrt{(1-\varepsilon)}}{(1+\eta)}\sqrt{H_1+H_2+H_3}\right), 
	\end{align*}  where 
	\begin{align}
	H_0 &\deq \left|\Delta^{(\mu)} \! - \EE[\Delta^{(\mu)}]\right|, 
	H_1 \deq \Tr{\!\!(\PiRho-\CutOff )\!\!\sum_{w^n}\lambdawA \rhotildwA\!}, \nonumber \\
	 H_2 &\deq \left\|\sum_{w^n}\lambdawA \rhotildwA - (1-\varepsilon)\sum_{w^n}\frac{\alpha_{w^n}\gammaWCoeff}{\EE[\Delta^{(\mu)}]} \rhotildwA\right\|_1, \nonumber \\
% 	 \quad & \text{ and  }  \quad 
	H_3  &\deq (1-\varepsilon)\left\|\sum_{w^n}\frac{\alpha_{w^n}\gammaWCoeff}{\Delta^{(\mu)}} \rhotildwA -  \sum_{w^n}\frac{\alpha_{w^n}\gammaWCoeff}{\EE[\Delta^{(\mu)}]} \rhotildwA\right\|_1,
	\end{align}
	 $ \Delta^{(\mu)} = \sum_{w^n\in\TDeltan(W)}\alpha_{w^n}\gammaWCoeff, \varepsilon \deq \sum_{w^n \notin \TDeltan(W)}\lambdawA $ and $\delta_\rho(\delta) \searrow 0$ as $\delta \searrow 0$. 
% 	$ c =  \sqrt{|\mathcal{W}^n|2^{-n(S(U)-\delta_u)}}.$
\end{lem}
\begin{proof}
	The proof is provided in Appendix \ref{appx:ProofofLemAandABar}
\end{proof}
%where the second inequality uses triangle inequality. 
Using the above lemma on $\widetilde{S}_2$ gives 
\begin{align}
\widetilde{S}_{2} &\leq \frac{2}{N} \sum_{\mu=1}^{N} {2^{3n\delta_\rho}}
\left(H_0 + \frac{\sqrt{(1-\varepsilon)}}{(1+\eta)}\sqrt{H_1+H_2 + H_3}\right). \nonumber
\end{align}
Let us first consider $ H_1 $. By observing $ \sum_{w^n}\lambdawA \rhotildwA \leq \PiRho\rho^{\tensor n}\PiRho  \leq 2^{-n(S(\rho)-\delta_\rho)}\PiRho$, we bound $ H_1 $ as
\begin{align}
H_1 \leq 2^{-n(S(\rho)-\delta_\rho)} \Tr{(\PiRho-\CutOff )}. \nonumber
\end{align}
Note that 
\begin{align}
    \EE[\Sigma^{(\mu)}] & = \EE\left[\sum_{w^n}\alpha_{w^n}\gammaWCoeff\sqrt{\rho^{\tensor n}}^{-1} \rhotildwA\sqrt{\rho^{\tensor n}}^{-1}\right] \nonumber \\
    & = \frac{1}{(1+\eta)}\sum_{w^n}\! \lambdawA\!\sqrt{\rho^{\tensor n}}^{-1}\!\rhotildwA \sqrt{\rho^{\tensor n}}^{-1} \!\!\leq \frac{\Pi_\rho}{(1+\eta)} . \nonumber
\end{align}  
Now, we use the Pruning Trace Inequality developed in Lemma \ref{lem:newMarkov} on $\Sigma^{(\mu)}$, with $\eta \in (0,1)$ to obtain
\begin{align}\label{eq:Simplification_H1}
\EE[H_1] &\leq 2^{-n(S(\rho)-\delta_\rho)}\frac{(1+\eta)}{\eta}\EE\left[\|\Sigma^{(\mu)} - \EE[\Sigma^{(\mu)}]\|_1\right] \nonumber \\ 
% & \leq 2^{2n\delta}\|\sum_{w^n}\gammaWCoeff\alpha_{w^n}\rhotildwA - \sum_{w^n}\lambdawA\rhotildwA\|_1  \nonumber \\
% &\leq 2^{2n\delta}c^2\|\sum_{w^n}\frac{1}{|\mathcal{W}|^n}\rhotildwA - \frac{1}{2^{nS}}\sum_{a,i>0}\tilde{\rho}^{}_{\WCodeword}\|_1 \nonumber \\
& \leq  2^{-n(S(\rho)-\delta_\rho)}\frac{(1+\eta)}{\eta}\left\|\Pi_\rho\sqrt{\rho^{\tensor n}}^{-1}\right\|_\infty\nonumber \\
& \hspace{7pt} \times\EE\left[{\left\|\sum_{w^n}\alpha_{w^n}\gammaWCoeff \rhotildwA - \EE \big[\sum_{w^n}\alpha_{w^n}\gammaWCoeff\rhotildwA\big] \right\|_1}\right]\nonumber \\
& \hspace{15pt} \times\left\|\Pi_\rho\sqrt{\rho^{\tensor n}}^{-1}\right\|_\infty \nonumber \\ 
& \leq {2^{2n\delta_\rho}}\frac{(1+\eta)}{\eta}\EE\left[\left\|\sum_{w^n}\frac{\lambdawA \rhotildwA}{(1+\eta)} - \frac{1}{(1+\eta)}\frac{p^n}{p^{k+l}} \right. \right. \nonumber \\& \hspace{55pt} \left.\left. \sum_{w^n}\sum_{a,i} \lambdawA\rhotildwA \11_{\{\WCodeword=w^n\}}\right\|_1\right] \nonumber \\
& = {2^{2n\delta_\rho}}\frac{(1-\varepsilon)}{\eta}\EE[\widetilde{H}],
\end{align}
where the second inequality follows from H\'olders inequality, and the equality follows by defining $\widetilde{H}$  as
\begin{align}
    \widetilde{H} \deq &\left\|\sum_{w^n}\frac{\lambdawA}{(1-\varepsilon)} \rhotildwA \right. \nonumber \\ & \hspace{10pt} \left. - \frac{p^n}{p^{k+l}} \sum_{w^n}\sum_{a,i} \frac{\lambdawA}{(1-\varepsilon)}\rhotildwA \11_{\{\WCodeword=w^n\}}\right\|_1.
\end{align}
Similarly, using $\EE[\Delta^{(\mu)}] = \frac{(1-\varepsilon)}{(1+\eta)}$, $H_2$ can be simplified as
\begin{align}\label{eq:Simplification_H2}
H_2 & =  \left\|\sum_{w^n}\lambdawA \rhotildwA \right. \nonumber \\ & \hspace{20pt} \left. - \frac{p^n}{p^{k+l}} \sum_{w^n}\sum_{a,i} \lambdawA\rhotildwA \11_{\{\WCodeword=w^n\}}\right\|_1 \nonumber \\
& = (1-\varepsilon)\tilde{H}.
% \nonumber \\ 
% & \leq   \eta\sum_{w^n}\lambdawA \|\rhotildwA\|_1 +  (1+\eta)\left\|\sum_{w^n}\lambdawA \rhotildwA - \frac{1}{(1+\eta)}\frac{p^n}{2^{nS}} \sum_{w^n}\sum_{a,i} \lambdawA\rhotildwA \11_{\{\WCodeword=w^n\}}\right\|_1 = \eta + (1+\eta)\widetilde{H}, 
% \|\sum_{w^n}\frac{1}{|\mathcal{W}^n|}\rhotildwA - \frac{1}{2^{nS}}\sum_{a,i>0}\tilde{\rho}^{}_{\WCodeword}\|_1\nonumber
\end{align}
% where the inequality above is obtained using the triangle inequality.
Now we consider $H_3$ and convert it into a similar expression as $H_0$.
\begin{align}\label{eq:Simplification_H3}
    H_3 & \leq (1-\varepsilon)\!\!\!\!\sum_{w^n\in\TDeltan(W)}\!\!\!\!\!\alpha_{w^n}\gammaWCoeff \left|\frac{1}{\Delta^{(\mu)}}  -  \frac{1}{\EE[\Delta^{(\mu)}]} \right| \nonumber \\
    & = (1+\eta) \left|{\Delta^{(\mu)}}  -  {\EE[\Delta^{(\mu)}]} \right| = (1+\eta)H_0.
\end{align}
Using the above simplification and the concavity of square-root function we obtain:
\begin{align}
\EE[\widetilde{S}_2] & \leq \frac{2}{N} {2^{3n\delta_\rho}}\sum_{\mu=1}^{N}
 \left(\EE[H_0] +  \frac{\sqrt{(1-\varepsilon)}}{(1+\eta)}\right. \nonumber \\ 
 & \hspace{10pt} \times \left.\sqrt{(1-\varepsilon)\left(\frac{2^{2n\delta_\rho}}{\eta}+1\right)\EE[\widetilde{H}] + {(1+\eta)}\EE[H_0]}\right) \nonumber \\
 & \leq \frac{2}{N} {2^{3n\delta_\rho}}\sum_{\mu=1}^{N}
 \Bigg(\EE[H_0]  +\frac{{(1-\varepsilon)}}{(1+\eta)}  \nonumber \\ & \hspace{20pt} \times\! \sqrt{\!\left(\!\frac{2^{2n\delta_\rho}}{\eta}\!+\!1\!\right)\!\EE[\widetilde{H}]}  + \sqrt{\frac{(1-\varepsilon)}{(1+\eta)}}\sqrt{\EE[H_0]}\Bigg). \nonumber 
 \end{align}
 The following proposition provides a bound on the above term.
 \begin{proposition}\label{prop:Lemma for p2p:S_13}
 For any $\epsilon \in (0,1)$, any $\eta, \delta \in (0,1)$ sufficiently small, and any $n$ sufficiently large, we have $\EE\left[\widetilde{S}_2\right] \leq\epsilon$, if  $\frac{k+l}{n}\log{p} > I(W;R)_{\sigma} - S(W)_{\sigma} +\log{p}  $,  where $\sigma$ is the auxiliary state defined in the theorem.
% There exists  $\epsilon_{\widetilde{S}_2}(\delta), \delta_{\widetilde{S}_2}(\delta)$ such that for  all sufficiently small $\delta$ and sufficiently large $n$, we have $\EE\left[\widetilde{S}_2\right] \leq \epsilon_{\widetilde{S}_2}$ if $S \geq I(W;R)_{\sigma} +\log{p} - S(W)_{\sigma}+ \delta_{\widetilde{S}_2}$, where $\sigma$ is the auxiliary state defined in the theorem and $\epsilon_{\widetilde{S}_2}, \delta_{\widetilde{S}_2} \searrow 0$ as $\delta \searrow 0$.
\end{proposition}
\begin{proof}
The proof is provided in Appendix \ref{appx:proof of p2p:TildeS_2}
\end{proof}

{\begin{remark}
The term corresponding to the operators that complete the sub-POVMs $M^{(n,\mu)}$, i.e., $ I - \sum_{w^n \in \TDeltaN(W)}\gammaWCoeff A_{w^n}^{(\mu)}$ is taken care in $\widetilde{T}$. The expression $T$ excludes these completing operators. 
% Therefore, we use $A_{w^n}^{(\mu)}$  and $B_{v^n}^{(\mu_2)} $ to denote the operators corresponding to $u^n \in \TDeltaN(U)$ and $v^n \in \TDeltaN(V)$, respectively.
\end{remark}
}
\noindent{\bf Step 2: Isolating the effect of error induced by binning}\\\
For this, we simplify $T$ as 
%let us define $e^{(\mu)}(u^n,v^n)$ as   $e^{(\mu)}(u^n,v^n)\deq F^{(\mu)}(i,j)$, for each $(u^n, v^n) \in \mathcal{B}^{(\mu)}_1(i)\times \mathcal{B}^{(\mu_2)}_2(j)$
% and $(u^n,v^n)\in \mathcal{C}^{(\mu)}$
%. For any $(u^n, v^n) \notin \mathbbm{C}^{(\mu)}$ let $e^{(\mu)}(u^n,v^n)=w_0^n$.
\begin{align}
T 
= & \frac{1}{N}\sum_{\mu}\sum_{w^n}\sum_{\substack{i>0}} \sum_{a\in \FF_p^{k}} \sqrt{\rho_{}^{\tensor n}} A_{w^n}^{(\mu)}  \sqrt{\rho_{}^{\tensor n}} \nonumber \\ & \hspace{20pt} \times P^n_{Z|W}(z^n|F^{(\mu)}(i))\mathbbm{1}_{\{aG^{} + h^{(\mu)}(i) = w^n\}}. \nonumber 
% = & \frac{1}{N_1N_2}\sum_{\mu,\mu_2}\sum_{u^{n},v^{n}} \sqrt{\rho_{}^{\tensor n}}\left (A_{w^n}^{(\mu)}\tensor B_{v^n}^{(\mu_2)}\right )\sqrt{\rho_{}^{\tensor n}}\sum_{\substack{i>0,\\j>0}}\mathbbm{1}_{\left \{u^n \in B^{(\mu)}_{1}(i), v^n \in B^{(\mu_2)}_{2}(j) \right \}} P^n_{Z|W}(z^n|e^{(\mu)}(u^n,v^n)) \nonumber \\
% = & \frac{1}{N_1N_2}\sum_{\mu,\mu_2}\sum_{u^{n},v^{n}} \sqrt{\rho_{}^{\tensor n}}\left (\gammA_{w^n}^{(\mu)}A_{w^n}^{(\mu)}\tensor\zeta_{v^n}^{(\mu_2)} B_{v^n}^{(\mu_2)}\right )\sqrt{\rho_{}^{\tensor n}} P^n_{Z|W}(z^n|e^{(\mu)}(u^n,v^n)). \nonumber
\end{align}
% {Note that the $w^n$ that appear in the above summation is confined to $\TDeltaN(W)$, however for ease of notation we do not make this explicit.}
We substitute the above expression into $S$ defined in \eqref{eq:p2pdef_S}, 
%to obtain
%\begin{align}
%{G} \leq S + S_2 + S_3 + S_4, \label{eq:G}
%\end{align}
% \begin{align}
% S =  \sum_{z^n} &\left \| \sum_{w^n}\sqrt{\rho_{}^{\tensor n}}\left( \bar{\Lambda}^{}_{u^n}\tensor \bar{\Lambda}^B_{v^n}P^n_{Z|W}(z^n|e^{(\mu)}(u^n+v^n)) -  \frac{1}{N_1N_2}\sum_{\mu,\mu_2}\gammaWCoeff A_{w^n}^{(\mu)}\tensor \zetaCoeff B_{v^n}^{(\mu_2)}\right)\sqrt{\rho_{}^{\tensor n}}P^n_{Z|W}(z^n|e^{(\mu)}(u^n,v^n))\right \|_1. \nonumber
% \end{align}
and isolate the effect of binning by adding and subtracting an appropriate term within $ S $ and applying triangle inequality to obtain $ S \leq S_{1}  + S_{2},$ where 
\begin{align}
S_{1} & \deq \sum_{z^n} \left\|\sum_{w^n}\sqrt{\rho_{}^{\tensor n}}  \left (\bar{\Lambda}^{}_{w^n}-\frac{1}{N} \sum_{\mu} \gammaWCoeff A_{w^n}^{(\mu)} \right )\sqrt{\rho_{}^{\tensor n}} \right. \nonumber \\
& \hspace{2.1in}  \left.\times P^n_{Z|W}(z^n|w^n)\right\|_1\!, \nonumber \\
S_{2} & \deq \sum_{z^n}\left \| \frac{1}{N}\! \sum_{\mu}\!\!\sum_{a,i>0} \!\!\sum_{w^n} \!\! \!\sqrt{\rho_{}^{\tensor n}}  A_{w^n}^{(\mu)}  \!\sqrt{\rho_{}^{\tensor n}}\mathbbm{1}_{\{aG^{} + h^{(\mu)}(i) = w^n\}} \right. \nonumber \\
& \hspace{0.6in} \left.  \times \left(P^n_{Z|W}(z^n|w^n)-   P^n_{Z|W}\left (z^n|F^{(\mu)}(i)\right )\right)\right \|_1\!, \nonumber
\end{align}
%We start by applying a triangle inequality as in the following, where the binning effect is isolated.  
%\begin{align}\nonumber
%    \sum_{z^n} & \norm{\sum_{w^n}\sqrt{\rho_{}^{\tensor n}} \left(  \bar{\Lambda}^{}_{u^n}\tensor \bar{\Lambda}^B_{v^n}\;P^n_{Z|W}(z^n|w^n) -\tilde{\Lambda}_{u^n, v^n}\right) \sqrt{\rho_{}^{\tensor n}}}_1 \nonumber \\
%    &\leq \sum_{z^n} \left\|\sum_{w^n}\sqrt{\rho_{}^{\tensor n}}  \left (\bar{\Lambda}^{}_{u^n}\tensor \bar{\Lambda}^B_{v^n} -\frac{1}{N} \sum_{\mu=1}^{N} \sum_{\mu_2=1}^{N_2} \gammaWCoeff A_{w^n}^{(\mu)} \tensor\zetaCoeff B_{v^n}^{(\mu_2)}\right )\sqrt{\rho_{}^{\tensor n}} P^n_{Z|W}(z^n|w^n)\right\|_1  \nonumber \\
%    &   \hspace{10pt} +\sum_{z^n}\left \| 	  \sqrt{\rho_{}^{\tensor n}}\left (\frac{1}{N} \sum_{\mu=1}^{N} \sum_{\mu_2=1}^{N_2}\gammaWCoeff A_{w^n}^{(\mu)} \tensor\zetaCoeff B_{v^n}^{(\mu_2)} P^n_{Z|W}(z^n|w^n)  -  \tilde{\Lambda}_{u^n, v^n}\right ) \sqrt{\rho_{}^{\tensor n}}\right \|_1 \label{eq:binningIsolation}
%%    &+\sum_{z^n} \norm{\sum_{w^n}\sqrt{\rho_{}^{\tensor n}} \left(  \bar{\Lambda}^{}_{u^n}\tensor \bar{\Lambda}^B_{v^n}\;P^n_{Z|W}(z^n|w^n) -\tilde{\Lambda}_{u^n, v^n}\right) \sqrt{\rho_{}^{\tensor n}}}_1
%	\end{align}
%\noindent Let us denote the first term on the RHS of the above equation \eqref{eq:binningIsolation} with $ S $ and second term with $ S_2 $.
where $F^{(\mu)}(\cdot)$ is as defined in \eqref{def:Fmu}. Note that the term $ S_{1} $ characterizes the error introduced by approximation of the original POVM with  the collection of  approximating sub-POVM $ \tilde{M}^{(n,\mu)} $, and the term $ S_2 $ characterizes the error caused by binning this approximating sub-POVM. In this step, we analyze $S_2$ and prove the following proposition. %We start with the analysis of $ S_{2} $ and follow with $ S_{1} $ after that.

\begin{proposition}\label{prop:Lemma for p2p:S_2}
For any $\epsilon \in (0,1)$, any $\eta, \delta \in (0,1)$ sufficiently small, and any $n$ sufficiently large, we have $\EE\left[{S}_2\right] \leq \epsilon$, if $ \frac{k+l}{n}\log{p} - R < \log{p} - S(W)_{\sigma}$,  where $\sigma$ is the auxiliary state defined in the statement of the theorem.
\end{proposition}
\begin{proof}
The proof is provided in Appendix \ref{appx:proof of p2p:S2}
\end{proof}
\noindent\textbf{Step 3: Isolating the effect of approximating measurement}\\
\noindent In this step, we finally analyze the error induced from employing the approximating measurement, given by the term $S_1$. We add and subtract appropriate terms
within $ S_1 $ and use triangle inequality to obtain $ S_1 \leq S_{11} + S_{12} + S_{13} $, where

\begin{align}
S_{11} & \deq  \sum_{z^n}\left\|\sum_{w^n}\sqrt{\rho_{}^{\tensor n}}  \left (\bar{\Lambda}^{}_{w^n} -\frac{1}{N} \sum_{\mu=1}^{N} \frac{\alpha_{w^n}\gammaWCoeff }{\lambdawA}\bar{\Lambda}^{}_{w^n} \right )\right. \nonumber \\ 
& \hspace{1.6in}\left. \times\sqrt{\rho_{}^{\tensor n}} P^n_{Z|W}(z^n|w^n)\right\|_1, \nonumber \\
S_{12} &\deq \sum_{z^n}\left\|\frac{1}{N}\!\sum_{\mu=1}^{N}\sum_{w^n}\!\sqrt{\rho_{}^{\tensor n}} \! \left (\! \frac{\alpha_{w^n}\gammaWCoeff }{\lambdawA}\bar{\Lambda}^{}_{w^n}\!\!-	\gammaWCoeff  \bar{A}_{w^n}^{(\mu)}\! \right )\right. \nonumber \\ 
& \hspace{1.6in}\left.\times \sqrt{\rho_{}^{\tensor n}} P^n_{Z|W}(z^n|w^n)\right\|_1, \nonumber \\
S_{13} &\deq \sum_{z^n}\left\|\frac{1}{N} \sum_{\mu=1}^{N}\sum_{w^n}\sqrt{\rho_{}^{\tensor n}}  \left ( 	\gammaWCoeff  \bar{A}_{w^n}^{(\mu)} -	\gammaWCoeff  A_{w^n}^{(\mu)} \right )\right. \nonumber \\
& \hspace{1.6in}\left.\times\sqrt{\rho_{}^{\tensor n}} P^n_{Z|W}(z^n|w^n)\right\|_1. \nonumber
\end{align}
Now with the intention of employing Lemma \ref{lem:nf_SoftCovering}, we express $S_{11}$ as
\begin{align}
S_{11} & = \left\|\sum_{w^n}\lambdawA\rhohatwA \tensor \phi_{w^n} - \frac{1}{N}\frac{1}{(1+\eta)}\frac{p^n}{p^{k+l}}\right. \nonumber \\ 
& \hspace{0.4in}\left.\times \sum_{\mu}\sum_{w^n}\sum_{a,i\neq 0}\mathbbm{1}_{\{\WCodeword = w^n\}}\rhohatwA \tensor \phi_{w^n}\right\|_1, \nonumber 
\end{align}
where the equality above is obtained by defining $ \phi_{w^n} = \sum_{z^n}P^n_{Z|W}(z^n|w^n)\tensor\ketbra{z^n} $ and using the definitions of $ \alpha_{w^n}, \gammaWCoeff $ and $ \rhohatwA $, followed by using the triangle inequality for the block diagonal operators, Note that the triangle inequality becomes an equality for such block diagonal operators. By identifying $\theta_{w}$ with $\hat{\rho}_w\tensor \phi_w$ in Lemma \ref{lem:nf_SoftCovering} we obtain the following: for all $\epsilon>0$ and $\eta,\delta \in (0,1)$ sufficiently small, and any $n$ sufficiently large, $\EE\left[{S}_{11}\right] \leq \epsilon$, if $ \frac{k+l}{n}\log{p} + \frac{1}{n}\log{N} > I(W;R,Z)_{\sigma} + \log{p} - S(W)_{\sigma}$,  where $\sigma$ is the auxiliary state defined in the theorem.

Now we consider the term corresponding to $ S_{12}$, and prove that its expectation is small. Recalling $ S_{12} $, we get
\begin{align}
%S_{12} 
%&= \sum_{z^n,v^n}\left\|\sum_{w^n}\sqrt{\rho_{}^{\tensor n}}  \left (\frac{1}{N}\sum_{\mu=1}^{N} \frac{\gammA_{w^n}^{(\mu)}}{\lambdawA}\bar{\Lambda}^{}_{u^n}\tensor \bar{\Lambda}^B_{v^n} -\frac{1}{N} \sum_{\mu=1}^{N} 	  A_{w^n}^{(\mu)} \tensor  \bar{\Lambda}^B_{v^n}\right )\sqrt{\rho_{}^{\tensor n}} P^n_{Z|W}(z^n|w^n)\right\|_1 \nonumber \\
S_{12}  & \leq \frac{1}{N}\sum_{\mu=1}^{N}\sum_{w^n}\sum_{z^n} P^n_{Z|W}(z^n|w^n) \nonumber\\
& \hspace{13pt}\times\left\|\sqrt{\rho_{}^{\tensor n}}  \left ( \frac{\alpha_{w^n}\gammaWCoeff}{\lambdawA} \bar{\Lambda}^{}_{w^n} -	 \gammaWCoeff \bar{A}_{w^n}^{(\mu)} \right )\sqrt{\rho_{}^{\tensor n}} \right\|_1\!, \nonumber \\
& = \frac{1}{N}\sum_{\mu=1}^{N}\sum_{w^n}\alpha_{w^n}\gammaWCoeff  \bigg\|\sqrt{\rho_{}^{\tensor n}}   \left ( \frac{1}{\lambdawA}\bar{\Lambda}^{}_{w^n} - \right.\nonumber\\
& \hspace{80pt}\left.\sqrt{\rho_{}^{\tensor n}}^{-1}\rhotildwA\sqrt{\rho_{}^{\tensor n}}^{-1}\right )\sqrt{\rho_{}^{\tensor n}} \bigg\|_1, \nonumber
\end{align}
where the inequality above is obtained by using triangle inequality.
% and the next equality follows from the fact that $ \sum_{z^n} P^n_{Z|W}(z^n|w^n) =1  $ for $ w^n \in \mathcal{W}^n, $ and using the definition of $ \bar{A}_{w^n}^{(\mu)} $.
Applying the expectation, we get
\begin{align*}
 \EE{\left[S_{12} \right]} & \leq \frac{1}{(1+\eta)}\sum_{w^n}\lambdawA \bigg\|\sqrt{\rho_{}^{\tensor n}} \left ( \frac{1}{\lambdawA}\bar{\Lambda}^{}_{w^n}-\right.\nonumber\\
& \hspace{1in}\left.\sqrt{\rho_{}^{\tensor n}}^{-1}\rhotildwA\sqrt{\rho_{}^{\tensor n}}^{-1}\right )\sqrt{\rho_{}^{\tensor n}} \bigg\|_1, \nonumber \\
& \leq  \frac{1}{(1+\eta)}\!\!\!\sum_{\substack{w^n \in \TDeltaN(W) }}\!\!\lambdawA\left\| \left (\rhohatwA -\rhotildwA \right ) \right\|_1 \\ 
& \hspace{1in}+ \frac{1}{(1+\eta)}\!\!\!\sum_{\substack{w^n \notin \TDeltaN(W) }}\!\!\lambdawA\left\| \rhohatwA  \right\|_1 \\
& \leq \frac{(2\sqrt{\varepsilon' 
} + 2\sqrt{\varepsilon''}) +\varepsilon}{(1+\eta)}   = \epsilon_{\scriptscriptstyle S_{12}}\nonumber,
\end{align*}
where we have used the fact that $ \EE{[\alpha_{w^n} \gammaWCoeff ]} = \frac{\lambdawA}{(1+\eta)}$, and the last inequality is obtained by the repeated usage of the Average Gentle Measurement Lemma \cite{Wilde_book} and
% and $ \epsilon_{\scriptscriptstyle S_{12}} \searrow 0$ as $\delta \searrow 0$
setting $ \epsilon_{\scriptscriptstyle{S}_{12}} = \frac{1}{(1+\eta)} (2\sqrt{\varepsilon'} + 2\sqrt{\varepsilon''} + \varepsilon) $ with  $ \epsilon_{\scriptscriptstyle {S}_{12}} \searrow 0 $ as $n \rightarrow \infty$ and {$ \varepsilon' \deq \varepsilon'_p + 2\sqrt{\varepsilon'_p} $ and $ \varepsilon'' \deq 2\varepsilon'_p + 2\sqrt{\varepsilon'_p} $ for $\varepsilon'_p \deq 1-\min\left\{\tr{\Pi_{\rho}\hat{\rho}_{w^n}}, \tr{\Pi_{w^n}\hat{\rho}_{w^n}},1-\varepsilon \right\}$} (see (35) in \cite{wilde_e} for details).
% \deq \frac{(2\sqrt{\varepsilon'} + 2\sqrt{\varepsilon''})+ \varepsilon}{(1+\eta)}  $ and  $ \epsilon_{\scriptscriptstyle S_{12}} \searrow 0 $ as $ \delta \searrow 0 $ and $ \varepsilon' = \varepsilon + 2\sqrt{\varepsilon} $ and $ \varepsilon'' = 2\varepsilon' + 2\sqrt{\varepsilon'} $ 
Now, we move on to bounding the last term within $ S_1 $, i.e., $ S_{13}. $ We start by applying triangle inequality to obtain
\begin{align}
S_{13} & \leq \sum_{z^n}\sum_{w^n}P^n_{Z|W}(z^n|w^n) \nonumber \\ 
& \hspace{10pt} \times \left\|\frac{1}{N} \sum_{\mu=1}^{N}\sqrt{\rho_{}^{\tensor n}}  \left ( 	\gammaWCoeff  \bar{A}_{w^n}^{(\mu)} - 	\gammaWCoeff  A_{w^n}^{(\mu)} \right ) \sqrt{\rho_{}^{\tensor n}} \right\|_1 \nonumber \\
& \leq  \frac{1}{N} \sum_{\mu=1}^{N} \sum_{w^n}\gammaWCoeff\left\|\sqrt{\rho_{}^{\tensor n}}  \left (\bar{A}_{w^n}^{(\mu)} -	A_{w^n}^{(\mu)} \right )\sqrt{\rho_{}^{\tensor n}}\right\|_1 \nonumber \\
& = \widetilde{S}_2.\label{eq:termS_13}
\end{align}
Since the above term is exactly same as $\widetilde{S}_2,$ we obtain the same rate constraints as in $\widetilde{S}_2$ to bound $S_{13},$ i.e., for all $\epsilon>0$ and $\eta,\delta \in (0,1)$ sufficiently small, and any $n$ sufficiently large, $\EE[S_{13}] \leq \epsilon$ if $ \frac{k+l}{n}\log{p} > I(W;R)_{\sigma} +\log{p} -S(W)_\sigma$.
% By taking expectation with respect to the Alice's codebook, we obtain
%  \begin{align}
%  \EE[S_{13}] & \leq 2^{-\frac{n}{4}(S - (I(W;R)_{\sigma} -S(W)_{\sigma} + \log{p} + \delta_{S_{13}}))}  = \epsilon_{S_{13}}
% \end{align}
% where in the above inequality we use the Covering Lemma developed in Lemma \ref{lem:Soft Covering Variance Based} on the ensemble $ \{\frac{1}{|\mathcal{W}^n|}, \Pi_{u^n}^{}\hat{\rho}_{u^n}^{}\Pi_{u^n}^{}\}_{u^n \in \mathcal{C}^{(\mu)}}, $ with $ \PiRho $ and $ \Pi_{u^n}^{} $ as the total subspace and the codeword subspace projectors. These projectors satisfy all the hypotheses of lemma for $ D = 2^{n(S(W) - \log{p} + S(\rho) +\delta + \delta_u}) $ and $ d = 2^{n(S(R|W)_{\sigma} - \delta_u')}, $ where $ \delta_u,\delta_u', \delta_{S_{13}}\searrow 0 $ as $ \delta \searrow 0. $ 
% This implies, if $ S \geq I(W;R)_{\sigma} -S(W)_{\sigma} + \log{p} + \delta_{S_{13}} $, we have $ \EE[S_{13}] \leq \epsilon_{S_{13}} $.

Since $S_1  \leq S_{11} + S_{12}+S_{13}$, $S_1$ can be made arbitrarily small for sufficiently large n, if $\frac{k+l}{n} \log p+\frac{1}{n} \log N > I(W;RZ)_{\sigma} - S(W)_{\sigma} + \log{p} $ and $ \frac{k+l}{n}\log{p} > I(W;R)_{\sigma} -S(W)_{\sigma} + \log{p}$.

 \subsubsection{Rate Constraints}
 To sum-up, we showed $\EE[K] \leq \epsilon$
%  that the trace distance inequality in \eqref{eq:nfctual trace dist} 
holds for sufficiently large $n$
% in the expected sense, 
if the following bounds hold:
 \begin{subequations}\label{eq:rate-region 1}
 	\begin{align}
 	R_1 + R &> I(W;R)_{\sigma} - S(W)_{\sigma} + \log{p}  ,\\
R_1 + R + C  &> I(W; RZ)_{\sigma} - S(W)_{\sigma} + \log{p} ,\\
 	R_1 &< \log{p} - S(W)_{\sigma},\\
 	R_1 &\geq 0,  \quad  C\geq 0,
 	\end{align}
 \end{subequations}
 where $R_1 \deq\frac{k}{n}\log{p} $ and $ C\deq \frac{1}{n}\log_2 N$,
 and $R=\frac{l}{n} \log p$. 
% This implies that $\tilde{M}_{AB}$ is $\epsilon$-faithful to $\bar{M}\tensor \bar{M}_B$ with probability sufficiently close to one, and hence, $\hat{M}_{AB}$ is also $\epsilon$-faithful to $M_{AB}$, i.e, \eqref{eq:faithful M hat} is satisfied. 
Therefore, there exists a distributed protocol with parameters $(n, 2^{nR}, 2^{nC})$ such that its overall POVM $\hat{M}_{}$ is $\epsilon$-faithful to $M_{}^{\tensor n}$ with respect to $\rho_{}^{\tensor n}$. This completes the proof of the theorem.
%  Lastly,  we complete the proof of the theorem using the following lemma.
%  \begin{lem}
%  	Let $\mathcal{R}_1$ denote the closure of the set of all $(R,C)$ for which there exists $R_1$ such that the triple $(R_1, R, C )$ satisfies the inequalities in \eqref{eq:rate-region 1}. Let, $\mathcal{R}_2$ denote the set of all tuples $(R, C)$ that satisfies  the inequalities in \eqref{eq:ratep2pTheorem} given in the statement of the theorem. Then, $
%  	\mathcal{R}_1=\mathcal{R}_2$.
%  \end{lem}
%  \begin{proof}
%  	This follows from Fourier-Motzkin elimination \cite{fourier-motzkin}.
%  	\end{proof}

\section{Proof of Theorem \ref{thm:dist POVM appx}}\label{sec:proof of thm dist POVM}
%\section{Proof of Theorem \ref{thm:nf_dist POVM appx}}
%\label{appx:nf_proof of thm dist POVM}
\addtocontents{toc}{\protect\setcounter{tocdepth}{1}}
% As with the approximation of distributed POVMs with deterministic mapping (Section \ref{sec:Appx_POVM}), we start by decomposing the operators of the original POVM $M_{AB}$ as. 
% However, instead of restricting to the deterministic maps, we here allow for stochastic maps. Let $ P_{Z|U,V} $ be a stochastic map and let the decomposition be
Suppose there exists a finite field $\mathbb{F}_p$, for a prime $p$, a pair of mappings $f_S:\mathcal{S} \rightarrow \mathbb{F}_p$ and 
$f_T:\mathcal{T} \rightarrow \mathbb{F}_p$, and a stochastic mapping $P_{Z|W}:\mathbb{F}_p \rightarrow \mathcal{Z}$ such that 
\[
P_{Z|S,T}(z|s,t)=P_{Z|W}(z|f_S(s)+f_T(t)), 
\]
$\forall s \in \mathcal{S}, t \in \mathcal{T}, z \in \mathcal{Z},$ yielding $U=f_S(S)$, and $V=f_T(T)$. This implies that 
we have POVMs $\bar{M}_A \deq $ $\{\bar{\Lambda}^A_u\}_{u\in \mathcal{U}}$ and $\bar{M}_B\deq\{\bar{\Lambda}^B_v\}_{v\in \mathcal{V}}$ with 
$\mathcal{U}=\mathcal{V}\deq \mathbb{F}_p$ and a stochastic map $ P_{Z|W}:\FF_p\rightarrow \mathcal{Z} $,  such that ${M}_{AB}$ can be decomposed as 
\begin{equation}\label{eq:nf_LambdaAB decompos}
\Lambda^{AB}_{z}=\sum_{u,v} P_{Z|W}(z|u+v) \bar{\Lambda}^A_{u}\tensor \bar{\Lambda}^B_{v}, ~ \forall z, 
\end{equation}
where $W$ is defined as $W = U + V$.
% The proof technique here uses ideas developed in the point-to-point Theorem (Section \ref{sec:p2p}). 
The coding strategy used here is based on Unionized Coset Codes, similar to the one employed in the point-to-point proof (Section \ref{sec:p2pProof}), but extended to a distributed setting. 
Further, the structure in these codes provide a method to exploit the structure present in the stochastic processing applied by Charlie on the classical bits received, i.e., $ P_{Z|U+V} $.
% Note that the proof technique here  uses the ideas developed in \cite{atif2019faithful}, however with a very different coding strategy. Theorem VII in \cite{atif2019faithful} was based on constructing a faithful simulation protocol using random unstructured codes. However, here we employ one of the structured coding techniques, namely Unionized Coset Codes \cite{pradhanalgebraic} in order to exploit the structure present in the stochastic processing applied by eve on the classical bits received, i.e., $ P_{Z|W} $. 
Using this technique, we aim to strictly reduce the rate constraints compared to the ones obtained in Theorem 6 of \cite{atif2019faithful}.
% This requires a considerably different methodology. More specifically, Lemma \ref{lem:faithful_equivalence} was employed in Theorem \ref{thm: dist POVM appx}, which guaranteed that any two point-to-point POVMs that can individually approximate their corresponding original POVMs, can also faithfully approximate a measurement formed by the tensor product of the original POVMs performed on any state in the tensor product Hilbert space. Such a lemma cannot be developed in the setting involving a stochastic decoder. This is due to the fact that bits received from Alice and Bob are jointly perturbed by the stochastic decoder which doesn't allow a straightforward segmentation into two point-to-point problems. However, the analysis performed in the Section \ref{appx:nf_proof of thm dist POVM} actually modularizes the problem, using an asymmetric partitioning.
Also note that, the results in \cite{atif2019faithful} are based on the assumption that approximating POVMs are all mutually independent. However, since the structured construction of the POVMs only guarantees pairwise independence among the operators of the POVM, the proofs below become significantly different from \cite{atif2019faithful}. 
% \error{Also Add 1) explanation about the new Markov Lemma, 2) new Soft Covering proof}.

We start by generating the canonical ensembles corresponding to $\bar{M}_{A}$ and $\bar{M}_B$, defined as
\begin{align}\label{eq:dist_canonicalEnsemble}
\lambda^A_u \deq \tr\{\bar{\Lambda}^A_u \rho_A\}, &\quad \lambda^B_v \deq \tr\{\bar{\Lambda}^B_v \rho_B\}, \nonumber \\
 \lambda^{AB}_{uv} \deq \tr\{(\bar{\Lambda}^A_u &\tensor \bar{\Lambda}^B_v) \rho_{AB}\}, \quad \text{and}\nonumber \\
\hat{\rho}^A_u \deq \frac{1}{\lambda^A_u}\sqrt{\rho_A}  \bar{\Lambda}^A_u \sqrt{\rho_A}, &\quad \hat{\rho}^B_v \deq \frac{1}{\lambda^B_v}\sqrt{\rho_B}\bar{\Lambda}^B_v \sqrt{\rho_B}, \quad \nonumber \\
 \hat{\rho}^{AB}_{uv} \deq \frac{1}{\lambda^{AB}_{uv}}\sqrt{\rho_{AB}} &(\bar{\Lambda}^A_u \tensor \bar{\Lambda}^B_v) \sqrt{\rho_{AB}}.
\end{align}
With this notation, corresponding to each of the probability distributions, we can associate a $\delta$-typical set. Let us denote $\mathcal{T}_{\delta}^{(n)}(U)$, $\mathcal{T}_{\delta}^{(n)}(V)$ and $\mathcal{T}_{\delta}^{(n)}(UV)$ as the $\delta$-typical sets defined for $\{\lambda^{A}_{u}\}$, $\{\lambda^{B}_{v}\}$ and $\{\lambda^{AB}_{uv}\}$, respectively. 

Let $\PiA$ and $\PiB$ denote the $\delta$-typical projectors (as in \cite[Def. 15.1.3]{Wilde_book}) for marginal density operators $\rho_A$ and $\rho_B$, respectively.
Also, for any $u^n\in \mathcal{U}^n$ and $v^n \in \mathcal{V}^n$, let $\PiuA$ and $\PivB$ denote the strong conditional typical projectors (as in \cite[Def. 15.2.4]{Wilde_book}) for the canonical ensembles $\{\lambda^A_u, \hat{\rho}^A_u\}$ and $\{\lambda^B_v, \hat{\rho}^B_v\}$, respectively.

% Let $\PiA$ and $\PiB$ denote the $\delta$-typical projectors (as in \cite{holevo}) for marginal density operators $\rho_A$ and $\rho_B$, respectively. Also, for any $u^n\in \mathcal{U}^n$ and $v^n \in \mathcal{V}^n$, let $\PiuA$ and $\PivB$ denote the conditional typical projectors (as in \cite{holevo}) for the canonical ensembles $\{\lambda^A_u, \hat{\rho}^A_u\}$ and $\{\lambda^B_v, \hat{\rho}^B_v\}$, respectively. 
For each $u^n\in \TDeltan(U)$ and $v^n \in \TDeltan(V)$ define 
\begin{equation*}
\rhotilduA \deq \PiA \PiuA \rhohatuA \PiuA \PiA, \quad 
\rhotildvB \deq \PiB \PivB \rhohatvB \PivB \PiB,
\end{equation*}
and $\rhotilduA = 0,  $ and $\rhotildvB = 0$ for $u^n \notin \TDeltan(U)$ and $v^n \notin \TDeltan(V)$, respectively, with $\rhohatuA \deq \bigotimes_{i} \hat{\rho}^A_{u_i}$ and $\rhohatvB \deq \bigotimes_{i} \hat{\rho}^B_{v_i}$. 
% \footnote{Note that $\Lambda^{A'}_{u^n}$ and $\Lambda^{B'}_{v^n}$ are not tensor products operators.}. 

\subsection{Construction of Structured POVMs}\label{subsec:stochastic_dist_POVM}
In what follows, we construct the random structured POVM elements. Fix a block length $n>0$, positive integers $ N_1 $ and $N_2$, and a finite field $ \FF_p $. Let $ \mu_1 \in [1,N_1]$ denote the common randomness shared between the first encoder and the decoder, and let $ \mu_2 \in [1,N_2]$ denote the common randomness shared between the second encoder and the decoder.
Let $\tilde{\mu}_1 \in [1,\tilde{N}_1]$ and $\tilde{\mu}_2\in [1,\tilde{N}_2]$ denote additional 
    pairwise shared randomness used for random coding purposes. This randomness is only used to show the existence of a desired distributed protocol (as defined in Definition \ref{def:distProtocol}), and is used only for bounding purposes. 
We denote $\bar{\mu}_i \deq (\mu_i,\tilde{\mu}_i)$, and 
$\bar{N}_i \deq N_i \cdot\tilde{N}_i$ for $i=1,2$.
Further, let $U$ and $ V $ be random variables defined on the alphabets $ \mathcal{U} $ and $ \mathcal{V} $, respectively, where $ \mathcal{U}=\mathcal{V}=\FF_p $. In building the code, we use the Unionized Coset Codes (UCCs)  \cite{pradhanalgebraic} as defined above in Definition \ref{def:UCC}.
% where the coarse code is a coset of the linear code and the fine code is the union of several cosets of the linear code.

% For a fixed $k \times n$ matrix $G \in \mathbb{F}_p^{k \times n}$ with 
% $k \leq n$, and a $1 \times n$
% vector $B \in \mathbb{F}_p^n$, define the coset code as 
% \[
% \mathbb{C}(G,B)\deq  \{ x^n: x^n = a^{k} G+B, \mbox{ for some } a^{k} \in \mathbb{F}_p^{k}
% \}.
% \]
% In other words, $\mathbb{C}(G,B)$ is a shift of the row space of
% the matrix $G$. The row space of $G$ is a linear code. 
% If the rank of $G$ is $k$, then there are $p^{k}$ codewords in the coset
% code. 

% \begin{definition}\label{def:UCC}
% An $(n,k,l,p)$ UCC is a pair $(G,h)$ consisting of a $k\times n$ matrix $G \in \mathbb{F}_p^{k \times n}$, and a mapping $h:\mathbb{F}_p^{l}
% \rightarrow \mathbb{F}_p^n$. In the context of UCC, define the composite code as
% $\mathbb{C}=\bigcup_{m \in \mathbb{F}_p^{l}} \mathbb{C}(G,h(m))$.
% \end{definition}

% \noindent \textbf{UCC Ensemble:} We construct a UCC ensemble by choosing $G$ uniformly from $\mathbb{F}_p^{k \times n}$, and $i(\cdot)$ uniformly from the set of all functions from $\mathbb{F}_p^{k}$ to $\mathbb{F}_p^{n}$. The two objects are chosen  independently. 

 For every $(\bar{\mu}_{1},\bar{\mu}_{2}) $, consider two UCCs $ (G^{},h_1^{(\bar{\mu}_1)}) $ and $ (G^{},h_2^{(\bar{\mu}_2)}) $, each with parameters $ (n,k,l_1,p) $ and $ (n,k,l_2,p) $, respectively. Note that, for every $(\bar{\mu}_1,\bar{\mu}_2), $ they share the same generator matrix $ G^{}. $

For each $ (\bar{\mu}_1,\bar{\mu}_2) $, the generator matrix $ G^{} $ along with the function $ h_1^{(\bar{\mu}_1)} $ and $ h_2^{(\bar{\mu}_2)} $ generates $ p^{k + l_1} $ and $ p^{k+l_2} $ codewords, respectively. Each of these codewords are characterized by a triple $ (a_i,m_i,\bar{\mu}_i)$, where $ a_i \in \FF^{k}_p $ and $ m_i \in \FF^{l_i}_p$  corresponds to the coarse code and the fine code indices, respectively, for $ i\in[1,2]$. Let $ \UBarCodeword $ and $ \VBarCodeword $ denote the codewords associated with   Alice and Bob,  generated using the above procedure, respectively, where 
\begin{align}
\UBarCodeword &\deq a_1 G^{} + h_1^{(\bar{\mu}_1)}(i) \quad \text{ and } \nonumber \\
\VBarCodeword &\deq a_2 G^{} + h_2^{(\bar{\mu}_2)}(j). \nonumber
\end{align}
Now, construct the operators
\begin{align}\label{eq:A_uB_v}
\bar{A}^{(\bar{\mu}_1)}_{u^n} &\deq  \alpha_{u^n} \bigg(\sqrt{\rho_{A}}^{-1}\rhotilduA\sqrt{\rho_{A}}^{-1}\bigg)\quad  \text{ and  } \nonumber \\
\bar{B}^{(\bar{\mu}_2)}_{v^n} &\deq \beta_{v^n} \bigg(\sqrt{\rho_{B}}^{-1}\rhotildvB\sqrt{\rho_{B}}^{-1}\bigg),
\end{align}
where 
\begin{align}\label{eq:gamma_mu}
\alpha_{u^n}& \deq  \frac{1}{(1+\eta)}\frac{p^n}{p^{k+l_1}}\lambda_{u^n}^A,\quad 
\beta_{v^n} \deq  \frac{1}{(1+\eta)}\frac{p^n}{p^{k+l_2}}\lambda_{v^n}^B, 
\end{align}
%Now, construct the operators
%\begin{align}\label{eq:A_uB_v}
%\bar{A}^{(\bar{\mu}_1)}_{u^n} \deq  \alpha_{u^n}\gamma^{(\bar{\mu}_1)}_{u^n} \bigg(\sqrt{\rho_{A}}^{-1}\LambdauA\sqrt{\rho_{A}}^{-1}\bigg)\quad  \text{ and  } \quad
%\bar{B}^{(\bar{\mu}_2)}_{v^n} \deq \beta_{v^n}\zeta^{(\bar{\mu}_2)}_{v^n} \bigg(\sqrt{\rho_{B}}^{-1}\LambdavB\sqrt{\rho_{B}}^{-1}\bigg),
%\end{align}
%where 
%\begin{align}\label{eq:gamma_mu}
%\gamma^{(\bar{\mu}_1)}_{u^n}& \deq  \frac{1-\varepsilon}{1+\eta}2^{-nS_1}|\{(a_1,m_1): \UCodeword = u^n|,\quad   \nonumber \\
%\zeta^{(\bar{\mu}_2)}_{v^n}& \deq  \frac{1-\varepsilon'}{1+\eta}2^{-nS_2}|\{(a_2,m_2): \VCodeword=v^n|, 
%\end{align}
% $ \alpha_{u^n} \deq |\mathcal{U}^n|\lambda_{u^n}^A $ and $ \beta_{v^n} \deq |\mathcal{V}^n|\lambda_{v^n}^B $,
 with $\eta \in (0,1)$ being a parameter to be determined.
Having constructed the operators $ \bar{A}^{(\bar{\mu}_1)}_{u^n} $ and $ \bar{B}^{(\bar{\mu}_2)}_{v^n} $, we normalize these operators, so that they constitute a valid sub-POVM. To do so, we first  define
\begin{align}
\Sigma_A^{(\bar{\mu}_1)} &\deq \sum_{u^n}\gammaBarCoeff\bar{A}^{(\bar{\mu}_1)}_{u^n} \quad \text{ and  } \nonumber\\
\Sigma_B^{(\bar{\mu}_2)} &\deq \sum_{v^n}\zetaBarCoeff\bar{B}^{(\bar{\mu}_2)}_{v^n}, \nonumber
\end{align} 
where $ \gammaBarCoeff $ and $ \zetaBarCoeff $ are defined as  
\begin{align}
\gammaBarCoeff & \deq |\{(a_1,i): \UBarCodeword = u^n\}| \quad \text{ and  } \nonumber \\
\zetaBarCoeff & \deq |\{(a_2,j): \VBarCodeword = v^n\}| .\nonumber
 \end{align}
Now, we define $ \CutOffBarA $ and $ \CutOffBarB $ as pruning operators for $\Sigma_A^{(\bar{\mu}_1)} $ and $ \Sigma_B^{(\bar{\mu}_2)}, $ with respect to $\PiA$ and $\PiB$, respectively (see Definition \ref{def:pruningOperator}).
Note that, these pruning operators, $\CutOffBarA$ and $\CutOffBarB$, depend on the triple $(G^{},h_1^{(\bar{\mu}_1)},h_2^{(\bar{\mu}_2)})$. 
% the cut-off projectors onto the subspace spanned by the eigen-states of $ \Sigma_A^{(\bar{\mu}_1)} $ and $ \Sigma_B^{(\bar{\mu}_2)}, $ with eigen-values less than 1, respectively. 
% For ease of analysis, the subspace of $ \CutOffA $ is restricted to $ \PiA $ and that of $ \CutOffB $ to $ \PiB $. 
Using these pruning operators, for each $\bar{\mu}_{1} \in [1,\bar{N}_{1}]$ and $ \bar{\mu}_{2} \in [1,\bar{N}_{2}]$, construct the sub-POVMs $M_1^{( n, \bar{\mu}_1)}$ and $M_2^{( n, \bar{\mu}_2)}$ as  
\begin{align}
M_1^{( n, \bar{\mu}_1)}& \deq \{\gammaBarCoeff A^{(\bar{\mu}_1)}_{u^n} : u^n \in \mathcal{U}^n\},\quad \text{ and  } \nonumber \\
M_2^{(n, \bar{\mu}_2)} &\deq \{\zetaBarCoeff B^{(\bar{\mu}_2)}_{v^n} : v^n \in  \mathcal{V}^n\},
\label{eq:POVM-14}
\end{align}   
where $ A^{(\bar{\mu}_1)}_{u^n} = \CutOffA\bar{A}^{(\bar{\mu}_1)}_{u^n}\CutOffA $ and $ B^{(\bar{\mu}_2)}_{v^n} = \CutOffB \bar{B}^{(\bar{\mu}_2)}_{v^n}\CutOffB $.
Further, using these operators $ \CutOffBarA  $ and $ \CutOffBarB, $ we have $ \sum_{u^n}\gammaBarCoeff A^{(\bar{\mu}_1)}_{u^n} = \CutOffBarA\Sigma_A^{(\bar{\mu}_1)}\CutOffBarA \leq I $ and $ \sum_{v^n}\zetaBarCoeff B^{(\bar{\mu}_2)}_{v^n} = \CutOffBarB\Sigma_B^{(\bar{\mu}_2)}\CutOffBarB \leq I,$ and thus $ M_1^{( n, \bar{\mu}_1)} $ and $ M_2^{(n, \bar{\mu}_2)} $ are valid sub-POVMs for all $ \bar{\mu}_1\in[1,\bar{N}_1] $ and $ \bar{\mu}_2 \in [1,\bar{N}_2]. $ Further, these collections $ M_1^{( n, \bar{\mu}_1)}$ and $ M_2^{( n, \bar{\mu}_2)}$ are completed using the operators $I - \sum_{u^n \in \mathcal{U}^n}\gammaBarCoeff A^{(\bar{\mu}_1)}_{u^n}$ and $I - \sum_{v^n\in\mathcal{V}^n}\zetaBarCoeff B^{(\bar{\mu}_2)}_{v^n}$.
% , and these operators are associated with sequences $u^n_0$ and $v^n_0$.
% , \checkComm{ which are additionally introduced in $ \mathcal{U}^n $ and $ \mathcal{V}^n $, respectively.}

\subsection{Binning of POVMs}\label{sec:POVM binning}
We next proceed to binning the above constructed collection of sub-POVMs.
% To do so, we employ a similar technique as was introduced in \cite{atif2019faithful}.
Since, UCC is already a union of several cosets, we associate a bin to each coset, and hence place all the codewords of a coset in the same bin. For each $i\in \FF_p^{l_1}$ and $j\in \FF_p^{l_2}$, let $\mathcal{B}^{(\bar{\mu}_1)}_1(i) \deq \mathbb{C}(G^{},h_1^{(\bar{\mu}_1)}(i))$ and $\mathcal{B}^{(\bar{\mu}_2)}_2(j) \deq \mathbb{C}(G^{},h_2^{(\bar{\mu}_2)}(j))$ denote the $i^{th}$ and the  $j^{th}$ bins, respectively. 
% More precisely, $\mathbb{C}(G^{},h_1^{(\bar{\mu}_1)})$ \deq \{a_1G^{}+h_1^{(\bar{\mu}_1)}(i):a_1 \in \FF_p^{k_1-l_1}\}$ and $\mathcal{B}^{(\bar{\mu}_2)}_2(j) \deq \{a_2G^{}+h_2^{(\bar{\mu}_2)}(j):a_2 \in \FF_p^{k_2-l_2}\}$. 
Formally, we define the following operators:
\begin{align*}
\Gamma^{A, ( \bar{\mu}_1)}_i &\deq  \sum_{u^n \in \mathcal{U}^n}\sum_{a_{1} \in \FF_p^{k}} A^{(\bar{\mu}_1)}_{u^n}\mathbbm{1}_{\{a_1G^{}+h_1^{(\bar{\mu}_1)}(i) = u^n\}},  \nonumber \\
 \Gamma^{B,( \bar{\mu}_2)}_j &\deq  \sum_{v^n \in \mathcal{V}^n}\sum_{a_{2} \in \FF_p^{k_2}} B^{(\bar{\mu}_2)}_{v^n}\mathbbm{1}_{\{a_2G^{}+h_2^{(\bar{\mu}_2)}(j)=v^n\}},
\end{align*}
for all  $i\in \FF_p^{l_1}$ and $j\in \FF_p^{l_2}$.  
Using these operators, we form the following collection:
% The above operators generate the following POVMs 
\begin{align}
M_A^{( n, \bar{\mu}_1)} \deq  \{\Gamma^{A,  (\bar{\mu}_1)}_i\}_{ i \in \FF_p^{l_1} }, \quad 
M_B^{( n, \bar{\mu})}\deq  \{\Gamma^{B, ( \bar{\mu}_2)}_j \}_{j \in \FF_p^{l_2} }.
\label{eq:POVM-16}
\end{align}
Note that if  $M_1^{( n, \bar{\mu}_1)}$ and $M_2^{( n, \bar{\mu}_2)}$ are sub-POVMs, then so are $M_A^{( n, \bar{\mu}_1)}$ and $M_B^{( n, \bar{\mu}_2)}$, which is due to the  relations
\begin{align}\label{eq:sumGamma}
\sum_{i\in \FF_p^{l_1}} \Gamma^{A, ( \bar{\mu}_1)}_i & = \sum_{u^n \in \mathcal{U}^n}\gammaCoeff A^{(\bar{\mu}_1)}_{u^n} \leq I,\quad \text{and} \nonumber \\
\sum_{j\in \FF_p^{l_2}} \Gamma^{B, ( \bar{\mu}_2)}_j & = \sum_{v^n\in \mathcal{V}^n}\zetaCoeff B^{(\bar{\mu}_2)}_{v^n} \leq I.
\end{align}
To make $ M_A^{( n, \bar{\mu}_1)} $ and $M_B^{( n, \bar{\mu}_2)}$ complete, we define $\Gamma^{A, ( \bar{\mu}_1)}_0$ and $ \Gamma^{B, ( \bar{\mu}_2)}_0 $ as $\Gamma^{A, ( \bar{\mu}_1)}_0=I-\sum_i \Gamma^{A, ( \bar{\mu}_1)}_i$ and $\Gamma^{B, ( \bar{\mu}_2)}_0=I-\sum_j \Gamma^{B, ( \bar{\mu}_2)}_j$, respectively\footnote{{Note that $\Gamma^{A, ( \bar{\mu}_1)}_0=I-\sum_i \Gamma^{A, ( \bar{\mu}_1)}_i = I - \sum_{u^n \in T_{\delta}^{(n)}(U)}A^{(\bar{\mu}_1)}_{u^n}$ and $\Gamma^{B, ( \bar{\mu}_2)}_0=I-\sum_j \Gamma^{B, ( \bar{\mu}_2)}_j = I - \sum_{v^n \in T_{\delta}^{(n)}(V)}B^{(\bar{\mu}_2)}_{v^n}$}.}. Now, we intend to use the completions $[M_A^{( n, \bar{\mu}_1)}]$ and $[M_B^{( n, \bar{\mu}_2)}]$ as the POVMs for encoders associated with Alice and Bob, respectively.
%We use the completion $[M_A^{( n, \bar{\mu}_1)}]$ and $[M_B^{( n, \bar{\mu}_2)}]$ as the POVMs for each encoder. Note that the index $i=0$ and $j=0$ are used, respectively, for $\Gamma^{A, ( \bar{\mu}_1)}_0=I-\sum_i \Gamma^{A, ( \bar{\mu}_1)}_i$ and $\Gamma^{B, ( \bar{\mu}_2)}_0=I-\sum_j \Gamma^{B, ( \bar{\mu}_2)}_j$.
Also, note that the effect of the binning is in reducing the communication rates from $(\frac{k+l_1}{n}\log{p}, \frac{k+l_2}{n}\log{p})$ to $(R_1,R_2)$, where $ R_i \deq \frac{l_i}{n}\log{p}, i\in\{1,2\}$. Now, we move on to describing the decoder.
%Let $M_X^{(n)}=\frac{1}{N}\sum_\bar{\mu}  M_X^{( n, \bar{\mu})}$ and  $M_Y^{(n)}=\frac{1}{N}\sum_\bar{\mu}  M_Y^{( n, \bar{\mu})} $. We use $( M_X^{(n)},  M_Y^{(n)}) $ as the two POVMs for each encoder. 

\subsection{Decoder mapping }

% Binning can be viewed as partitioning of the set of classical outcomes into bins. Suppose an outcome $(U^n,V^n)$ occurred after the measurement. Then, if the bins are small enough, one might be able to recover the outcomes by knowing the bin numbers. For that 
We create a decoder that takes as an input a pair of bin numbers and produces a sequence  $W^n \in \FF^n_p$. More precisely, we define a mapping $F^{(\bar{\mu}_1,\bar{\mu}_2)}$, acting on the outputs of $[M_A^{( n, \bar{\mu}_1)}] \tensor [M_B^{( n, \bar{\mu}_2)}]$ as follows. 
% Let $\mathcal{C}^{(\bar{\mu})}(i,j)$ denote the codebook containing all  pairs of  codewords  $(\UCodeword,\VCodeword)$. 
On observing $(\bar{\mu}_1,\bar{\mu}_2)$ and the classical indices $ (i,j) \in \FF_p^{l_1}\times \FF_p^{l_2} $ communicated by the encoder, the decoder constructs $D^{(\bar{\mu}_1,\bar{\mu}_2)}$ and $F^{(\bar{\mu}_1,\bar{\mu}_2)}(\cdot,\cdot)$ as, 
\begin{align}
D^{(\bar{\mu}_1,\bar{\mu}_2)}_{i,j} & \!\deq\!  \Big \{\!\tilde{a} \in \FF_p^{k}:   \tilde{a}G^{} + h_1^{(\bar{\mu}_1)}(i)+h_2^{(\bar{\mu}_2)}(j) \!\in\! \mathcal{T}_{\hat{\delta}}^{(n)}(W)\!\Big \}, \nonumber \\
F^{(\bar{\mu}_1,\bar{\mu}_2)}&(i,j)\deq \nonumber \\ &
\!\!\! \begin{cases}
	  \tilde{a}G^{} + h_1^{(\bar{\mu}_1)}(i)+h_2^{(\bar{\mu}_2)}(j) &\; \text{ if  } D^{(\bar{\mu}_1,\bar{\mu}_2)}_{i,j} \equiv \{\tilde{a}\} \\
	w^n_0 &\; \text{ otherwise  }, 
\end{cases}
\label{eq:POVM-17}
\end{align}
where $\hat{\delta} = p\delta$ and
% For every $ \bar{\mu} \in [1:N] $, $(i,j)\in \FF_p^{l_1}\times\FF_p^{l_2}$, define the function $F^{(\bar{\mu}_1,\bar{\mu}_2)}(i,j)=w^n$ if  $w^n$ is the only element of $D^{(\bar{\mu}_1,\bar{\mu}_2)}_{i,j}$; otherwise $F^{(\bar{\mu}_1,\bar{\mu}_2)}(i,j)=w_0^n$, where
{$w_0^n$ is an arbitrary sequence in $ \FF_p^n\backslash\mathcal{T}_{\hat{\delta}}^{(n)}(W) $}.
Further, $F^{(\bar{\mu}_1,\bar{\mu}_2)}(i,j)=w_0^n$ for $i=0$ or $j=0$. Given this, we obtain the sub-POVM 
$\tilde{M}_{AB}$ with the following operators. 
\begin{align*}
\tilde{\Lambda}_{w^n}^{AB} \deq \frac{1}{\bar{N}_1\bar{N}_2}\sum_{\bar{\mu}_1=1}^{\bar{N}_1}\sum_{\bar{\mu}_2=1}^{\bar{N}_2} 
\sum_{(i,j):F^{(\bar{\mu}_1,\bar{\mu}_2)}(i,j)=w^n} \hspace{-15pt}
\Gamma^{A, ( \bar{\mu}_1)}_i\tensor \Gamma^{B, (\bar{\mu}_2)}_j,% \mathcal{U}^n \times \mathcal{V}^n.   
\end{align*}
$
\forall w^n \in \FF_p^n \medcup \{w_0^n\}.$ Now, we use the stochastic mapping $ \P_{Z|W} $ to define the approximating sub-POVM $\hat{M}^{(n)}_{AB} \deq \{\hat{\Lambda}_{z^n}\}$ as
\begin{align*}
\hat{\Lambda}^{AB}_{z^n}\deq\sum_{w^n}  \tilde{\Lambda}_{w^n}^{AB}P^n_{Z|W}(z^n|w^n), ~ \forall z^n\in \mathcal{Z}^n.
%, \quad \forall (u^n,v^n) \in \mathcal{U}^n\times \mathcal{V}^n.
\end{align*}
{Note that $\tilde{\Lambda}_{w^n}^{AB} = 0$ for $w^n \notin \TDeltan(W)\medcup\{w^n_0\}.$}

\noindent \textbf{UCC Grand Ensemble:} The generator matrix $G^{}$ and the functions $h_1^{(\bar{\mu}_1)}$ and $h_2^{(\bar{\mu}_2)}$ are chosen randomly uniformly and independently, for $\bar{\mu}_1 \in [1,\bar{N}_1]$ and $\bar{\mu}_2 \in [1,\bar{N}_2].$ 

\subsection{Trace Distance}
In what follows, we show that $\hat{M}_{AB}^{(n)}$ is $\epsilon$-faithful to $M_{AB}^{\tensor n}$ with respect to $\rho_{AB}^{\tensor n}$ (according to Definition \ref{def:faith-sim}), where $\epsilon>0$ can be made arbitrarily small. More precisely, using \eqref{eq:nf_LambdaAB decompos}, we show that, $ \EE[K] \leq \epsilon, $ where 
%\error{with probability sufficiently close to 1} , the following inequality holds
%\begin{equation}\label{eq:faithful M hat}
%\sum_{z^n} \|\sqrt{\rho_{AB}^{\tensor n}} \left(\Lambda_{z^n}-\hat{\Lambda}_{z^n}\right) \sqrt{\rho_{AB}^{\tensor n}}\|_1\leq \epsilon.
%\end{equation}
%According to the decomposition of $\Lambda_{z}$, given in \eqref{eq:LambdaAB decompos}, the above inequality is equivalent to
\begin{align}\label{eq:nf_actual trace dist}
{K} &\deq \sum_{z^n} \bigg\| \sum_{u^n,v^n}\sqrt{\rho_{AB}^{\tensor n}} (\bar{\Lambda}^A_{u^n}\tensor \bar{\Lambda}^B_{v^n})\sqrt{\rho_{AB}^{\tensor n}}  \nonumber \\
& \hspace{0.4in}\times P^n_{Z|W}(z^n|u^n+ v^n)-   \sqrt{\rho_{AB}^{\tensor n}}\hat{\Lambda}_{z^n}^{AB}\sqrt{\rho_{AB}^{\tensor n}}\bigg\|_1,
%\sum_{z^n} \norm{ \sum_{u^n,v^n}\sqrt{\rho_{AB}^{\tensor n}} \bar{\Lambda}^A_{u^n}\tensor \bar{\Lambda}^B_{v^n}\sqrt{\rho_{AB}^{\tensor n}} P^n_{Z|W}(z^n|u^n+ v^n) - \sum_{w^n}\sqrt{\rho_{AB}^{\tensor n}}\tilde{\Lambda}_{w^n}^{AB}\sqrt{\rho_{AB}^{\tensor n}}P^n_{Z|W}(z^n|w^n)}_1
%\leq \epsilon.
\end{align}
and the expectation is with respect to the codebook generation.

\noindent \textbf{Step 1: Isolating the effect of error induced by not covering}\\ Consider the second term within $ {K} $, which can be written as
\begin{align}
\sum_{w^n}&\sqrt{\rho_{AB}^{\tensor n}}\tilde{\Lambda}_{w^n}^{AB}\sqrt{\rho_{AB}^{\tensor n}}P^n_{Z|W}(z^n|w^n)\nonumber\\
&= \frac{1}{\bar{N}_1\bar{N}_2}\sum_{\bar{\mu}_1,\bar{\mu}_2}\sum_{i,j} \sqrt{\rho_{AB}^{\tensor n}}\left (\Gamma^{A, ( \bar{\mu}_1)}_i\tensor \Gamma^{B, (\bar{\mu}_2)}_j\right ) \sqrt{\rho_{AB}^{\tensor n}}\nonumber \\
& \hspace{0.32in} \times P^n_{Z|W}(z^n|F^{(\bar{\mu}_1,\bar{\mu}_2)}(i,j))\underbrace{\sum_{w^n}\mathbbm{1}_{\{F^{(\bar{\mu}_1,\bar{\mu}_2)}(i,j) = w^n\}}}_{=1} \nonumber \\\vspace{-20pt}
& = T + \widetilde{T}, \nonumber
\end{align}
where \begin{align}
 T \deq & \frac{1}{\bar{N}_1\bar{N}_2}\sum_{\bar{\mu}_1,\bar{\mu}_2}\sum_{\{i>0\} \medcap \{j>0\}}\hspace{-10pt} \sqrt{\rho_{AB}^{\tensor n}}\left (\Gamma^{A, ( \bar{\mu}_1)}_i\tensor \Gamma^{B, (\bar{\mu}_2)}_j\right )\nonumber \\
  & \hspace{1.1in}\times\sqrt{\rho_{AB}^{\tensor n}} P^n_{Z|W}(z^n|F^{(\bar{\mu}_1,\bar{\mu}_2)}(i,j)), \nonumber \\
 \widetilde{T} \deq &\frac{1}{\bar{N}_1\bar{N}_2}\sum_{\bar{\mu}_1,\bar{\mu}_2}\sum_{\{i=0\}\medcup\{j=0\}}\hspace{-10pt} \sqrt{\rho_{AB}^{\tensor n}}\left (\Gamma^{A, ( \bar{\mu}_1)}_i\tensor \Gamma^{B, (\bar{\mu}_2)}_j\right )\nonumber \\
  & \hspace{1.7in}\times\sqrt{\rho_{AB}^{\tensor n}} P^n_{Z|W}(z^n|w^{n}_{0}). \nonumber 
\end{align}
Hence, we have
\begin{align}
    K \leq S+ \widetilde{S} \label{eq:G_separation_1},
\end{align}
where
\begin{align}
    S &\deq \sum_{z^n} \bigg\| \sum_{u^n,v^n}\sqrt{\rho_{AB}^{\tensor n}}\bigg( \bar{\Lambda}^A_{u^n}\tensor \bar{\Lambda}^B_{v^n}  \nonumber \\
    & \hspace{0.7in}  \times P^n_{Z|W}(z^n|u^n+ v^n)\bigg)\sqrt{\rho_{AB}^{\tensor n}} - T \bigg\|_1, \label{eq:def_S}
\end{align}
and $\widetilde{S} \deq \sum_{z^n}\|\widetilde{T}\|_1$. Note that $\widetilde{S}$ captures the error induced by not covering the state $\rho_{AB}^{\tensor n}.$
% \noindent{\bf Analysis of  $ S_{2}, S_3 $ and $S_4$:}
% 
% further bound it as 
% \begin{align}
% \tilde{S} \leq 
% \end{align}
% apply triangle inequality and obtain $\tilde{S} \leq \tilde{S}_1 + \tilde{S}_2.$
 
% \begin{proposition}\label{prop:Lemma for Tilde_S}
% For any $\epsilon \in (0,1)$, any $\eta, \delta \in (0,1)$ sufficiently small, and any $n$ sufficiently large, we have
% % There exist functions  $\epsilon_{\widetilde{S}}(\delta), $ and $\delta_{\widetilde{S}}(\delta)$, 
% % such that for  all sufficiently small $\delta$ and sufficiently large $n$, we have 
% $\EE[\widetilde{S}] \leq\epsilon$, if  $\frac{k+l_1}{n}\log{p} >  I(U;RB)_{\sigma_1} - S(U)_{\sigma_1}  +\log{p}  $ and $\frac{k+l_2}{n}\log{p} > I(V;RA)_{\sigma_2} - S(V)_{\sigma_2} + \log{p} ,$ where $\sigma_1$ and $\sigma_2$ are auxiliary states defined in the statement of the theorem.
% \end{proposition}
% \begin{proof}
% The proof is provided in Appendix \ref{appx:proof of Tilde_S}.
% \end{proof}
{\begin{remark}
The terms corresponding to the operators that complete the sub-POVMs $M_A^{(n,\bar{\mu}_1)}$ and $M_B^{(n,\bar{\mu}_2)}$, i.e., $ I - \sum_{u^n \in \TDeltaN(U)}\gammaBarCoeff A_{u^n}^{(\bar{\mu}_1)}$ and $ I - \sum_{v^n \in \TDeltaN(V)}\zetaBarCoeff B_{v^n}^{(\bar{\mu}_2)}$ are taken care in $\widetilde{T}$. The expression $T$ excludes these completing operators. 
% Therefore, we use $A_{u^n}^{(\bar{\mu}_1)}$  and $B_{v^n}^{(\bar{\mu}_2)} $ to denote the operators corresponding to $u^n \in \TDeltaN(U)$ and $v^n \in \TDeltaN(V)$, respectively.
\end{remark}
}
\noindent{\bf Step 2: Isolating the effect of error induced by binning}\\\
% For this, let us define $e^{(\bar{\mu})}(u^n,v^n)$ as   $e^{(\bar{\mu})}(u^n,v^n)\deq F^{(\bar{\mu}_1,\bar{\mu}_2)}(i,j)$, for each $(u^n, v^n) \in \mathcal{B}^{(\bar{\mu}_1)}_1(i)\times \mathcal{B}^{(\bar{\mu}_2)}_2(j)$
% and $(u^n,v^n)\in \mathcal{C}^{(\bar{\mu})}$
% . For any $(u^n, v^n) \notin \mathbbm{C}^{(\bar{\mu})}$ let $e^{(\bar{\mu})}(u^n,v^n)=w_0^n$.
We begin by simplifying $ T $ as
\begin{align}
T 
= & \frac{1}{\bar{N}_1\bar{N}_2}\sum_{\bar{\mu}_1,\bar{\mu}_2}\sum_{u^n,v^n}\sum_{\substack{i>0,\\j>0}} \sqrt{\rho_{AB}^{\tensor n}}  \nonumber \\ 
& \hspace{0pt} \times \!\bigg(\sum_{a_1\in \FF_p^{k}}\sum_{a_2 \in \FF_p^{k}}A_{u^n}^{(\bar{\mu}_1)}\tensor B_{v^n}^{(\bar{\mu}_2)}\mathbbm{1}_{\{\substack{a_1G^{} + h_1^{(\bar{\mu}_1)}(i) = u^n}\}}   \nonumber \\
& \hspace{0pt}\times\!\mathbbm{1}_{\{\substack{a_2G^{} + h_2^{(\bar{\mu}_2)}(j)  = v^n}\}}\!\bigg)\!\sqrt{\rho_{AB}^{\tensor n}} P^n_{Z|W}(z^n|F^{(\bar{\mu}_1,\bar{\mu}_2)}(i,j)). \nonumber 
\end{align}
{Note that the $(u^n,v^n)$ that appear in the above summation is confined to $(\TDeltaN(U)\times\TDeltaN(V))$, however for ease of notation, we do not make this explicit.} We substitute the above expression into $S$ as in \eqref{eq:def_S}, 
%to obtain
%\begin{align}
%{G} \leq S_1 + S_2 + S_3 + S_4, \label{eq:G}
%\end{align}
% \begin{align}
% S =  \sum_{z^n} &\left \| \sum_{u^n,v^n}\sqrt{\rho_{AB}^{\tensor n}}\left( \bar{\Lambda}^A_{u^n}\tensor \bar{\Lambda}^B_{v^n}P^n_{Z|W}(z^n|e^{(\bar{\mu})}(u^n+v^n)) -  \frac{1}{\bar{N}_1\bar{N}_2}\sum_{\bar{\mu}_1,\bar{\mu}_2}\gammaCoeff A_{u^n}^{(\bar{\mu}_1)}\tensor \zetaCoeff B_{v^n}^{(\bar{\mu}_2)}\right)\sqrt{\rho_{AB}^{\tensor n}}P^n_{Z|W}(z^n|e^{(\bar{\mu})}(u^n,v^n))\right \|_1. \nonumber
% \end{align}
and add and subtract an appropriate term within $ S $ and apply the triangle inequality to isolate the effect of binning as $ S \leq S_{1}  + S_{2},$ where 
\begin{align}
S_{1} & \deq \sum_{z^n} \bigg\|\sum_{u^n,v^n}\sqrt{\rho_{AB}^{\tensor n}}  \bigg (\bar{\Lambda}^A_{u^n}\tensor \bar{\Lambda}^B_{v^n}  -\frac{1}{\bar{N}_1\bar{N}_2} \sum_{\bar{\mu}_1,\bar{\mu}_2} \nonumber \\
& \hspace{05pt}\times \gammaBarCoeff\! A_{u^n}^{(\bar{\mu}_1)}\! \tensor\zetaBarCoeff \!B_{v^n}^{(\bar{\mu}_2)}\!\bigg )\sqrt{\rho_{AB}^{\tensor n}} P^n_{Z|W}(z^n|u^n+ v^n)\bigg\|_1\!\!, \nonumber \\
S_{2} & \deq \sum_{z^n}\bigg \| \frac{1}{\bar{N}_1\bar{N}_2}\!\! \sum_{\bar{\mu}_1,\bar{\mu}_2}\sum_{\substack{i>0\\j>0}}\sum_{a_1,a_2} \sum_{u^n,v^n} \!\! \sqrt{\rho_{AB}^{\tensor n}}\! \left (\! A_{u^n}^{(\bar{\mu}_1)}  \right.\nonumber \\ & \hspace{10pt}  \left.\tensor B_{v^n}^{(\bar{\mu}_2)}\right )\!
\sqrt{\rho_{AB}^{\tensor n}}\mathbbm{1}_{\{a_1G^{} + h_1^{(\bar{\mu}_1)}(i) = u^n, a_2G^{} + h_2^{(\bar{\mu}_2)}(j) = v^n\}}   \nonumber \\ 
& \hspace{15pt}\times \left(\!P^n_{Z|W}(z^n|u^n\!+\! v^n)-   P^n_{Z|W}\!\left (\!z^n|F^{(\bar{\mu}_1,\bar{\mu}_2)}(i,j)\!\right )\! \right)\!\!\bigg \|_1\!. \nonumber
\end{align}
%We start by applying a triangle inequality as in the following, where the binning effect is isolated.  
%\begin{align}\nonumber
%    \sum_{z^n} & \norm{\sum_{u^n,v^n}\sqrt{\rho_{AB}^{\tensor n}} \left(  \bar{\Lambda}^A_{u^n}\tensor \bar{\Lambda}^B_{v^n}\;P^n_{Z|W}(z^n|u^n+ v^n) -\tilde{\Lambda}_{u^n, v^n}\right) \sqrt{\rho_{AB}^{\tensor n}}}_1 \nonumber \\
%    &\leq \sum_{z^n} \left\|\sum_{u^n,v^n}\sqrt{\rho_{AB}^{\tensor n}}  \left (\bar{\Lambda}^A_{u^n}\tensor \bar{\Lambda}^B_{v^n} -\frac{1}{N_1N_2} \sum_{\bar{\mu}_1=1}^{N_1} \sum_{\bar{\mu}_2=1}^{N_2} \gammaCoeff A_{u^n}^{(\bar{\mu}_1)} \tensor\zetaCoeff B_{v^n}^{(\bar{\mu}_2)}\right )\sqrt{\rho_{AB}^{\tensor n}} P^n_{Z|W}(z^n|u^n+ v^n)\right\|_1  \nonumber \\
%    &   \hspace{10pt} +\sum_{z^n}\left \| 	  \sqrt{\rho_{AB}^{\tensor n}}\left (\frac{1}{N_1N_2} \sum_{\bar{\mu}_1=1}^{N_1} \sum_{\bar{\mu}_2=1}^{N_2}\gammaCoeff A_{u^n}^{(\bar{\mu}_1)} \tensor\zetaCoeff B_{v^n}^{(\bar{\mu}_2)} P^n_{Z|W}(z^n|u^n+ v^n)  -  \tilde{\Lambda}_{u^n, v^n}\right ) \sqrt{\rho_{AB}^{\tensor n}}\right \|_1 \label{eq:binningIsolation}
%%    &+\sum_{z^n} \norm{\sum_{u^n,v^n}\sqrt{\rho_{AB}^{\tensor n}} \left(  \bar{\Lambda}^A_{u^n}\tensor \bar{\Lambda}^B_{v^n}\;P^n_{Z|W}(z^n|u^n+ v^n) -\tilde{\Lambda}_{u^n, v^n}\right) \sqrt{\rho_{AB}^{\tensor n}}}_1
%	\end{align}
%\noindent Let us denote the first term on the RHS of the above equation \eqref{eq:binningIsolation} with $ S_1 $ and second term with $ S_2 $.
Note that the term $ S_{1} $ characterizes the error introduced by approximation of the original POVM with  the collection of  approximating sub-POVMs $ M_1^{(n,\bar{\mu}_1)} $ and $ M_2^{{(n,\bar{\mu}_2)}} $, and the term $ S_2 $ characterizes the error caused by binning of these approximating sub-POVMs. In this step, we analyze $S_2$ and prove the following proposition. 
%We start with the analysis of $ S_{2} $ and follow with $ S_{1} $ after that.

\begin{proposition}[Mutual Packing]\label{prop:Lemma for S_2}
For any $\epsilon \in (0,1)$, any $\eta, \delta \in (0,1)$ sufficiently small, and any $n$ sufficiently large, we have
% There exist  $\epsilon_{S_{2}}(\delta),$ such that for  all sufficiently small $\delta$ and sufficiently large $n$, we have 
$\EE\left[{S}_2\right] \leq \epsilon $, 
if $\frac{k+l_1}{n} \log p > I(U;RB)_{\sigma_1}-S(U)_{\sigma_3}+\log p$, $\frac{k+l_2}{n} \log p  > I(V;RA)_{\sigma_2}-S(V)_{\sigma_3}+\log p$, 
$\frac{k+l_1}{n} \log p +\frac{1}{n} \log \bar{N}_1 > \log p$, $\frac{k+l_2}{n} \log p +\frac{1}{n} \log \bar{N}_2 > \log p$, 
$ \frac{k}{n}\log{p} < \log{p} - S(W)_{\sigma_{3}}$,   where $\sigma_1, \sigma_2$  and $\sigma_3$ are the auxiliary states as defined in the statement of the theorem.
\end{proposition}
\begin{proof}
The proof is provided in Appendix \ref{appx:proof of S_2}
\end{proof}
Since averaged over $\tilde{\mu}_1 \in [1,\tilde{N}_1],\tilde{\mu}_2\in [1,\tilde{N}_2]$, the quantity
$\mathbb{E}[S_2]$ can be made arbitrarily small, there must exist a pair $(\tilde{\mu}_1,\tilde{\mu}_2)$ such that $\EE[S_2]$ is small for this pair of $(\tilde{\mu}_1,\tilde{\mu}_2)$.
For the rest of the proof, we fix $(\tilde{\mu}_1,\tilde{\mu}_2)$ to be this pair. The dependence of functions defined in the sequel on this pair is not made explicit for ease of notation.
For the term corresponding to $\widetilde{S}$, we  prove the following result.
\begin{proposition}\label{prop:Lemma for Tilde_S}
For any $\epsilon \in (0,1)$, any $\eta, \delta \in (0,1)$ sufficiently small, and any $n$ sufficiently large, we have
% There exist functions  $\epsilon_{\widetilde{S}}(\delta), $ and $\delta_{\widetilde{S}}(\delta)$, 
% such that for  all sufficiently small $\delta$ and sufficiently large $n$, we have 
$\EE[\widetilde{S}] \leq\epsilon$, if  $\frac{k+l_1}{n}\log{p} >  I(U;RB)_{\sigma_1} - S(U)_{\sigma_1}  +\log{p}  $ and $\frac{k+l_2}{n}\log{p} > I(V;RA)_{\sigma_2} - S(V)_{\sigma_2} + \log{p} ,$ where $\sigma_1$ and $\sigma_2$ are auxiliary states defined in the statement of the theorem.
\end{proposition}
\begin{proof}
The proof is provided in Appendix \ref{appx:proof of Tilde_S}.
\end{proof}

\noindent\textbf{Step 3: Isolating the effect of Alice's approximating measurement}\\
\noindent In this step, we separately analyze the effect of approximating measurements  at the two distributed parties in the term $S_1$. For that, we split $S_{1}$ as $ S_{1} \leq Q_1 + Q_2 $, where
% Now, we further split $S_{1}$ by adding and subtracting an appropriate term and using triangle inequality as   where

\begin{widetext}
\begin{align}
Q_1 &\deq \sum_{z^n} \bigg\|\!\sum_{u^n,v^n}\!\!\!\sqrt{\rho_{AB}^{\tensor n}}  \bigg (\!\bar{\Lambda}^A_{u^n}\tensor \bar{\Lambda}^B_{v^n} -\frac{1}{N_1}\! \sum_{\mu_1=1}^{N_1}\!\! \gammaCoeff A_{u^n}^{(\mu_1)}  \tensor \bar{\Lambda}^B_{v^n}\bigg )\sqrt{\rho_{AB}^{\tensor n}} P^n_{Z|W}(z^n|u^n+ v^n)\bigg\|_1\!\!, \nonumber \\
Q_2 &\deq \sum_{z^n} \left\|\frac{1}{N_1} \sum_{\mu_1=1}^{N_1}\sum_{u^n,v^n}\sqrt{\rho_{AB}^{\tensor n}}  \left ( \gammaCoeff A_{u^n}^{(\mu_1)} \tensor \bar{\Lambda}^B_{v^n}-\frac{1}{N_2} \sum_{\mu_2=1}^{N_2} \gammaCoeff A_{u^n}^{(\mu_1)} \tensor\zetaCoeff B_{v^n}^{(\mu_2)}\right )\sqrt{\rho_{AB}^{\tensor n}} P^n_{Z|W}(z^n|u^n+ v^n)\right\|_1.\label{eq:binningIsolation}
\end{align} 
\end{widetext}

 With this partition, the terms within the trace norm of $ Q_1 $ differ only in the action of Alice's measurement. And similarly, the terms within the norm of $ Q_2 $ differ only in the action of Bob's measurement. Showing that these two terms are small forms a major portion of the achievability proof. 
%\noindent Let us denote the terms on the RHS of the above equation \eqref{eq:binningIsolation} using $ Q_1 $ and $ Q_2 $.
% \noindent The above partition into $ Q_1 $ and $ Q_2  $ is performed so as to separately analyze the effect of approximating the measurements at the two distributed parties. With this partition, the terms within the trace norm of $ Q_1 $ differ only in the action of Alice's measurement. And similarly, the terms within the norm of $ Q_2 $ differ only in the action of Bob's measurement. Showing that these two terms are small forms a major portion of the achievability proof. 

\noindent{\bf Analysis of $ Q_1$:}  
% We start by finding the rate constraints such that the term corresponding to $ Q_1 $ can be made small. 
To prove $ Q_1 $ is small,
% we use a technique analogous to the one in \cite{atif2020source}. 
we characterize the rate constraints which ensure that an upper bound to $ Q_1 $ can be made to vanish in an expected sense.
% Using this upper bound and monotonicity of trace norm, we prove that the rate constraints above also guarantee that the expectation of term corresponding to $ Q_1 $ can be made arbitrarily small as well. 
In addition, this upper bound becomes lucrative in obtaining a single-letter characterization for the rate needed to make the term corresponding to $ Q_2 $ vanish. For this, we define $ J $ as
\begin{align}\label{def:J}
J &\deq \sum_{z^n,v^n}\bigg\|\sum_{u^n}\sqrt{\rho_{AB}^{\tensor n}}  \bigg (\bar{\Lambda}^A_{u^n}\tensor \bar{\Lambda}^B_{v^n} - \nonumber \\ & \hspace{7pt}\frac{1}{N_1}\! \sum_{\mu_1=1}^{N_1} \!	 \gammaCoeff  A_{u^n}^{(\mu_1)} \tensor  \bar{\Lambda}^B_{v^n}\bigg )\sqrt{\rho_{AB}^{\tensor n}} P^n_{Z|W}(z^n|u^n+ v^n)\bigg\|_1\!\!\!.
\end{align}
By defining $J$ and using triangle inequality for block operators (which holds with equality), we add the sub-system $V$ to $RZ$, resulting in the joint system $RZV$, corresponding to the state $\sigma_3$ as defined in the theorem. Then we approximate the joint system $RZV$ using an approximating sub-POVM $M_A^{(n)}$ producing outputs on the alphabet $\mathcal{U}^n$. To make $J$ small for sufficiently large n, we expect the sum of the rate of the approximating sub-POVM and common randomness, i.e., $\frac{k+l_1}{n}\log{p} + \frac{1}{n}\log{N_1}$, to be larger than $I(U;RZV)_{\sigma_3}$. We prove this in the following.

% in the approximation, resulting in the approximation of the sub-system $RZV$, using ,  
Note that from the triangle inequality, we have $ Q_1 \leq J. $
% We can simplify $ J $ by using the definition of $ A^{(\mu_1)}_{u^n} $ and $ B^{(\mu_2)}_{v^n} $ from \eqref{eq:nf_A_uB_v} as 
Further, we add and subtract appropriate terms
%\begin{align}
%\sum_{u^n}\frac{1}{N_1} \sum_{\mu_1=1}^{N_1} \sqrt{\rho_{AB}^{\tensor n}}  \left (\frac{\gamma_{u^n}^{(\mu_1)}}{\lambdauA}\bar{\Lambda}^A_{u^n}\tensor \bar{\Lambda}^B_{v^n}\right )\sqrt{\rho_{AB}^{\tensor n}} P^n_{Z|W}(z^n|u^n+ v^n)\nonumber
%\end{align} 
within $ J$, and again use the triangle inequality to obtain $ J \leq J_1 + J_2$, where
% \begin{widetext}
\begin{align}
J_1 & \deq  \sum_{z^n,v^n}\bigg\|\sum_{u^n}\sqrt{\rho_{AB}^{\tensor n}}  \bigg (\bar{\Lambda}^A_{u^n}\tensor \bar{\Lambda}^B_{v^n} - \nonumber \\
& \hspace{33pt}\gammaCoeff  \bar{A}_{u^n}^{(\mu_1)} \tensor  \bar{\Lambda}^B_{v^n} \bigg )\sqrt{\rho_{AB}^{\tensor n}} P^n_{Z|W}(z^n|u^n+ v^n)\bigg\|_1, \nonumber \\
% J_2 &\deq \sum_{z^n,v^n}\left\|\frac{1}{N_1}\sum_{\mu_1=1}^{N_1}\sum_{u^n}\sqrt{\rho_{AB}^{\tensor n}}  \left ( \frac{\alpha_{u^n}\gamma_{u^n}^{(\mu_1)}}{\lambdauA}\bar{\Lambda}^A_{u^n}\tensor \bar{\Lambda}^B_{v^n} -	\gammaCoeff  \bar{A}_{u^n}^{(\mu_1)} \tensor  \bar{\Lambda}^B_{v^n}\right )\sqrt{\rho_{AB}^{\tensor n}} P^n_{Z|W}(z^n|u^n+ v^n)\right\|_1, \nonumber \\
J_2 &\deq \sum_{z^n,v^n}\bigg\|\frac{1}{N_1} \sum_{\mu_1=1}^{N_1}\sum_{u^n}\sqrt{\rho_{AB}^{\tensor n}}  \bigg ( 	\gammaCoeff  \bar{A}_{u^n}^{(\mu_1)} \tensor  \bar{\Lambda}^B_{v^n} - \nonumber \\
& \hspace{33pt}	\gammaCoeff  A_{u^n}^{(\mu_1)} \tensor  \bar{\Lambda}^B_{v^n}\bigg )\sqrt{\rho_{AB}^{\tensor n}} P^n_{Z|W}(z^n|u^n+ v^n)\bigg\|_1. \nonumber
\end{align}
% \end{widetext}
Now we use the following proposition to bound the term corresponding to $J_1$.
 \begin{proposition}\label{prop:Lemma for p2p:J_1}
 For any $\epsilon \in (0,1)$, any $\eta, \delta \in (0,1)$ sufficiently small, and any $n$ sufficiently large, we have $\EE\left[J_{1}\right] \leq \epsilon$ if $\frac{k+l_1}{n} \log p + \frac{1}{n}\log{N_1} > I(U;RZV)_{\sigma_3} +\log{p} - S(U)_{\sigma_3}$, where  $\sigma_3$ is the auxiliary state defined in the statement of the theorem.
% There exist  $\epsilon_{J_{1}}(\delta), \delta_{J_{1}}(\delta)$ such that for  all sufficiently small $\delta$ and sufficiently large $n$, we have $\EE\left[J_{1}\right] \leq \epsilon_{J_{1}}$ if $S_1 + \frac{1}{n}\log{N_1} \geq I(U;RZV)_{\sigma_3} +\log{p} - S(U)_{\sigma_3} + \delta_{J_{1}}$, where  $\sigma_3$ is the auxiliary state defined in the theorem and $\epsilon_{{J}_{1}}, \delta_{{J}_{1}} \searrow 0$ as $\delta \searrow 0$.
\end{proposition}
\begin{proof}
The proof of proposition is provided in Appendix \ref{appx:proof of J_1}.
\end{proof}

Now we move on to bounding the term corresponding to $ J_2. $ We start by applying triangle inequality followed by Lemma \ref{lem:Separate} on $ J_2 $ to obtain
\begin{align}
J_2 \leq & \sum_{z^n}\sum_{u^n,v^n}P^n_{Z|W}(z^n|u^n+ v^n)\bigg\|\frac{1}{N_1} \sum_{\mu_1=1}^{N_1}\sqrt{\rho_{A}^{\tensor n}} \nonumber \\
&\hspace{15pt} \times  \left(\left ( 	\gammaCoeff  \bar{A}_{u^n}^{(\mu_1)} - 	\gammaCoeff  A_{u^n}^{(\mu_1)} \right )\tensor \bar{\Lambda}^B_{v^n} \right)\sqrt{\rho_{A}^{\tensor n}} \bigg\|_1 \nonumber \\
= & \sum_{u^n,v^n}\bigg\|\frac{1}{N_1} \sum_{\mu_1=1}^{N_1}\sqrt{\rho_{A}^{\tensor n}}  \bigg (\left(	\gammaCoeff  \bar{A}_{u^n}^{(\mu_1)} - 	\gammaCoeff  A_{u^n}^{(\mu_1)} \right )\nonumber \\
& \hspace{2.1in}\tensor \bar{\Lambda}^B_{v^n} \bigg)\sqrt{\rho_{A}^{\tensor n}} \bigg\|_1 \nonumber \\
\leq & \frac{1}{N_1} \sum_{\mu_1=1}^{N_1} \sum_{u^n}\gammaCoeff\left\|\sqrt{\rho_{A}^{\tensor n}}  \left (\bar{A}_{u^n}^{(\mu_1)} -	A_{u^n}^{(\mu_1)} \right )\sqrt{\rho_{A}^{\tensor n}}\right\|_1\!. \label{eq:term J3}
\end{align}
%where the second inequality uses triangle inequality. 
Now we use the following proposition to bound the term corresponding to $J_2.$
 \begin{proposition}\label{prop:Lemma for J_2}
 For any $\epsilon \in (0,1)$, any $\eta, \delta \in (0,1)$ sufficiently small, and any $n$ sufficiently large, we have  $\EE\left[J_{2}\right] \leq \epsilon$ if $\frac{k+l_1}{n}\log{p} > I(U;RB)_{\sigma_1} +\log{p} - S(U)_{\sigma_3}$, where $\sigma_1$ and $\sigma_3$ are the auxiliary states defined in the statement of the theorem.
% There exist  $\epsilon_{J_{2}}(\delta), \delta_{J_{2}}(\delta)$ such that for  all sufficiently small $\delta$ and sufficiently large $n$, we have $\EE\left[J_{2}\right] \leq \epsilon_{J_{2}}$ if $S_1 \geq I(U;RB)_{\sigma_1} +\log{p} - S(U)_{\sigma_3} + \delta_{J_2}$, where $\sigma_1$, $\sigma_3$ are the auxiliary state defined in the theorem and $\epsilon_{{J}_{2}}, \delta_{{J}_2} \searrow 0$ as $\delta \searrow 0$.
\end{proposition}
\begin{proof}
% The proof follows from Proposition \ref{prop:Lemma for p2p:S_13}
The proof is provided in Appendix \ref{appx:proof of J_2}.
\end{proof}

Since $Q_1 \leq J \leq J_1 + J_2$, hence $\mathbb{E}[J]$, and consequently $\mathbb{E}[Q_1]$, can be made arbitrarily small for sufficiently large n, if $\frac{k+l_1}{n}\log{p}+\frac{1}{n}\log{N_1} > I(U;RZV)_{\sigma_{3}} - S(U)_{\sigma_{3}} + \log{p} $ and $\frac{k+l_1}{n}\log{p} > I(U;RB)_{\sigma_{1}} -S(U)_{\sigma_{3}} + \log{p} $. Now we move on to bounding $ Q_2 $.

\noindent{\bf Step 4: Analyzing the effect of Bob's approximating measurement}\\
Step 3 ensured that the sub-system $RZV$ is close to a tensor product state in trace-norm. In this step, we approximate the state corresponding to the sub-system $RZ$ using the approximating POVM $M_B^{(n)}$, producing outputs on the alphabet $\mathcal{V}^n$. We proceed with the following proposition. 
%\begin{align}
%Q_2 = \sum_{z^n} \left\|\sum_{u^n,v^n}\sqrt{\rho_{AB}^{\tensor n}}  \left ( \frac{1}{N_1} \sum_{\mu_1=1}^{N_1}A_{u^n}^{(\mu_1)} \tensor \bar{\Lambda}^B_{v^n} -\frac{1}{N_1N_2} \sum_{\mu_1,\mu_2} \gammaCoeff A_{u^n}^{(\mu_1)} \tensor\zetaCoeff B_{v^n}^{(\mu_2)}\right )\sqrt{\rho_{AB}^{\tensor n}} P^n_{Z|W}(z^n|u^n+ v^n)\right\|_1
%\end{align}
\begin{proposition} \label{prop:Lemma for Q2}
For any $\epsilon \in (0,1)$, any $\eta, \delta \in (0,1)$ sufficiently small, and any $n$ sufficiently large, we have $\EE[{Q}_2] \leq\epsilon$, if $ \frac{k+l_1}{n}\log{p}+\frac{1}{n}\log{N_1} >  I(U;RZV)_{\sigma_3} - S(U)_{\sigma_{3}} + \log{p}$, $ \frac{k+l_2}{n}\log{p}+\frac{1}{n}\log{N_2} > I(V;RZ)_{\sigma_3} - S(V)_{\sigma_{3}} + \log{p}$,  $ \frac{k+l_1}{n}\log{p} > I(U;RB)_{\sigma_{1}} - S(U)_{\sigma_{3}} + \log{p} $,  and $ \frac{k+l_2}{n}\log{p} > I(V;RA)_{\sigma_{2}} - S(V)_{\sigma_{3}} + \log{p}$  where $\sigma_1, \sigma_2$, $\sigma_3$ are the auxiliary states defined in the statement of the theorem.
% There exist functions  $\epsilon_{Q_{2}}(\delta) $ and $\delta_{{Q}_{2}}(\delta)$, such that for  all sufficiently small $\delta$ and sufficiently large $n$, we have $\EE[{Q}_2] \leq\epsilon_{{Q_2}}$, if $ S_1+\frac{1}{n}\log{N_1} \geq  I(U;RZV)_{\sigma_3} - S(U)_{\sigma_{3}} + \log{p} + \delta_{Q_2}$, $ S_2+\frac{1}{n}\log{N_2} \geq  I(V;RZ)_{\sigma_3} - S(V)_{\sigma_{3}} + \log{p} + \delta_{Q_2}$,  $ S_1 \geq I(U;RB)_{\sigma_{1}} - S(U)_{\sigma_{3}} + \log{p} + \delta_{Q_2}$,  and $ S_2 \geq I(V;RA)_{\sigma_{2}} - S(V)_{\sigma_{3}} + \log{p} + \delta_{Q_2}$  where $\sigma_1, \sigma_2$, $\sigma_3$ are the auxiliary states defined in the theorem and $\epsilon_{Q_2},\delta_{{Q}_2}  \searrow 0$ as $\delta \searrow 0$.
%For all $\epsilon_{Q_2},\delta_{Q_2} \in (0,1)$, and all sufficiently large $n$, we have $\PP(Q_2 > \epsilon_{Q_2})\leq \delta_{Q_2}$, if $S_2+C_2 >  I(V;RZ)_{\sigma_4}.$
\end{proposition}

\begin{proof}
 The proof is provided in Appendix \ref{appx:proof of Q2}.
\end{proof}

 \subsection{Rate Constraints}
 To sum-up, we showed $\EE[K] \leq \epsilon$
%  that the trace distance inequality in \eqref{eq:nf_actual trace dist} 
holds for sufficiently large $n$
% in the expected sense, 
if the following bounds hold:
 \begin{subequations}\label{eq:nf_rate-region 1}
 	\begin{align}
 	\tilde{R}+R_1 &> I(U; RB)_{\sigma_1} - S(U)_{\sigma_{3}} + \log{p}  ,\\
 	\tilde{R}+R_2 &> I(V; RA)_{\sigma_2} - S(V)_{\sigma_{3}} + \log{p} ,\\
% 	C+S_1  \geq S(U)_{\sigma_3}, \quad C+S_2 & \geq S(V)_{\sigma_3}, \quad  C+S_1+S_2  \geq  S(U)_{\sigma_3}+S(V)_{\sigma_3}\\
	\tilde{R}+R_1 + C_1  &> I(U; RZV)_{\sigma_3} - S(U)_{\sigma_{3}} + \log{p} ,\\
	\tilde{R}+R_2 + C_2  &> I(V; RZ)_{\sigma_3} - S(V)_{\sigma_{3}} + \log{p}  ,\\
% 	S_1 - R_1  &= S_2 - R_2,\\
   0 \leq  \tilde{R} &< \log{p} - S(U+ V)_{\sigma_3},\\
%  	&S_2-R_2  < \log{p} - S(U+ V)_{\sigma_3},\\
    C_1  &\geq 0,\quad C_2\geq 0,
 	\end{align}
 \end{subequations}
 where $C_i\deq \frac{1}{n}\log_2 N_i, i \in \{1,2\}$ and $\tilde{R}\deq \frac{k}{n}\log{p}$. 
% This implies that $\tilde{M}_{AB}$ is $\epsilon$-faithful to $\bar{M}_A\tensor \bar{M}_B$ with probability sufficiently close to one, and hence, $\hat{M}_{AB}$ is also $\epsilon$-faithful to $M_{AB}$, i.e, \eqref{eq:faithful M hat} is satisfied. 
Therefore, there exists a distributed protocol with parameters $(n, 2^{nR_1}, 2^{nR_2}, 2^{nC_1}, 2^{nC_2})$ such that its overall POVM $\hat{M}_{AB}$ is $\epsilon$-faithful to $M_{AB}^{\tensor n}$ with respect to $\rho_{AB}^{\tensor n}$. 

Let us denote the above achievable rate-region by $\CalR_1$. By doing an exact symmetric analysis, but by replacing the first encoder by a product distribution instead of the second encoder in $S_1$ (as performed in \eqref{eq:binningIsolation}), all the constraints remain the same, except that the constraints on $\tilde{R}+R_1 + C_1$ and $\tilde{R}+R_2 + C_2$ change as follows
\begin{align}
    \tilde{R}+R_1 + C_1 &\geq I(U;RZ)_{\sigma_3} - S(U)_{\sigma_{3}} + \log{p}, \quad\nonumber\\
    \tilde{R}+R_2 + C_2 &\geq  I(V;RZU)_{\sigma_3} - S(V)_{\sigma_{3}} + \log{p} . \label{SCrate2}
\end{align}
Let us denote the above achievable rate-region by $\mathcal{R}_2$.
By time sharing between the any two points of $\mathcal{R}_1$ and $\mathcal{R}_2$ one can achieve any point in the convex closure of $(\CalR_1\bigcup\CalR_2).$ 
The following lemma gives a symmetric  characterization of the closure of convex hull of the union of the above achievable rate-regions.

\begin{lem}
For the above defined rate regions $\CalR_1$ and $\CalR_2$, we have $\CalR_3 = \text{Convex Closure}(\CalR_1 \bigcup \CalR_2)$, where $\CalR_3$ is given by the set of all the quintuples $(\tilde{R}, R_1, R_2, C_1, C_2)$ satisfying the following constraints: 
\begin{align}\label{eq:rate-regionR3}
    \tilde{R}+R_1 & \geq I(U;RB)_{\sigma_1}- S(U)_{\sigma_{3}} + \log{p}, \nonumber \\
    \tilde{R}+R_2  &\geq I(V;RA)_{\sigma_2}- S(V)_{\sigma_{3}} + \log{p}, \nonumber \\ %+4\delta
    \tilde{R}+R_1+ C_1 & \geq I(U;RZ)_{\sigma_3}- S(U)_{\sigma_{3}} + \log{p},\nonumber \\ %+\delta'
    \tilde{R}+R_2+ C_2  &\geq I(V;RZ)_{\sigma_3}- S(V)_{\sigma_{3}} + \log{p}, \nonumber \\ %+\delta'
    % 2\tilde{R}+R_1+R_2 +C_1 +C_2 &\geq I(U;RZ)_{\sigma_3} + I(V;RZ)_{\sigma_3} + I(U;V|RZ)_{\sigma_3} - S(U)_{\sigma_{3}} -  S(V)_{\sigma_{3}} + 2\log{p} , \nonumber \\ 
    2\tilde{R}\!+\!R_1\!+R_2 +C_1 +C_2 &\geq I(UV;RZ)_{\sigma_3} - S(U,V)_{\sigma_{3}}\nonumber \\
    & \hspace{10pt}+ 2\log{p} , \nonumber \\ 
   0 \leq  \tilde{R}  &\leq \log{p} - S(U+V)_{\sigma_3}, \\
     R_1 \geq 0, R_2 &\geq  0 \quad
    % 0\leq & \;C_1+C_2\leq C
    C_1  \geq 0, C_2 \geq 0. 
\end{align}
\end{lem} 
\begin{proof}
The proof follows from elementary convex analysis.
\end{proof}

 Lastly,  we complete the proof of the theorem using the following lemma.
\begin{lem}
Let $\bar{\mathcal{R}}_3$ denote the set of all quadruples $(R_1,R_2,C_1,C_2)$ for which there exists $\tilde{R}$ such that the quintuple $(R_1,R_2, C_1, C_2, \tilde{R})$ satisfies the inequalities in (\ref{eq:rate-regionR3}). Let $\mathcal{R}_F$ denote the set of all quadruples $(R_1, R_2, C_1, C_2)$ that satisfy  the inequalities in \eqref{eq:nf_dist POVm appx rates} given in the statement of the theorem. Then, $
\bar{\mathcal{R}}_3=\mathcal{R}_F$.
\end{lem}
%  \begin{lem}
%  	Let $\mathcal{R}_1$ denote the set of all $(R_1,R_2,C_1,C_2)$ for which there exists $\tilde{R}$ such that the septuple $(R_1,R_2,\tilde{R}, {C}_1, {C}_2)$ satisfies the inequalities in \eqref{eq:nf_rate-region 1}. Let, $\mathcal{R}_2$ denote the set of all quadruples $(R_1, R_2, C_1,C_2)$ that satisfies  the inequalities in \eqref{eq:nf_dist POVm appx rates} given in the statement of the theorem. Then, $
%  	\mathcal{R}_1=\mathcal{R}_2$.
%  \end{lem}
 \begin{proof}
 	The proof follows from Fourier-Motzkin elimination \cite{fourier-motzkin}.
%  	For that, we eliminate $(S_1, S_2)$ from the system of inequalities given by \eqref{eq:rate-region 1}. This gives us an equivalent rate-region described by all $(C,R_1,R_2)$ that satisfies the following set of inequalities:
%  	\begin{subequations}
%  		\begin{align}
%  		R_1 &\geq I(U;RB)_{\sigma_1}-I(U;V)_{\sigma_3},\\
%  		R_2 &\geq I(V;RA)_{\sigma_2}-I(U;V)_{\sigma_3},\\
%  		R_1+R_2 &\geq I(U;RB)_{\sigma_1}+I(V;RA)_{\sigma_2}-I(U;V)_{\sigma_3},\\
%  		R_1+R_2+C &\geq  I(UV;RZ)_{\sigma_3}, 
% % 		R_1+R_2+C &\geq  I(U;RVZ)_{\sigma_4}+I(V;RZ)_{\sigma_5}-I(U;V)_{\sigma_3},  
% % 		C+R_2 &\geq S(V|U)_{\sigma_3},\\
% % 		C+R_1 &\geq S(U|V)_{\sigma_3},\\
% % 		2C+R_1+R_2 &\geq S(U,V)_{\sigma_3},\label{subeq:C+R inactive 1}\\
% % 		C+R_1+R_2 &\geq S(U,V)_{\sigma_3},\label{subeq:C+R1+R2}\\
% % 		C+R_1+R_2 &\geq I(V;RA)_{\sigma_1}+S(U|V)_{\sigma_3},\label{subeq:C+R inactive 2}\\
% % 		C+R_1+R_2 &\geq I(U;RB)_{\sigma_2}+S(V|U)_{\sigma_3}\label{subeq:C+R inactive 3}
%  		\end{align}
%  	\end{subequations}
% % 	Having \eqref{subeq:C+R1+R2}, the bounds given by \eqref{subeq:C+R inactive 1},\eqref{subeq:C+R inactive 2}, and \eqref{subeq:C+R inactive 3} are redundant. 
% Hence we get the rate-region $\mathcal{R}_2$, and the proof is complete.
 	\end{proof}
\section{Conclusion}\label{sec:conclusion}
We developed a technique of randomly generating structured POVMs using algebraic codes. Using this technique, we demonstrated a new achievable information-theoretic rate-region for the task of faithfully simulating a distributed quantum measurement and function computation. We further devised a Pruning Trace inequality which is a tighter version of the known operator Markov inequality, and a covering lemma which is independent of the operator Chernoff inequality, so as to analyse pairwise-independent POVM elements. Finally, combining these techniques,  we demonstrated rate gains for this problem over traditional coding schemes, and provided a multi-party distributed faithful simulation and function computation protocol. 

\textbf{Acknowledgement:} We thank Arun Padakandla for his valuable discussion and inputs in developing the proof techniques.

\appendix

\section{Proof of Lemmas}
% \addtocontents{toc}{\protect\setcounter{tocdepth}{1}}

%%------------------xxxxxxxxxxxxxxx----------------

% \subsection{Proof of Lemma \ref{lem:trace_expec}}
% \begin{proof}\label{appx:proof of trace_expec}
%     Since trace operation in finite dimensions is linear, we can write 
% 	\begin{align}
% 	\EE\left[\Tr{\sqrt{A^{\dagger}A}}\right] = \Tr{\EE\left[\sqrt{A^{\dagger}A}\right]}. \nonumber
% 	\end{align}
% 	Further, since square-root is an operator concave function (Theorem $2.6$ in \cite{carlen2010trace}), we write
% 	\begin{align}
% 	\EE\left[\sqrt{A^{\dagger}A}\right] \leq \sqrt{\EE\left[A^{\dagger}A\right]}.\nonumber
% 	\end{align}
% 	Tracing both sides of the above equation gives the required result.
% \end{proof}

%%------------------xxxxxxxxxxxxxxx----------------

\subsection{Proof of Lemma \ref{lem:Change Measure Soft Covering Variance Based}}
\begin{proof} \label{appx:ProofUniformSoftCovering}
We begin by defining the ensemble $\{\lambda_x,\tilde{\sigma}_x\}_{x \in \mathcal{X}}$ where  $ \tilde{\sigma}_x = \Pi\Pi_x\sigma_x\Pi_{x}\Pi $ for all $\eleInSet{x}{X} $. Further, let $ S$ be defined as
\begin{align}
    S \deq  \Big\|\sum_{\eleInSet{x}{X}}\lambda_x\sigma_x - \frac{1}{M}\sum_{\eleInSet{x}{X}}\sum_{m=1}^M \frac{\lambda_x}{\mu_x}\sigma_{x}\11_{\{C_m = x\}}\Big\|_1. \nonumber
\end{align}
By adding an subtracting appropriate terms within the trace norm of $ S $ and using the triangle inequality we obtain, $ S \leq S_{1} +S_{2} + S_{3},$ where
\begin{align}
S_{1}  &\deq \big\|\sum_{\eleInSet{x}{X}}\lambda_x\sigma_x - \sum_{\eleInSet{x}{X}}\lambda_x\sigTilde\big\|_1,   \nonumber\\
S_{2}  &\deq \big\| \frac{1}{M}\sum_{m=1}^M \frac{\lambda_{C_m}}{\mu_{C_m}}\tilde{\sigma}_{C_m} - \frac{1}{M}\sum_{m=1}^M\frac{\lambda_{C_m}}{\mu_{C_m}} {\sigma}_{C_m}\big\|_1,   \quad \text{and}\nonumber \\
S_{3}  &\deq  \big\|\sum_{\eleInSet{x}{X}}\lambda_x\sigTilde - \frac{1}{M}\sum_{\eleInSet{x}{X}}\sum_{m=1}^M \frac{\lambda_{x}}{\mu_{x}}\sigTilde\11_{\{C_m = x\}}\big\|_1. \nonumber 
\end{align}
We begin by bounding the term corresponding to $ S_{1}$ and $ S_{2} $ as follows:
\begin{align}\label{eq:result_S1}
S_1 &\leq \sum_{\eleInSet{x}{X}}\lambda_x \|\sigma_x - \Pi\Pi_x\sigma_x\Pi_{x}\Pi\|_1 \nonumber \\
 &\leq \sum_{\eleInSet{x}{X}}\lambda_x \|\sigma_x - \Pi\sigma_x\Pi\big\|_1 \nonumber \\
 & \hspace{20pt}+ \sum_{\eleInSet{x}{X}}\lambda_x \|\Pi\sigma_x\Pi - \Pi\Pi_x\sigma_x\Pi_{x}\Pi\|_1  \nonumber \\
 & \leq 2\sqrt{\epsilon} + \sum_{\eleInSet{x}{X}}\lambda_x \|\Pi\|_{\infty}\|\sigma_x - \Pi_x\sigma_x\Pi_{x}\|_1\|\Pi\|_{\infty} \nonumber \\
 &\leq 4\sqrt{\epsilon} = \delta(\epsilon),
\end{align}
where the first two inequalities use the triangle inequality, the third uses the gentle measurement lemma (given the assumption  \eqref{Soft_Covering-constraints1} from the statement of the Lemma) for the first term, and operator Holder's inequality (Exercise 12.2.1 in \cite{book_quantum}) for the second term. The last inequality follows again from the gentle measurement given the assumption \eqref{Soft_Covering-constraints2}.
Similarly, for $ S_{2} $ we have
\begin{align}\label{eq:result_S2}
\EE_{\mathbbm{C}}[S_{2}] &\leq \EE_{\mathbbm{C}}\left[\frac{1}{M}\sum_{m=1}^{M}\sum_{\eleInSet{x}{X}}\frac{\lambda_x}{\mu_x}\mathbbm{1}_{\{C_m = x\}}\|\sigma_x - \tilde{\sigma}_x\|_1\right]\nonumber \\ 
& =  \frac{1}{M}\sum_{m=1}^{M}\sum_{\eleInSet{x}{X}}\lambda_x\|\sigma_x - \tilde{\sigma}_x\|_1 \leq 4\sqrt{\epsilon} = \delta(\epsilon),
\end{align}
where we use the fact that $\EE_{\mathbbm{C}}[\mathbbm{1}_{\{c_{m} = x\}}] = \mu_x$, and the last inequality uses similar arguments as in \eqref{eq:result_S1}.
Finally, we proceed to bound the term corresponding to $ S_3$. 
% The ideas employed here are inspired from \cite{cuff2009thesis}, however improved by incorporating substantial sophistication (or refinement) to be applicable to the quantum (or an operator) setting.
Firstly, note that, $\EE_{\mathbbm{C}}[\frac{1}{M}\sum_{m}\frac{\lambda_{C_m}}{\mu_{C_m}}\tilde{\sigma}_{C_m}] = \sum_{\eleInSet{x}{X}}\lambda_x\tilde{ \sigma_x}$.  This gives
\begin{widetext}
\begin{align}\label{eq:trace_sqrt}
\EE_{\mathbbm{C}}[S_3] &= \EE_{\mathbbm{C}}\left[\Big\|\frac{1}{M}\sum_m\frac{\lambda_{C_m}}{\mu_{C_m}} \tilde{\sigma}_{C_m} - \EE_{\mathbbm{C}}\bigg[\frac{1}{M}\sum_{m}\frac{\lambda_{C_m}}{\mu_{C_m}}\tilde{\sigma}_{C_m}\bigg]\Big\|_1\right] \nonumber \\
& \leq \Tr{\sqrt{\EE_{\mathbbm{C}}\left[\left(\frac{1}{M}\sum_m \frac{\lambda_{C_m}}{\mu_{C_m}}\tilde{\sigma}_{C_m} - \EE_{\mathbbm{C}}\bigg[\frac{1}{M}\sum_{m}\frac{\lambda_{C_m}}{\mu_{C_m}}\tilde{\sigma}_{C_m}\bigg] \right)^2\right]}} \nonumber \\
& = \Tr{\sqrt{\EE_{\mathbbm{C}}\left[\left(\frac{1}{M}\sum_m \frac{\lambda_{C_m}}{\mu_{C_m}}\tilde{\sigma}_{C_m}  \right)^2\right] - \left(\EE_{\mathbbm{C}}\bigg[\frac{1}{M}\sum_{m}\frac{\lambda_{C_m}}{\mu_{C_m}}\tilde{\sigma}_{C_m}\bigg] \right)^2}} \nonumber \\
& = \Tr{\sqrt{\frac{1}{M^2}\sum_m\EE_{\mathbbm{C}}\left[ \left(\frac{\lambda_{C_m}}{\mu_{C_m}}\tilde{\sigma}_{C_m}\right)^2\right] + \frac{1}{M^2}\sum_{\substack{m,m'\\m\neq m'}} \EE_{\mathbbm{C}}\left[ \frac{\lambda_{C_m}\tilde{\sigma}_{C_m}}{\mu_{C_m}} \frac{\lambda_{C_{m'}}\tilde{\sigma}_{C_{m'}}}{\mu_{C_{m'}}}\right] - \left(\frac{1}{M}\sum_{m}\EE_{\mathbbm{C}}\bigg[\frac{\lambda_{C_m}\tilde{\sigma}_{C_m}}{\mu_{C_m}}\bigg] \right)^2}} \nonumber \\
& = \Tr{\sqrt{\frac{1}{M}\EE_{\mathbbm{C}}\left[ \left(\frac{\lambda_{C_1}\tilde{\sigma}_{C_1}}{\mu_{C_1}}\right) ^2\right]  - \frac{1}{M} \left(\EE_{\mathbbm{C}}\left[\frac{\lambda_{C_1}\tilde{\sigma}_{C_1}}{\mu_{C_1}}\right] \right)^2}}\nonumber \\
& \leq \Tr{\sqrt{\frac{1}{M}\EE_{\mathbbm{C}}\left[ \left(\frac{\lambda_{C_1}\tilde{\sigma}_{C_1}}{\mu_{C_1}}\right)^2\right]  }}, 
% = \Tr{\sqrt{\frac{1}{M}\sum_{\eleInSet{x}{X}}\lambda_x \tilde{\sigma}_x^2}}
\end{align}
\end{widetext}
where the first inequality follows from concavity of operator square-root function (L\"{o}wner-Heinz theorem, see Theorem $2.6$ in \cite{carlen2010trace}). The last equality uses the fact that codewords of the random code $ \mathbbm{C} $ are pairwise independent, and the last inequality follows from monotonicity of the operator square-root function (Theorem $2.6$ in \cite{carlen2010trace}). 

\noindent Moving on, we now bound the operator within the square root of \eqref{eq:trace_sqrt} as
\begin{align}
\EE_{\mathbbm{C}}\left[ \left(\frac{\lambda_{C_1}\tilde{\sigma}_{C_1}}{\mu_{C_1}}\right)^2\right] & = 
\sum_{\eleInSet{x}{X}}\frac{\lambda_x^2}{\mu_x}\tilde{\sigma}_x^2  \leq \sum_{\eleInSet{x}{X}}\kappa\lambda_x\tilde{\sigma}_x^2 \nonumber \\
& = \kappa\! \sum_{\eleInSet{x}{X}}\!\lambda_x \Pi \left(\Pi_{x}\sigma_x\Pi_x\right)\Pi \left(\Pi_{x}\sigma_x\Pi_x\right)\Pi, \nonumber
\end{align}
where we use the assumption $\frac{\lambda_x}{\mu_x} \leq \kappa $ for all $\eleInSet{x}{X}$. Further since, $ \Pi \leq I $, we have $ \left(\Pi_{x}\sigma_x\Pi_x\right)\Pi \left(\Pi_{x}\sigma_x\Pi_x\right) \leq \left (\Pi_{x}\sigma_x\Pi_x\right )^2 $, which gives
\begin{align}
\EE_{\mathbbm{C}}\left[ 
\left( \frac{\lambda_{C_1} \tilde{\sigma}_{C_1}}{\mu_{C_1}} \right)^2  \right] & \leq \kappa\sum_{\eleInSet{x}{X}}\lambda_x \Pi \left(\Pi_{x}\sigma_x\Pi_x\right)^2\Pi. \nonumber
\end{align}
Moreover, using the assumption \ref{Soft_Covering-constraints4}, i.e., $ \Pi_{x}\sigma_x\Pi_x \leq \frac{1}{d}\Pi_x \leq \frac{1}{d}I $, we get
\begin{align}
\left(\Pi_{x}\sigma_x\Pi_x\right)^2 & = \sqrt{\Pi_{x}\sigma_x\Pi_x}\left(\Pi_{x}\sigma_x\Pi_x\right)\sqrt{\Pi_{x}\sigma_x\Pi_x}  \nonumber \\ 
& \leq \frac{1}{d}\Pi_{x}\sigma_x\Pi_x, \quad \text{for all } \eleInSet{x}{X}. \nonumber
\end{align}
Thus,
\begin{align}\label{eq:bound_expec_sigma}
\EE_{\mathbbm{C}}\left[ \left( \frac{\lambda_{C_1} \tilde{\sigma}_{C_1}}{\mu_{C_1}}\right)^2 \right] & \leq \frac{\kappa}{d}\Pi \left(\sum_{\eleInSet{x}{X}}\lambda_x\Pi_{x}\sigma_x\Pi_x\right)\Pi \leq \frac{\kappa}{d}\Pi\sigma\Pi,
\end{align}
where the second inequality uses the assumption  \eqref{Soft_Covering-constraints5} from the statement of the Lemma. Substituting the simplification obtained in \eqref{eq:bound_expec_sigma} into \eqref{eq:trace_sqrt} and using the monotonicity of square-root operator, we obtain
\begin{align}\label{eq:result_S3}
\EE_{\mathbbm{C}}[S_3] \leq \Tr{\sqrt{\frac{\kappa}{Md}\Pi\sigma\Pi}} \leq \sqrt{\frac{\kappa D}{Md}},
\end{align}
where the second inequality uses the assumption \eqref{Soft_Covering-constraints3}. Combining the bounds \eqref{eq:result_S1}, \eqref{eq:result_S2}, and \eqref{eq:result_S3} we get the desired result.

\end{proof}

\subsection{Proof of Lemma \ref{lem:newMarkov}}
% \begin{proof}
\label{appx:newMarkov}
% 	We begin by making an following observation. Since $ I_A - P $ is a projector on to the subspace of $ X $ with eigenvalues greater than 1, we have
% 	\begin{align}\label{eq:traceRankBound2}
% 	\Tr{I_A - P} = \mbox{rank}(I_A -P) = \mbox{rank}((I_A -P)X(I_A -P)) \leq \Tr{(I_A -P)X} \leq \Tr{X},
% 	\end{align}
% 	where the first inequality follows from the definition of $ P $ and the next follows from the fact that $ X \geq 0. $
	Note that if $ P $ prunes $ X $, then $ P $ also prunes $ \frac{1}{\eta}(X - (1-\eta)I_A) $ with respect to $ I_A. $ Using Lemma \ref{lem:newMarkov0}, we obtain
	\begin{align}
	\Tr{I_A - P} \leq \frac{1}{\eta}\Tr{X - (1-\eta) I_A}. \nonumber
	\end{align}
	Applying expectation and using the assumption on $\EE[X]$, 
% 	and using the fact that $\frac{1}{(1+\eta)} \leq (1-\eta^2)$. for $ \eta \in (0,\frac{1}{2}), $ 
we get
	\begin{align}
	\EE[\Tr{\!I_A\! -\! P}] 
% 	& \leq \frac{1}{\eta^2}\EE\left[\Tr{X - \frac{I_A}{1+\eta}}\right] \nonumber \\
	 \!&\leq \!\frac{1}{\eta}\EE\left[\Tr{X\! - \!\EE[X]}\right]  \leq \frac{1}{\eta}\EE\left[\|X - \EE[X]\|_1\right].\nonumber
	\end{align}
% 	where the first inequality uses the assumption that $ \EE[X]  \leq (1-\eta^2)I_A$ and the last one is based on the inequality $ \Tr{A} \leq \norm{A}_1 $ for any operator $ A \in \mathcal{B}(\mathcal{H}_A),$ which completes the proof.
% \end{proof}
% \ref{lem:mutual covering}, \ref{lem:tensor and trace dist for POVM}, \ref{lem:omega lambda}, \ref{lem:nf_SoftCovering} and \ref{lem:Separate}}
%-------------------------------------%

% \addtocontents{toc}{\protect\setcounter{tocdepth}{1}}

%%----------------------------xxxxxxxxxxxxxxxxxxxxxxx-----------------------------

\subsection{Proof of Lemma \ref{lem:nf_SoftCovering}}
\label{appx:nf_SoftCovering_Proof}
% \noindent We begin by defining $ \theta_{w^n} $ as $\theta_{w^n} = \sum_{w^n}\rhohatwA \tensor \phi_{w^n}. $Note that the definition of $ \theta_{w^n} $ contains all the operators in product form, and thus it can be written as $ \theta_{w^n} = \bigtensor_{i=1}^{n}\theta_{w_i}. $
% This simplifies $ S_{11} $ as
We begin by defining $L$ as 
\begin{align}
L &\deq \bigg\|\sum_{w^n}\lambdawA\theta_{w^n} - \nonumber \\ 
& \frac{1}{(1+\eta)}\frac{p^n}{p^{k+l}N'}\sum_{\mu=1}^{N'}\sum_{a,m}\sum_{w^n} \lambda_{w^n} \theta_{w^n}\mathbbm{1}_{\{W^{n,(\mu)}(a,m) = w^n\}} \bigg\|_1. \nonumber%\label{eq:J1} 
\end{align}
Further, let $ \theta \deq \sum_{w \in \mathcal{W}}\lambda_w\theta_w$ and let  $ \Pi_{\theta} $ and $ \PiTauWn $ denote the $\delta$-typical projector of $ \theta $ and conditional typical projector of $ \theta_{w^n}, $ respectively. 
% Using these projectors, define $ \widetilde{\theta}_{w^n} $ as $ \widetilde{\theta}_{w^n} =  \PiTauWn\theta_{w^n}\PiTauWn $ for $w^n \in \TDeltan(W)$ and $0$ (null operator) otherwise.
% Let $ L $ denote the term corresponding to $S_{11}$ as
%  \begin{align}
% L \deq S_{11} = \left\|\sum_{w^n}\lambda_{w^n}^A\theta_{w^n} - \frac{1}{(1+\eta)}\frac{|\mathcal{W}^n|}{2^{n(S+C)}}\sum_{\mu}\sum_{a,m}\sum_{w^n}
% \lambda_{w^n}^A\theta_{w^n} \mathbbm{1}_{\{W^{n,(\mu)}(a,m) = w^n\}}\right\|_1 \nonumber
% \end{align}
Define $\Tilde{\lambda}_{w^n} = \frac{\lambda_{w^n}}{1-\varepsilon}$ for $w^n \in \TDeltan(W)$, and $0$ otherwise, where $\varepsilon = \sum_{w^n \notin \TDeltan(W)}\lambda_{w^n}.$
Using the triangle inequality we can bound $ L $ as $ L \leq L_1 + L_2+L_3, $ where
\begin{align}
L_1 & \deq \bigg\|\sum_{w^n}\lambda_{w^n}\theta_{w^n} - \sum_{w^n}\tilde{\lambda}_{w^n}{ \theta}_{w^n}\bigg\|_1, \nonumber \\
L_2 & \deq \bigg\|\sum_{w^n }\tilde{\lambda}_{w^n}{ \theta}_{w^n} - \nonumber \\
& \hspace{22pt}\frac{p^n}{p^{k+l}N'}\sum_{\mu=1}^{N'}\sum_{a,m}\sum_{w^n}\tilde{\lambda}_{w^n}{ \theta}_{w^n} \mathbbm{1}_{\{W^{n,(\mu)}(a,m) = w^n\}}\bigg\|_1, \nonumber \\
 L_3  & \deq \bigg\| \frac{p^n}{p^{k+l}N'} \sum_{\mu}\sum_{a,m}\sum_{w^n}\left(\tilde{\lambda}_{w^n}-{\frac{\lambda_{w^n}}{(1+\eta)}}\right)\nonumber \\
 &\hspace{1.5in}\times\theta_{w^n} \mathbbm{1}_{\{W^{n,(\mu)}(a,m) = w^n\}}\bigg\|_1. \nonumber 
% \\& \leq 2^{-n(S(W)_{\sigma_{\theta}})}|\mathcal{W}^n|\left\|\sum_{w^n}\frac{1}{|\mathcal{W}^n|}\widetilde{ \theta}_{w^n} - \frac{1}{(1+\eta)}\frac{|\mathcal{W}^n|}{2^{n(S+C)}}\sum_{\mu}\sum_{a,m}\sum_{w^n \in \TDeltaN(W)}\widetilde{\theta}_{W^{n,(\mu)}(a,m)}\right\|_1, \nonumber
\end{align}
% where, in the equality above we have bounded $ \lambda_{w^n}^A \leq 2^{-n(S(W)_{\sigma_{\theta}} -\delta_w)} $ for $ w^n \in \TDeltaN(W) $ and $ \sigma_{\theta}^{RW} $ defined as $ \sigma_{\theta}^{RW} =  \sum_{w \in \mathcal{W}}\lambda_{w}^{A} \theta_w^R \tensor \ketbra{w}^W$ for some orthogonal basis $ \{\ketbra{w}\}_{\eleInSet{w}{W}} $.
\noindent We begin by bounding the term corresponding to $ L_1 $ as 
\begin{align}
    L_1 & \leq \sum_{w^n \in \TDeltan(W)}\!\!\!\!\lambda_{w^n}\frac{\varepsilon}{1-\varepsilon}\underbrace{\left\|\theta_{w^n}\right\|_1}_{=1} + \sum_{w^n\notin \TDeltan(W)}\!\!\!\!\lambda_{w^n}\underbrace{\left\|\theta_{w^n}\right\|_1}_{=1} \nonumber \\
    & = 2\varepsilon. \label{eq:L1Term}
    % \leq \epsilon_{L_1}' + \sum_{w^n\notin \TDeltan(W)}\!\!\!\!\lambda_{w^n} = \epsilon_{L_1}, 
\end{align}
% where the second inequality follows from the Average Gentle Measurement Lemma \cite{Wilde_book} by setting $\epsilon_{L_1}' \deq {2\sqrt{\varepsilon'}}$, where $\varepsilon' \deq \varepsilon + 2\sqrt{\varepsilon}$ and $ \varepsilon \deq \sum_{w^n\notin \TDeltan(W)}\lambda_{w^n} $, and the equality above follows by defining $\epsilon_{L_1}$ as $\epsilon_{L_1} \deq \epsilon_{L_1}' + \varepsilon.$ As a result, we have $L_1 \leq \epsilon_{L_1}$.

Now consider the term corresponding to $ L_2$, for which we employ Lemma \ref{lem:Change Measure Soft Covering Variance Based}. Toward this, we consider the following identification: $\lambda_x$  with $\LambdaTildew$, $\sigma_x$ with  $\theta_{w^n}$, $ \mathcal{X}$ with $\TDeltan(W)$, $\Bar{\mathcal{X}}$ with $\FF_p^n$, $\sigma$ with $\widetilde{\theta} \deq \sum_{w^n}\LambdaTildew{ \theta}_{w^n} $, $\Pi$ with $\Pi_{\theta}$, $\Pi_x$ with $\PiTauWn$, and $\mu_x = \frac{1}{p^n}$ for all $x\in \bar{\mathcal{X}}$.  Since the collection of random variables $ \{W^{n,(\mu)}(a,m)\} $ are generated using Unionized Coset Codes, we have  
\begin{align}
\PP\left(\mathbbm{1}_{\{W^{n,(\mu)}(a,m) = w^n\}}=1\right) = \frac{1}{p^n}, \quad \text{for all}\quad w^n \in \FF_p^n. \nonumber
\end{align}
Note that $\frac{\LambdaTildew}{1/p^n} \leq 2^{-n{(S(W)_{\sigma_\theta}-\log{p}-\delta_w)}}$ for all $w^n \in \FF_p^n$, where $\delta_{w}(\delta) \searrow 0$ as $\delta \searrow 0$, and $\sigma_\theta$ is defined in the statement of the lemma.
% \footnote{Note that the $\sigma$ in $S(W)_{\sigma} $ is the classical quantum state defined in the statement of Theorem \ref{thm:p2pTheorem}.}.
With these, we check the hypotheses of Lemma \ref{lem:Change Measure Soft Covering Variance Based}. Firstly, using the pinching arguments described in \cite[Property 15.2.7]{Wilde_book}, we have $\Tr{\Pi_{\theta}\theta_{w^n}} \geq 1-\epsilon$ for all $\epsilon \in(0,1), \delta>0$ and sufficiently large $n$, satisfying hypothesis \eqref{Soft_Covering-constraints1}. Secondly, \eqref{Soft_Covering-constraints2} and \eqref{Soft_Covering-constraints5} are satisfied from the construction of $\PiTauWn$. Next, we consider the hypothesis \eqref{Soft_Covering-constraints3}. We have
\begin{align}
    \left\|\Pi^{\theta}\sqrt{\Tilde{\theta}}\right\|_1 \!\!
    &= \Tr{\sqrt{\Pi^{\theta}\tilde{\theta}\Pi^{\theta}}} \nonumber \\ 
    &\leq \frac{1}{\sqrt{(1-\varepsilon)}}\Tr{\sqrt{\Pi^{\theta}\theta^{\tensor n}\Pi^{\theta}}} \leq 2^{\frac{n}{2}(S(R)_{\sigma_{\theta}} +\delta_w')}, \nonumber
\end{align}
where the first inequality above follows from the fact that $ \sum_{w^n}\LambdaTildew\theta_{w^n} \leq \frac{1}{(1-\varepsilon)}\sum_{w^n}\lambda_{w^n}\theta_{w^n} = \frac{\theta^{\tensor n}}{(1-\varepsilon)}$ and using the operator monotonicty of the  square-root function (Theorem $2.6$ in \cite{carlen2010trace}). The second inequality follows from the property of the typical projector for some $\delta_w'$ such that $\delta_w' \searrow 0$ as $\delta \searrow 0$. This gives \[{D} = {2^{n(S(R)_{\sigma_{\theta}} + \delta_{w}')}}.\] Finally, the hypotheses \eqref{Soft_Covering-constraints4} is satisfied from the property of conditional typical projectors  for $ d = 2^{n(S(R|W)_{\sigma_{\theta}} - \delta_w'')}$,
where $\delta_w'' \searrow 0$ as $\delta \searrow 0 $ (see \cite[Property 15.2.6]{Wilde_book}).
%  Using the projectors $\Pi_{\theta}$, and $\PiTauWn$, one can note that the hypotheses in \eqref{Soft_Covering-constraints} are satisfied for some $\epsilon_{\theta}(\delta), D$ and $d$ such that $\epsilon_{\theta}(\delta) \searrow 0$ as $\delta \searrow 0$ and $D = ,$ and  Here we define
Next we check the pairwise independence of $\WCodewordm$ and $W^{n,(\mu)}(\tilde{a},\tilde{m})$. Since these are constructed using randomly and uniformly generated $G^{}$ and $h^{(\mu)}$, we have $\{\WCodewordm\}_{a\in \FF^k_p,m\in\FF_p^l,\mu\in[1:N']}$ to be pairwise independent for each (see \cite{pradhanalgebraic} for details).
% We simplify $L_2$ using the triangle inequality as
% \begin{align}
%     L_2 & \leq \frac{1}{N}\sum_{\mu} \left\|\sum_{w^n }\tilde{\lambda}_{w^n}{ \theta}_{w^n} - \frac{p^n}{2^{nS}}\sum_{a,m}\sum_{w^n}\tilde{\lambda}_{w^n}{ \theta}_{w^n} \mathbbm{1}_{\{W^{n,(\mu)}(a,m) = w^n\}}\right\|_1. \nonumber 
% \end{align}
 Therefore, employing Lemma \ref{lem:Change Measure Soft Covering Variance Based} we get
%  \begin{widetext}
\begin{align}
    \EE[L_2] & \leq  \sqrt{\frac{2^{n(S(R)_{\sigma_{\theta}} + \delta_{w}')}2^{-n{(S(W)_{\sigma_\theta}-\log{p}-\delta_w)}}}{N'p^{k+l}2^{n(S(R|W)_{\sigma_{\theta}} - \delta_w'')}}} + 8\sqrt{\epsilon}\nonumber \\ & \leq \text{exp}_2\bigg[-\frac{n}{2}\bigg(\frac{k+l}{n}\log{p} + \frac{1}{n}\log{N'} -  I(R;W)_{\sigma_{\theta}}  \nonumber \\
& \hspace{30pt}- \log{p} +S(W)_{\sigma_{\theta}} - \delta_{w} - \delta_{w}' -\delta_{w}'' \bigg)\bigg] + 8\sqrt{\epsilon}, \label{eq:L2Term}
\end{align} 
%  \end{widetext}
where exp$_2(x)\deq 2^x.$
\noindent As for $ L_3 $, taking expectation and using $ \EE[\mathbbm{1}_{\{W^{n,(\mu)}(a,m) = w^n\}}] = \frac{1}{p^n} $ gives 
\begin{align}\label{eq:L3Term}
    \EE[L_3] \leq \frac{\eta+\varepsilon}{ (1+\eta)} + \frac{\varepsilon}{(1+\eta)} =\frac{\eta+2 \varepsilon}{1+\eta}.
\end{align}

% , applying triangle inequality yields
% % \begin{align}
% % L_3 & \leq  \frac{1}{(1+\eta)}\frac{p^n}{2^{n(S+C)}} \sum_{\mu}\sum_{a,m}\sum_{w^n}\lambda_{w^n}^A\mathbbm{1}_{\{W^{n,(\mu)}(a,m) = w^n\}}\left\|\theta_{w^n} - \widetilde{ \theta}_{w^n} \right\|_1. \nonumber 
% % \end{align}
% Taking expectation and using $ \EE[\mathbbm{1}_{\{W^{n,(\mu)}(a,m) = w^n\}}] = \frac{1}{p^n} $ gives 
% \begin{align}
% \EE[L_3] \leq \sum_{w^n}\frac{\lambda_{w^n}}{(1+\eta)}\left\|\theta_{w^n} - \widetilde{ \theta}_{w^n} \right\|_1 \leq \sum_{w^n \in \TDeltaN(W)}\!\!\frac{\lambda_{w^n}}{(1+\eta)}\left\|\theta_{w^n} - \widetilde{ \theta}_{w^n} \right\|_1 + \sum_{w^n\notin \TDeltaN(W)}\!\!\frac{\lambda_{w^n}}{(1+\eta)}  \leq \frac{2\epsilon_{L_1}}{(1+\eta)} 
% \end{align}
\noindent Combining the bounds from $\eqref{eq:L1Term}, \eqref{eq:L2Term}$ and $\eqref{eq:L3Term}$ gives the desired result.

\subsection{Proof of Lemma \ref{lem:LemAandABar}} \label{appx:ProofofLemAandABar}
We begin by using the H\'older's inequality \cite{book_quantum,carlen2010trace} for operator norm, i.e., ($ \|AB\|_1 \leq \|A\|_{\infty}\|B\|_1 $), and defining $\hat{\Lambda}_{w^n} = \sqrt{\rho^{\tensor n}}^{-1}\rhotildwA\sqrt{\rho^{\tensor n}}^{-1}  $. This gives us
\begin{widetext}
\begin{align}
\sum_{w^n}\!\gammaWCoeff\left\|\sqrt{\rho^{\tensor n}} \! \left (\!\bar{A}_{w^n}^{(\mu)} -	A_{w^n}^{(\mu)} \!\right)\!\sqrt{\rho^{\tensor n}}\right\|_1 = &  \sum_{w^n}\alpha_{w^n}\gammaWCoeff\left\|\Pi_\rho \sqrt{\rho^{\tensor n}} \hat{\Lambda}_{w^n}\sqrt{\rho^{\tensor n}}\Pi_\rho - \Pi_\rho \sqrt{\rho^{\tensor n}}\CutOff \hat{\Lambda}_{w^n}  \CutOff\sqrt{\rho^{\tensor n}}\Pi_\rho\right\|_1 \nonumber \\ 
\leq &  \sum_{w^n}\alpha_{w^n}\gammaWCoeff\left\|\Pi_\rho \sqrt{\rho^{\tensor n}}\right\|_{\infty}^2\left\|\hat{\Lambda}_{w^n} - \CutOff \hat{\Lambda}_{w^n} \CutOff\right\|_1 \nonumber \\ 
\leq & \;2^{-n(S(\rho)-\delta_\rho)} \sum_{w^n}\alpha_{w^n}\gammaWCoeff 2\sqrt{\Tr{(\PiA-\CutOff)\hat{\Lambda}_{w^n}}\Tr{\hat{\Lambda}_{w^n}}}, \nonumber
\end{align}
\end{widetext}
where the equality follows from the fact that $\Pi_{\rho}$ and $\CutOff$ commute, the first inequality follows from the H\'older's inequality, and the second inequality uses the following bounds
\begin{align}
\bigg\|&\hat{\Lambda}_{w^n} - \CutOff \hat{\Lambda}_{w^n}  \CutOff\bigg\|_1 \nonumber \\
& \leq \left\|\hat{\Lambda}_{w^n} - \CutOff \hat{\Lambda}_{w^n} \right\|_1 + \left\|\CutOff \hat{\Lambda}_{w^n} - \CutOff \hat{\Lambda}_{w^n}  \CutOff\right\|_1 \nonumber \\ 
& = \Tr\left\{\left|(\Pi_{\rho} - \CutOff) \sqrt{\hat{\Lambda}_{w^n}} \sqrt{\hat{\Lambda}_{w^n}}\right|\right\}  \nonumber \\
& \hspace{1in}+ \Tr\left\{\left| \CutOff\sqrt{\hat{\Lambda}_{w^n}}\sqrt{\hat{\Lambda}_{w^n}}(\Pi_{\rho} - \CutOff)\right|\right\} \nonumber \\
& \leq \sqrt{\Tr\left\{(\Pi_{\rho} - \CutOff)^2 {\hat{\Lambda}_{w^n}}\right\}\Tr\left\{ {\hat{\Lambda}_{w^n}}\right\}} \nonumber \\
& \hspace{0.85in} + \sqrt{\Tr\left\{ \CutOff{\hat{\Lambda}_{w^n}}\right\}\Tr\left\{{\hat{\Lambda}_{w^n}}(\Pi_{\rho} - \CutOff)^2\right\}} \nonumber \\
& \leq 2\sqrt{\Tr{(\Pi_{\rho}-\CutOff )\hat{\Lambda}_{w^n}}\Tr{\hat{\Lambda}_{w^n}}}, \nonumber
\end{align}
where the second inequality uses Cauchy-Schwarz inequality along with the polar decomposition (see the usage in  \cite[Lemma 9.4.2]{book_quantum}) and the last inequality uses the arguments: 
(i) $\CutOff$ is a projector onto a subspace of $\Pi_\rho$ and (ii) $ \Tr{\CutOff\hat{\Lambda}_{w^n}} \leq \Tr{\hat{\Lambda}_{w^n}} $.
Further, using the  fact that for $ w^n \in \TDeltaN(W),$  
\begin{align*}
    \Tr\{\hat{\Lambda}_{w^n}\} & = \|\Pi_{\rho}\hat{\Lambda}_{w^n}\Pi_{\rho}\|_1\\
    & \leq \|\Pi_{\rho}\sqrt{\rho^{\tensor n}}^{-1}\|_{\infty}\underbrace{\|\rhotildwA\|_1}_{\leq 1}\|\Pi_{\rho}\sqrt{\rho^{\tensor n}}^{-1}\|_{\infty} \\
    & \leq  \|\Pi_{\rho}\sqrt{\rho^{\tensor n}}^{-1}\|_{\infty}^2 \leq 2^{n(S(\rho)+\delta_\rho)},
\end{align*}
it follows that
% \begin{widetext}
\begin{align}
&\sum_{w^n}\gammaWCoeff\left\|\sqrt{\rho^{\tensor n}}\!  \left (\bar{A}_{w^n}^{(\mu)} -	A_{w^n}^{(\mu)} \right )\!\sqrt{\rho^{\tensor n}}\right\|_1 \nonumber \\ 
& \leq 2\cdot  2^{-\frac{n}{2}(S(\rho)-4\delta_\rho)}
\sum_{w^n}{\alpha_{w^n}\gammaWCoeff} \sqrt{\Tr{(\Pi_\rho-\CutOff )\hat{\Lambda}_{w^n}}} \nonumber \\
%& \leq \frac{2}{N} {2^{3n\delta}}\sum_{\mu=1}^{N} \sum_{w^n}\frac{ c^2\gammaWCoeff}{2^{nS}}\sqrt{\Tr{(\PiA-\CutOff ) \rhotildwA}} \nonumber \\ 
& \leq 2 \cdot {2^{3n\delta_\rho}}  \Delta^{(\mu)}
\sqrt{\sum_{w^n}\frac{\alpha_{w^n} \gammaWCoeff}{\Delta^{(\mu)}} \Tr{(\Pi_\rho-\CutOff )\rhotildwA}} \nonumber \\ 
& = 2\cdot {2^{3n\delta_\rho}} \left(\Delta^{(\mu)} -\EE[\Delta^{(\mu)}] + \EE[\Delta^{(\mu)}]\right) \nonumber \\
&\hspace{0.9in}\times \sqrt{\Tr{(\Pi_\rho-\CutOff )\sum_{w^n}\frac{\alpha_{w^n} \gammaWCoeff}{\Delta^{(\mu)}} \rhotildwA}}  \nonumber \\
& \leq 2 \cdot {2^{3n\delta_\rho}} \EE[\Delta^{(\mu)}]
 \sqrt{\Tr{(\Pi_\rho-\CutOff )\sum_{w^n}\frac{\alpha_{w^n} \gammaWCoeff}{\Delta^{(\mu)}} \rhotildwA}}   \nonumber \\
&\hspace{1.7in} + 2\cdot2^{3n\delta_\rho}\underbrace{\left|\Delta^{(\mu)} - \EE[\Delta^{(\mu)}]\right|}_{H_0} \nonumber \\   
% & = \frac{2}{N} {2^{3n\delta}}\sum_{\mu=1}^{N} \sqrt{\Tr{(\PiA-\CutOff )\sum_{w^n}\left(\frac{c^2 \gammaWCoeff}{2^{nS}} \rhotildwA - \lambdawA \rhotildwA + \lambdawA \rhotildwA \right)}} \nonumber \\  
& \leq 2 \cdot{2^{3n\delta_\rho}}
 \left(H_0 + \frac{\sqrt{(1-\varepsilon)}}{(1+\eta)}\sqrt{H_1+H_2 + H_3}\right) , \nonumber  
\end{align}
% \end{widetext}
% \begin{align}
% \sum_{w^n}\gammaWCoeff\!\left\|\sqrt{\rho_{A}^{\tensor n}}\!  \left (\bar{A}_{w^n}^{(\mu)} -	A_{w^n}^{(\mu)} \right )\!\sqrt{\rho_{A}^{\tensor n}}\right\|_1 & \leq \frac{2|\mathcal{W}^n|}{N_1} 2^{-\frac{n}{2}(S(\rho)-4\delta)}{2^{-n(S(W)-\delta_u)}}
% \sum_{\mu=1}^{N}\sum_{w^n}\frac{\gammaWCoeff}{2^{nS}} \sqrt{\Tr{(\PiA-\CutOff )\hat{\Lambda}_{w^n}}} \nonumber \\
% %& \leq \frac{2}{N} {2^{3n\delta}}\sum_{\mu=1}^{N} \sum_{w^n}\frac{ c^2\gammaWCoeff}{2^{nS}}\sqrt{\Tr{(\PiA-\CutOff ) \rhotildwA}} \nonumber \\ 
% & \leq \frac{2}{N} {2^{3n\delta}}
% \sum_{\mu=1}^{N} \sqrt{\Tr{(\PiA-\CutOff )\sum_{w^n}\frac{c^2 \gammaWCoeff}{2^{nS}} \rhotildwA}} \nonumber \\   
% & = \frac{2}{N} {2^{3n\delta}}
% \sum_{\mu=1}^{N} \sqrt{\Tr{(\PiA-\CutOff )\sum_{w^n}\left(\frac{c^2 \gammaWCoeff}{2^{nS}} \rhotildwA - \lambdawA \rhotildwA + \lambdawA \rhotildwA \right)}} \nonumber \\  
% & \leq \frac{2}{N} {2^{3n\delta}}
% \sum_{\mu=1}^{N} \sqrt{H_1+H_2}, \nonumber  
% \end{align}
where the second inequality above follows by defining $ \Delta^{(\mu)} = \sum_{w^n\in\TDeltaN(W)}{\alpha_{w^n}\gammaWCoeff}$ and using the concavity of the square-root function, the third inequality follows by  using the fact that 
\begin{align}
\sum_{w^n}\frac{\alpha_{w^n} \gammaWCoeff}{\Delta^{(\mu)}} & \Tr{(\Pi_\rho-\CutOff )\rhotildwA} \nonumber \\
 & \leq \sum_{w^n}\frac{\alpha_{w^n} \gammaWCoeff}{\Delta^{(\mu)}}\Tr{\rhotildwA}\leq 1, 
\end{align}
and defining $H_0$ as above. and the last one follows by first using $\EE[\Delta^{(\mu)}]=\frac{(1-\varepsilon)}{(1+\eta)}$ and then defining $ H_1, H_2 $  and $ H_3 $ as in the statement of the lemma and using the inequality $ \Tr{\Lambda(\omega-\sigma)} \leq  \|\Lambda(\omega-\sigma)\|_1 \leq \|\Lambda\|_\infty\|\omega-\sigma\|_1 $. This completes the proof.

\section{Proof of Propositions}

\subsection{Proof of Proposition \ref{prop:Lemma for p2p:Tilde_S1}}\label{appx:proof of p2p:Tilde_S1}
Applying the triangle inequality on $\widetilde{S}_1$ gives $ \widetilde{S}_1 \leq  \widetilde{S}_{11} +  \widetilde{S}_{12}$, where
\begin{align}
    \widetilde{S}_{11} & \deq {\frac{1}{N}\sum_{\mu}\left\| \sum_{w^n}\lambdawA\rhohatwA - \sum_{w^n}\alpha_{w^n}\gammaWCoeff\hat{\rho}_{w^n} \right\|_1}, \nonumber\\   \widetilde{S}_{12} &\deq { \frac{1}{N}\sum_{\mu}\sum_{w^n}\alpha_{w^n}\gammaWCoeff\left\| \rhohatwA -  \rhotildwA \right\|_1}. \nonumber  \label{eq:p2pS2simplification}
\end{align}
% To provide a bound on $\widetilde{S}$, we begin by bounding each of the $ \widetilde{S}_1,  \widetilde{S}_2$ and $ \widetilde{S}_3$. 
For the first term $\widetilde{S}_{11}$, we use Lemma \ref{lem:nf_SoftCovering}, and identify $\theta_{w^n}$ with $\rhohatwA$ and $N'= 1$. Using this lemma, we obtain the following: For any 
$\epsilon>0$, and any $\eta,\delta \in (0,1)$ sufficiently small and any $n$ sufficiently large,  $\EE[\widetilde{S}_{11}] \leq \epsilon$, if the $ \frac{k+l}{n} \log p > I(W;R)_{\sigma} - S(W)_{\sigma} + \log{p}$, where $ \sigma$ is defined in the statement of the theorem. 
%Note that the  requirements we obtain on $ S_1$ and $ S_2$ here were already imposed earlier in Section \ref{subsec:stochastic_dist_POVM}.
% by Lemma \ref{lem:M_XM_Y POVM}.
As for the second term $\widetilde{S}_{12}$, we use the gentle measurement lemma and bound its expected value as
\begin{align}
\EE&\left[\frac{1}{N}\sum_{\mu}\sum_{w^n}\alpha_{w^n}\gammaWCoeff \left\| \rhohatwA -  \rhotildwA \right\|_1\right]  \nonumber\\
&\leq \hspace{-17pt}\sum_{w^n \in \TDeltaN(W)}\hspace{-6pt}\frac{\lambdawA}{(1+\eta)}\left\| \rhohatwA -  \rhotildwA \right\|_1 + \hspace{-17pt} \sum_{w^n \notin \TDeltaN(W)}\hspace{-6pt}\frac{\lambdawA}{(1+\eta)} \leq \epsilon_{\scriptscriptstyle \widetilde{S}_{12}},\nonumber
\end{align}
where the inequality is based on the repeated usage of the average gentle measurement lemma by setting $ \epsilon_{\scriptscriptstyle\widetilde{S}_{12}} = \frac{(1-\varepsilon)}{(1+\eta)} (2\sqrt{\varepsilon'} + 2\sqrt{\varepsilon''}) $ with  $ \epsilon_{\scriptscriptstyle \widetilde{S}_{12}} \searrow 0 $ as $n \rightarrow \infty$ and {$ \varepsilon' = \varepsilon'_p + 2\sqrt{\varepsilon'_p} $ and $ \varepsilon'' = 2\varepsilon'_p + 2\sqrt{\varepsilon'_p} $ for $\varepsilon'_p \deq 1-\min\left\{\tr{\Pi_{\rho}\hat{\rho}_{w^n}}, \tr{\Pi_{w^n}\hat{\rho}_{w^n}},1-\varepsilon \right\}$}
% where the second inequality is based on the repeated usage of the Average Gentle Measurement Lemma \cite{Wilde_book} and $ \epsilon_{\scriptscriptstyle\widetilde{S}_{12}} \searrow 0$ as $\delta\searrow0$
% \deq \frac{1}{(1+\eta)} (2\sqrt{\varepsilon'} + 2\sqrt{\varepsilon''}+\varepsilon) $ for  $ \epsilon_{\scriptscriptstyle \widetilde{S}_{12}} \searrow 0 $ as $ \delta \searrow 0 $ and $ \varepsilon' = \varepsilon + 2\sqrt{\varepsilon} $ and $ \varepsilon'' = 2\varepsilon' + 2\sqrt{\varepsilon'} $ and $\varepsilon = \sum_{w^n \notin\TDeltan(W)}\lambdawA$ 
(see (35) in \cite{wilde_e} for more details).

% Finally, for the last term, using an exactly symmetric analysis as in $ S_{13} $ term, we can make $ \EE[\widetilde{S}_{22}] $ arbitrarily small. Hence, we postpone the proof for this term to proposition \ref{prop:Lemma for p2p:S_13}, but use its result which gives us the following. For sufficiently large $ n, $ if  $ S \geq I(W;R)_{\sigma} - S(W)_{\sigma} + \log{p}, $ then $\EE[\widetilde{S}_3] \leq \epsilon_{S_{13}}$, where $\epsilon_{S_{13}} \searrow 0$ as  $\delta \searrow 0.$

\subsection{Proof of Proposition \ref{prop:Lemma for p2p:S_13}}\label{appx:proof of p2p:TildeS_2}
To provide a bound for $\widetilde{S}_2$, we individually bound the terms corresponding to $H_0$ and $\widetilde{H}$ in an expected sense. Let us first consider $\widetilde{H}$. To provide a bound for $\tilde{H}$ we use Lemma \ref{lem:Change Measure Soft Covering Variance Based} with the following identification: $\lambda_x$ with $\frac{\lambdawA}{(1-\varepsilon)}$, $\sigma_x$ with $\rhohatwA$, $\mathcal{X}$ with $\TDeltan(W)$, $\mathcal{\bar{X}}$ with $\FF_p^n$, $\Pi$ with $\Pi_\rho$, $\Pi_x$ with $\Pi_{w^n}$, and  $\mu_x$ with $\frac{1}{p^n}$.
% , and $\sigma$ with $ \sum_{w^n\in \TDeltan(W)}\lambdawA\rhohatwA$.

\noindent Firstly, we have $\frac{\lambdawA}{1/p^n} \leq 2^{-n{(S(W)_{\sigma}-\log{p}-\delta_w)}}$ for all $w^n \in \FF_p^n$, where $\delta_{w}(\delta) \searrow 0$ as $\delta \searrow 0$, which gives $$\kappa = 2^{-n{(S(W)_{\sigma}-\log{p}-\delta_w)}}.$$
With these, we check the hypotheses of Lemma \ref{lem:Change Measure Soft Covering Variance Based}. As for the first hypothesis \eqref{Soft_Covering-constraints1}, using the pinching arguments described in  \cite[Property 15.2.7]{Wilde_book}, we have $\Tr{\Pi_{\rho}\rhohatwA} \geq 1-\epsilon$ for all $\epsilon \in(0,1), \delta>0$ and sufficiently large $n$. Then the hypotheses \eqref{Soft_Covering-constraints2} and \eqref{Soft_Covering-constraints5} are satisfied from the construction of $\Pi_{w^n}$. Next, consider the hypothesis \eqref{Soft_Covering-constraints3}. We have
\begin{align}
    & \Bigg\|\Pi_\rho\sqrt{\bigg(\sum_{w^n\in\TDeltan(W)}\frac{\lambdawA}{(1-\varepsilon)}\rhohatwA\bigg)}\Bigg\|_1 \nonumber \\
    & \leq \frac{1}{\sqrt{1-\varepsilon}}\Tr{\sqrt{\Pi_\rho\rho^{\tensor n}\Pi_{\rho}}} \leq 2^{\frac{n}{2}(S(R)_{\sigma_{}} +\delta_\rho')}, \nonumber
\end{align}
where the first inequality above follows from using $ \sum_{w^n\in\TDeltan(W)}\frac{\lambdawA}{(1-\varepsilon)}\rhohatwA \leq\frac{1}{(1-\varepsilon)}\rho^{\tensor n}$ and the operator monotonicty of the  square-root function. The second inequality follows from the property of the typical projector for some $\delta_\rho'$ such that $\delta_w' \searrow 0$ as $\delta \searrow 0$. This gives \[{D} = {2^{n(S(R)_{\sigma} + \delta_{\rho}')}},\] where $\sigma$ is as defined in the statement of the theorem. Finally, the hypotheses \eqref{Soft_Covering-constraints4} is satisfied from the property of conditional typical projectors for $ d = 2^{n(S(R|W)_{\sigma_{}} - \delta_w'')}$,
where $\delta_w'' \searrow 0$ as $\delta \searrow 0.$
%  Using the projectors $\Pi_{\theta}$, and $\PiTauWn$, one can note that the hypotheses in \eqref{Soft_Covering-constraints} are satisfied for some $\epsilon_{\theta}(\delta), D$ and $d$ such that $\epsilon_{\theta}(\delta) \searrow 0$ as $\delta \searrow 0$ and $D = ,$ and  Here we define
Next we check the pairwise independence of $\WCodewordm$ and $W^{n,(\mu)}(\tilde{a},\tilde{m})$. Since these are constructed using randomly and uniformly generated $G^{}$ and $h^{(\mu)}$, we have $\{\WCodewordm\}_{a\in \FF^k_p,m\in\FF_p^l,\mu\in[1,N]}$ to be pairwise independent  (see \cite{pradhanalgebraic} for details).
% We simplify $L_2$ using the triangle inequality as
% \begin{align}
%     L_2 & \leq \frac{1}{N}\sum_{\mu} \left\|\sum_{w^n }\tilde{\lambda}_{w^n}{ \theta}_{w^n} - \frac{p^n}{2^{nS}}\sum_{a,m}\sum_{w^n}\tilde{\lambda}_{w^n}{ \theta}_{w^n} \mathbbm{1}_{\{W^{n,(\mu)}(a,m) = w^n\}}\right\|_1. \nonumber 
% \end{align}
 Therefore, employing inequality \eqref{lem:ChangeMeasure:ineq2} of Lemma \ref{lem:Change Measure Soft Covering Variance Based}, we get
\begin{align}
   & \EE[\tilde{H}]  \leq  \sqrt{\frac{2^{n(S(R)_{\sigma_{}} + \delta_{\rho}')}2^{-n{(S(W)_{\sigma}-\log{p}-\delta_w)}}}{N2^{nS}2^{n(S(R|W)_{\sigma_{}} - \delta_w'')}}} \nonumber \\
    &\;\leq 2^{-\frac{n}{2}\left(\frac{k+l}{n}\log{p} + \frac{1}{n}\log{N} -  I(R;W)_{\sigma_{}}  - \log{p} +S(W)_{\sigma_{}} - \delta_{w} - \delta_{\rho}' -\delta_{w}'' \right)}. \nonumber
\end{align}

% Applying the triangle inequality gives,
% \begin{align}
%     \widetilde{H}  & \leq \left\|\sum_{w^n}\lambdawA \rhohatwA - \sum_{w^n}\lambdawA \rhotildwA \right\|_1+ \left\|\sum_{w^n}\lambdawA \rhohatwA - \frac{1}{(1+\eta)}\frac{p^n}{2^{nS}} \sum_{w^n}\sum_{a,i} \lambdawA\rhohatwA \11_{\{\WCodeword=w^n\}}\right\|_1  \nonumber \\ 
%     & \hspace{10pt}+ \left\|\frac{1}{(1+\eta)}\frac{p^n}{2^{nS}} \sum_{w^n}\sum_{a,i} \lambdawA(\rhohatwA - \rhotildwA) \11_{\{\WCodeword=w^n\}}\right\|_1.
% \end{align}
% The first and the third term in the right hand side of the above equation can be shown to be small using the Gentle Measurement Lemma \cite{Wilde_book}. For the second term, we employ Lemma \ref{lem:nf_SoftCovering} which gives $\EE[\widetilde{H}] \leq \epsilon_{\Tilde{H}}$ if $S \geq I(W;R)_{\sigma} + \log{p} - S(W)_{\sigma} - \delta_{\widetilde{H}}$.

% gives
% \begin{align}
%     \EE[\widetilde{H}] \leq  2^{-n(S - (I(W;R)_{\sigma} - \log{p} + S(W)_{\sigma} - \delta_{\widetilde{H}))}} \nonumber
% \end{align}
\noindent Next, consider $H_0$ and perform the following simplification
% \begin{widetext}
\begin{align}
\EE[H_0] & =\frac{(1-\varepsilon)}{(1+\eta)}\EE\bigg|\sum_{w^n\in \TDeltan(W)}\!\!\!\!\frac{\lambdawA}{(1-\varepsilon)} \nonumber \\
& \hspace{10pt}- \frac{p^n}{p^{k+l}}\!\!\!\! \sum_{w^n\in \TDeltan(W)}\sum_{a,i}\frac{\lambdawA}{(1-\varepsilon)}\11_{\{\WCodeword = w^n\}}\bigg| \nonumber \\
& = \frac{(1-\varepsilon)}{(1+\eta)}\EE\bigg\|\sum_{w^n\in \TDeltan(W)}\!\!\!\!\frac{\lambdawA}{(1-\varepsilon)}\omega_0^{\tensor n} - \frac{p^n}{p^{k+l}}\hspace{-15pt}  \nonumber \\
& \hspace{5pt}\times\sum_{w^n\in \TDeltan(W)}\sum_{a,i}\frac{\lambdawA}{(1-\varepsilon)}\11_{\{\WCodeword = w^n\}}\omega_0^{\tensor n}\bigg\|_1,
\end{align}
% \end{widetext}
where $\omega_0 \in \mathcal{D}(\mathcal{H})$ is any state independent of $W$. We again apply Lemma \ref{lem:Change Measure Soft Covering Variance Based} to the above term with the following identification: $\lambda_x$ with $\frac{\lambdawA}{(1-\varepsilon)}$, $\sigma_x$ with $\omega_0^{\tensor n}$, $\mathcal{X}$ with $\TDeltan(W)$, $\mathcal{\bar{X}}$ with $\FF_p^n$, $\Pi$ and $\Pi_x$ with Identity operator $I$, and  $\mu_x$ with $\frac{1}{p^n}$. With this identification, $\kappa$ remains as above, $\kappa = 2^{-n{(S(W)_{\sigma}-\log{p}-\delta_w)}}$  and $D= d = 1$. 
% and using the fact that $I(W;R)_{\omega} = 0$ for $\omega^{RW} \deq \sum_{w}\lambda_w \omega_0 \tensor\ketbra{w}$, 
Hence, using in inequality \eqref{lem:ChangeMeasure:ineq2} of Lemma \ref{lem:Change Measure Soft Covering Variance Based}, we obtain 
\begin{align}
    \EE[H_0] & \leq 2^{-\frac{n}{2}\left(\frac{k+l}{n}\log{p}  - \log{p} +S(W)_{\sigma_{}} - \delta_{w} \right)}. \nonumber
\end{align}
% $\EE[H_0] \leq \epsilon_{H_0}$ if $\frac{k+l}n{\log{p}} >  \log{p} - S(W)_{\sigma}  + \delta_{w}$, where $\epsilon_{H_0}(\delta) \searrow 0$ as $\delta \searrow 0.$  
This completes the proof.

% To ensure $\EE[\widetilde{H}]$ is arbitrarily small for sufficiently large $n$, we need $S \geq I(W;R)_{\sigma} + \log{p} - S(W)_{\sigma} - \delta_{\widetilde{H}}$. Since $I(W;R)_{\sigma} \geq 0$, by choosing  $S \geq I(W;R)_{\sigma} + \log{p} - S(W)_{\sigma} - \max{(\delta_{\widetilde{H}},\delta_{H_0})}$, $\EE[H_0]$ can also be made arbitrary small for sufficiently large $n$. Substituting the above bounds in $S_{13}$, completes the proof.
 
\subsection{Proof of Proposition \ref{prop:Lemma for p2p:S_2}}\label{appx:proof of p2p:S2}
We begin using the definition of $ A_{w^n}^{(\mu)} $ and applying triangle inequality to $ S_2 $ to obtain
% \begin{widetext}
\begin{align}
S_{2} & \leq \frac{1}{(1+\eta)}\frac{1}{N}\sum_{\mu}\sum_{a,i>0} \sum_{w^n,z^n} \cfrac{\lambdawA p^n}{p^{k+l}} \mathbbm{1}_{\{aG^{} + h^{(\mu)}(i) = w^n\}} \nonumber \\
& \hspace{15pt}\times{\left \|\sqrt{\rho^{\tensor n}}  \CutOff\sqrt{\rho^{\tensor n}}^{-1}\rhotildwA\sqrt{\rho^{\tensor n}}^{-1}\CutOff\sqrt{\rho^{\tensor n}}^{}\right \|_1}\nonumber\\ 
& \hspace{0.7in}\times\left|P^n_{Z|W}(z^n|w^n)-   P^n_{Z|W}\left (z^n|F^{(\mu)}(i)\right )\right| \nonumber \\
& \leq \frac{2^{2n\delta_\rho}}{(1+\eta)}\frac{1}{N}\sum_{\mu}\sum_{a,i>0} \sum_{w^n,z^n} \cfrac{\lambdawA p^n}{p^{k+l}} \mathbbm{1}_{\{aG^{} + h^{(\mu)}(i) = w^n\}}  \nonumber \\
& \hspace{20pt} \times  \left|P^n_{Z|W}(z^n|w^n)-   P^n_{Z|W}\left (z^n|F^{(\mu)}(i)\right )\right| \nonumber \\
& \leq \frac{2^{2n\delta_\rho}}{(1+\eta)}\frac{1}{N}\sum_{\mu}\sum_{a,i>0} \sum_{w^n} 2\cfrac{\lambdawA p^n}{p^{k+l}}  \nonumber \\
& \hspace{.9in}\times\mathbbm{1}_{\{aG^{} + h^{(\mu)}(i) = w^n\}}\mathbbm{1}^{(\mu)}(w^n,i),   \label{eq:S_2ineq}
\end{align}
% \end{widetext}
where the second inequality above uses the following arguments
\begin{align}
   &\left\| \sqrt{\rho^{\tensor n}}  \CutOff\sqrt{\rho^{\tensor n}}^{-1}\rhotildwA\sqrt{\rho^{\tensor n}}^{-1}\CutOff\sqrt{\rho^{\tensor n}}^{}\right\|_1  \nonumber \\
   & = \left\|\sqrt{\rho^{\tensor n}}\Pi_\rho  \CutOff\sqrt{\rho^{\tensor n}}^{-1}\Pi_\rho\rhotildwA\Pi_\rho\sqrt{\rho^{\tensor n}}^{-1}\CutOff\Pi_\rho\sqrt{\rho^{\tensor n}}^{} \right\|_1\nonumber \\
    & \leq \left\|\sqrt{\rho^{\tensor n}}\Pi_\rho\right\|_\infty\left\|  \CutOff\sqrt{\rho^{\tensor n}}^{-1}\Pi_\rho\rhotildwA\Pi_\rho\sqrt{\rho^{\tensor n}}^{-1}\CutOff \right\|_1 \nonumber\\ & \hspace{20pt}\times \left\|\sqrt{\rho^{\tensor n}}\Pi_\rho\right\|_\infty \nonumber \\
    % & \leq 2^{-n(S(\rho)-\delta_\rho)}\left\|  \CutOff\right\|_{\infty}\left\|\sqrt{\rho^{\tensor n}}^{-1}\Pi_\rho\rhotildwA\Pi_\rho\sqrt{\rho^{\tensor n}}^{-1}\right\|_1\left\|\CutOff \right\|_\infty \nonumber \\
     & \leq 2^{-n(S(\rho)-\delta_\rho)}\left\|  \CutOff\right\|^2_{\infty}\left\|\sqrt{\rho^{\tensor n}}^{-1}\Pi_\rho\rhotildwA\Pi_\rho\sqrt{\rho^{\tensor n}}^{-1}\right\|_1 \nonumber \\
    & \leq 2^{2n\delta_\rho}\left\|\rhotildwA\right\|_1 \leq 2^{2n\delta_\rho},
\end{align}
where the above inequalities follow from the H\'older's inequality. 
% fact that $ {\left \|  \CutOff\rhotildwA\CutOff\right \|_1} \leq 1,
Finally, the last inequality in \eqref{eq:S_2ineq} follows by defining $\mathbbm{1}^{(\mu)}(w^n,i) $ as
\begin{align}
\mathbbm{1}^{(\mu)}(w^n,i)  \deq \11\bigg\{\exists & (\tilde{w}^n,\tilde{a}^n): \tilde{w}^n = \tilde{a}^nG^{}+ h^{(\mu)}(i), \nonumber \\
& \tilde{w}^n \in \mathcal{T}_{{\delta}}^{(n)}(W), \tilde{w}^n \neq w^n\bigg\}. \nonumber
\end{align}
Observe that
\begin{align}
\EE[\mathbbm{1}^{(\mu)}(w^n,i)\mathbbm{1}_{\{ aG^{} + h^{(\mu)}(i)= w^n\}}] \leq 
 \sum_{\tilde{a} \in \FF_p^{k}}\sum_{\substack{\tilde{w} \in \mathcal{T}_{{\delta}}^{(n)}(W) \\ \tilde{w} \neq w^n}}\frac{1}{p^np^n}, \nonumber
\end{align}
which follows from the pairwise independence of the codewords.
Using this, we obtain
\begin{align}
\EE[S_2] &\leq \frac{2\;2^{2n\delta_\rho}}{(1+\eta)}\frac{2^{-nR}p^{k+l}}{p^n}\sum_{\tilde{w}^n \in \mathcal{T}_{{\delta}}^{(n)}(W)} \sum_{{w}^n \in \mathcal{T}_{{\delta}}^{(n)}(W)} \lambdawA  \nonumber \\
& \leq 2\; {2^{n(\frac{k+l}{n}\log{p}-R - \log{p} + S(W)_{\sigma} + \delta_{S_2})}}, \nonumber
\end{align}
where $ \delta_{S_2} \searrow 0 $ as $  \delta \searrow 0 $, and $ \sigma $ is as defined in the statement of the theorem. This completes the proof.

\subsection{Proof of Proposition \ref{prop:Lemma for S_2}}\label{appx:proof of S_2}

Recalling $ S_{2} $, we have $S_2 \leq S_{21}+S_{22}$, where
\begin{align}
% S_{2}& \leq \frac{1}{N_1N_2} \sum_{\mu_1,\mu_2}\sum_{z^n}\sum_{u^n,v^n}\left|P^n_{Z|U,V}(z^n|u^n,v^n)- P^n_{Z|U+V}\left (z^n|e^{(\mu)}(u^n,v^n)\right )\right| \left \|   \sqrt{\rho_{AB}^{\tensor n}}\left (\gammaCoeff A_{u^n}^{(\mu_1)} \tensor \zetaCoeff B_{v^n}^{(\mu_2)}\right ) \sqrt{\rho_{AB}^{\tensor n}} \right \|_1 \nonumber \\
% & \leq \frac{1}{N_1N_2} \sum_{\mu_1,\mu_2}\sum_{u^n,v^n}\sum_{z^n}\left|P^n_{Z|U,V}(z^n|u^n,v^n)- P^n_{Z|U,V}\left (z^n|e^{(\mu)}(u^n,v^n)\right )\right| \alpha_{u^n}\beta_{v^n}\gamma_{u^n}^{(\mu_1)}\zeta_{v^n}^{(\mu_2)}\Omega_{u^n,v^n} \nonumber \\
% S_{21} & \deq  \cfrac{1}{(1+\eta)^2}\frac{2}{N_1N_2}  \sum_{\mu_1,\mu_2}\sum_{i,j}\sum_{a_1,a_2} \sum_{u^n,v^n}\frac{p^np^n\lambda_{u^n}^A\lambda_{v^n}^B}{2^{n(S_1+S_2)}}\Omega_{u^n,v^n} \mathbbm{1}^{(\mu)}(u^n+ v^n,i,j)
% \mathbbm{1}_{\{ a_1G^{} + h_1^{(\mu_1)}(i)= u^n\}}\mathbbm{1}_{\{ a_2G^{} + h_2^{(\mu_2)}(j)= v^n\}},\nonumber \\ 
% S_{22} & \deq  \cfrac{1}{(1+\eta)^2}\frac{2}{N_1N_2}  \sum_{\mu_1,\mu_2}\sum_{i,j}\sum_{a_1,a_2} \sum_{u^n,v^n}\frac{p^np^n\lambda_{u^n}^A\lambda_{v^n}^B}{2^{n(S_1+S_2)}}\Omega_{u^n,v^n} 
% \mathbbm{1}_{\{(u^n,v^n)\not \in\TDeltaN(U,V)\}}
% \mathbbm{1}_{\{ a_1G^{} + h_1^{(\mu_1)}(i)= u^n\}}\mathbbm{1}_{\{ a_2G^{} + h_2^{(\mu_2)}(j)= v^n\}}, \nonumber
S_{21} & \deq \frac{2}{N_1N_2}  \sum_{\bar{\mu}_1,\bar{\mu}_2} \sum_{u^n,v^n}\alpha_{u^n}\beta_{v^n}\gammaBarCoeff\zetaBarCoeff\Omega_{u^n,v^n} \nonumber \\ & \hspace{1.6in}\times 
\mathbbm{1}_{\{(u^n,v^n)\not \in\TDeltaN(U,V)\}}
, \nonumber \\
S_{22} & \deq  \frac{2}{\bar{N}_1\bar{N}_2}  \sum_{\bar{\mu}_1,\bar{\mu}_2} \sum_{u^n,v^n}\alpha_{u^n}\beta_{v^n}\gammaBarCoeff\zetaBarCoeff\Omega_{u^n,v^n} \nonumber \\ & \hspace{1.6in} \times  \mathbbm{1}^{(\bar{\mu}_1,\bar{\mu}_2)}(u^n+ v^n,i,j),\nonumber 
\end{align}
where $ \Omega_{u^n,v^n} $ and $ \mathbbm{1}^{(\bar{\mu}_1,\bar{\mu}_2)}(w^n,i,j)$ are defined as
% \begin{widetext}
\begin{align*}
&\Omega_{u^n,v^n} \deq \tr\Big\{\left[\left( \CutOffBarA\tensor\CutOffBarB\right)\sqrt{\rho^{\tensor n}_A\tensor \rho^{\tensor n}_B}^{-1}(\rhotilduA\tensor\right. \\
& \hspace{0.9in}\left. \rhotildvB)  \sqrt{\rho^{\tensor n}_A\tensor \rho^{\tensor n}_B}^{-1}\left( \CutOffBarA\tensor\CutOffBarB\right)\right] \rho^{\tensor n}_{AB}\Big\},\\
& \mathbbm{1}^{(\bar{\mu}_1,\bar{\mu}_2)}(w^n,i,j) \\ & \hspace{10pt} \deq \11\bigg\{\exists (\tilde{w}^n,\tilde{a}^n): \tilde{w}^n = \tilde{a}^nG^{}+ h_1^{(\bar{\mu}_1)}(i) + h_2^{(\bar{\mu}_2)}(j), \\& \hspace{1.5in} \tilde{w}^n \in \mathcal{T}_{\hat{\delta}}^{(n)}(U+ V), \tilde{w}^n \neq w^n\bigg\}.
\end{align*}
% \end{widetext}

We begin by bounding the term corresponding to $S_{21}$. Consider the following argument. 
\begin{widetext}
\begin{align*}
    S_{21} 
    &\leq \Bigg| \frac{2}{\bar{N}_1\bar{N}_2}  \sum_{\bar{\mu}_1,\bar{\mu}_2}  \sum_{u^n,v^n}\alpha_{u^n}\beta_{v^n}\gammaBarCoeff\zetaBarCoeff\Omega_{u^n,v^n} \mathbbm{1}_{\{(u^n,v^n)\not \in\TDeltaN(U,V)\}} - \hspace{-20pt}\sum_{
    \substack{(u^n,v^n) \not \in \TDeltaN(UV) \nonumber \\
    u^n \in \TDeltaN(U),v^n \in \TDeltaN(V)}}\hspace{-30pt} 2\lambda_{u^n,v^n}^{AB}\Bigg| + \hspace{-10pt} \sum_{(u^n,v^n) \not \in \TDeltaN(U,V)} \hspace{-10pt} 2\lambda_{u^n,v^n}^{AB} \\
    &\overset{(a)}{\leq} 2\sum_{u^n \in \mathcal{U}^n} \sum_{v^n \in \mathcal{V}^n} \left|\lambda^{AB}_{u^n,v^n} - \frac{1}{\bar{N}_1\bar{N}_2} \sum_{\bar{\mu}_1,\bar\mu_2}\alpha_{u^n}\beta_{v^n} \gamma^{(\bar\mu_1)}_{u^n} \zeta^{(\bar\mu_2)}_{v^n} \Omega_{u^n,v^n} \right|  
    + \sum_{(u^n,v^n) \not \in \TDeltaN{(UV)}}2\lambda^{AB}_{u,v} \\
    &\overset{(b)}{\leq} 2\tilde{S}_1+2\sum_{(u^n,v^n) \not \in \TDeltaN{(UV)}}\lambda^{AB}_{u^n,v^n}, 
\end{align*}
\end{widetext}
where 
\begin{align*}
\tilde{S}_1 &\deq    \bigg \| (\text{id}\tensor \bar{M}_A^{\tensor n}\tensor \bar{M}_B^{\tensor n}) (\Psi^\rho_{RAB})^{\tensor n}- \\& \hspace{15pt} \frac{1}{\bar{N}_1\bar{N}_2}\sum_{\bar{\mu}_1,\bar{\mu}_2}(\text{id}\tensor  [M_1^{(\bar{\mu}_1)}]\tensor [M_2^{(\bar{\mu}_2)}]) (\Psi^\rho_{RAB})^{\tensor n}\bigg \|_1,
\end{align*}
(a) follows by applying the triangle inequality, and (b) follows from the Lemma \ref{lem:classical prob} given below. 
Note that in $\tilde{S}_1$, the average over the entire common information sequence $(\bar{\mu}_1,\bar{\mu}_2)$ is inside the norm. 
\begin{lem} \label{lem:classical prob}
{We have 
\begin{align}
    \sum_{u^n \in \mathcal{U}^n}&\sum_{v^n \in \mathcal{V}^n}
\bigg|\lambda^{AB}_{u^n,v^n} - \nonumber \\
&\frac{1}{\bar{N}_1\bar{N}_2} \sum_{\bar{\mu}_1,\bar{\mu}_2} \alpha_{u^n}\beta_{v^n}\gammaBarCoeff \zetaBarCoeff \Omega_{u^n,v^n}
\bigg|   \leq S_1.
\end{align}}
\end{lem}
\begin{proof}
{The proof follows from Lemma 2 in \cite{wilde_e}}.
\end{proof}

Next we use Theorem \ref{thm:p2pTheorem} twice with (a) 
 $\rho=\rho_A$, $M=\bar{M}_A$, 
$\mathcal{W}=\mathcal{U}$, $\mathcal{Z}=\mathcal{U}$
and $P_{Z|W}(z|w)=\mathbbm{1}{\{z=w\}}$, and (b) $\rho=\rho_B$, $M=\bar{M}_B$, 
$\mathcal{W}=\mathcal{V}$, $\mathcal{Z}=\mathcal{V}$
and $P_{Z|W}(z|w)=\mathbbm{1}{\{z=w\}}$, and \cite[Lemma 5]{atif2019faithful} to yield the following: for any $\epsilon \in (0,1)$, and any $\eta,\delta \in (0,1)$ sufficiently small, and any $n$ sufficiently large $\EE[S_1] \leq  2\epsilon$ 
if $\frac{k+l_1}{n} \log p > I(U;RB)_{\sigma_1}-S(U)_{\sigma_3}+\log p$, $\frac{k+l_2}{n} \log p  > I(V;RA)_{\sigma_2}-S(V)_{\sigma_3}+\log p$, 
$\frac{k+l_1}{n} \log p +\frac{1}{n} \log \bar{N}_1 > \log p$, $\frac{k+l_2}{n} \log p +\frac{1}{n} \log \bar{N}_2 > \log p$, where 
$\sigma_1, \sigma_2$ and $\sigma_3$ are defined as in the statement of the theorem.
Consequently, we have $\EE[S_{21}] \leq 4\epsilon $ for all sufficiently large $n$.

% Using the Lemmas \ref{lem:M_XM_Y POVM} and \ref{lem:mutual covering}, and the proof of Theorem \ref{thm:one_shot_ptp}, for any $\epsilon \in (0,1)$, and any $\eta,\delta \in (0,1)$ sufficiently small, and any $n$ sufficiently large,   if 
% \begin{align}
%     \tilde{R}_1 >  I(U;RB)_{\sigma_1}, &\quad
%     \tilde{R}_2 > I(V;RA)_{\sigma_2}, \nonumber \\
%     \tilde{R_1} + \frac{1}{n}\log(\bar{N}_1) > S(U)_{\sigma_3},  &\quad \tilde{R_2} + \frac{1}{n}\log(\bar{N}_2) > S(V)_{\sigma_3},
% \label{eq:interm_rates}
% \end{align}
% then  $\EE[\tilde{S}_1] \leq \epsilon$.
% Consequently, we have  
% \[
% \frac{1}{\tilde{N}_1\tilde{N}_2} \sum_{\tilde{\mu}_1,\tilde{\mu}_2} \EE\left[S_{21}(\tilde{\mu}_1,\tilde{\mu}_2)
% \mathbbm{1}_{\{\mbox{sP-1}\}}(\tilde{\mu}_1,\tilde{\mu}_2)\mathbbm{1}_{\{\mbox{sP-2}\}}(\tilde{\mu}_1,\tilde{\mu}_2)\right] \leq 4\epsilon.
% \]  

In regards to $S_{22}$, note that
\begin{align}
& \EE\big[\mathbbm{1}^{(\bar\mu_1, \bar\mu_2)}(u^n+ v^n,i,j)\mathbbm{1}_{\{ a_1G^{} + h_1^{(\bar\mu_1)}(i)= u^n\}}\nonumber \\
&\hspace{30pt}\mathbbm{1}_{\{ a_2G^{} + h_2^{(\bar\mu_2)}(j)= v^n\}}\bigg]  \leq  \sum_{\substack{\tilde{a} \in \FF_p^{k}\\\tilde{a}\neq a}}\;\;\sum_{\substack{\tilde{w} \in \mathcal{T}_{\hat{\delta}}^{(n)}(U+ V) \\ \tilde{w} \neq u^n + v^n}}\frac{1}{p^n p^n p^n}. \nonumber
\end{align}
Using this, we obtain
\begin{align}
\EE[S_{22}] &\leq \cfrac{2}{(1+\eta)^2}\frac{p^{k+l_1}2^{nR_1}}{p^n}\sum_{\tilde{w}^n \in \mathcal{T}_{\hat{\delta}}^{(n)}(U+ V)} \nonumber \\
& \hspace{30pt} \sum_{u^n \in \TDeltaN(U)}\sum_{v^n \in \TDeltaN(V)}\lambda_{u^n}^A\lambda_{v^n}^B \Omega_{u^n,v^n} \nonumber \\
& \leq \frac{2\; {2^{n(\frac{k+l_1}{n}\log{p}-R_1 - \log{p} + S(U+ V)_{\sigma_3} + \delta_{\rho_{AB}}+\hat{\delta}_W)}}}{(1+\eta)^2},\nonumber
\end{align}
where $\hat{\delta}_W \searrow 0$ as $\delta \searrow 0$ and the above inequality follows from the following lemma (Lemma \ref{lem:omega lambda}). Hence, $\EE[S_{21}] \leq \epsilon$ if the conditions in the proposition are satisfied.
\begin{lem}\label{lem:omega lambda}
	For $ \lambdauA, \lambdavB $  and $ \Omega_{u^n,v^n} $ as defined above,
	% 	$\Omega_{u^n,v^n}\deq \tr\Big\{ \sqrt{\rho^{\tensor n}_A\tensor \rho^{\tensor n}_B}^{-1}(\LambdauA\tensor\LambdavB)  \sqrt{\rho^{\tensor n}_A\tensor \rho^{\tensor n}_B}^{-1} \rho^{\tensor n}\Big\}$, 
	we have $$ \sum_{u^n \in \TDeltan(U)} \sum_{v^n \in \TDeltan(V)} \Omega_{u^n,v^n}\lambdauA\lambdavB \leq 2^{n\delta_{\rho_{AB}}},$$ for some $ \delta_{\rho_{AB}} \searrow 0 $ as $ \delta \searrow 0. $
\end{lem}
\begin{proof}
Firstly, note that 
% \begin{widetext}
\begin{align}
& \sum_{\substack{u^n,v^n}}\Omega_{u^n,v^n}\lambdauA\lambdavB \nonumber \\ 
& \hspace{1pt} = \text{Tr}\bigg\{\bigg[\!\!\left( \CutOffBarA\tensor\CutOffBarB\right)\!\bigg(\!\sqrt{\rho_{A}^{\tensor n}}^{-1}\!\!\Big(\sum_{u^n}\lambdauA\rhotilduA\Big)\sqrt{\rho_{A}^{\tensor n}}^{-1}\hspace{-5pt}  \nonumber\\
& \;\;\;\tensor\sqrt{\rho_{B}^{\tensor n}}^{-1}\!\!\Big(\sum_{v^n}\lambdavB\rhotildvB\Big) \sqrt{\rho_{B}^{\tensor n}}^{-1}\bigg)\!\!\left( \CutOffBarA\tensor\CutOffBarB\right) \!\!\bigg] \rho^{\tensor n}_{AB}\bigg\}. \label{eq:omegalambdaUlambdaV}
\end{align} 
% \end{widetext}
We know, $\sum_{u^n}\lambdauA\rhotilduA  \leq  2^{-n(S(\rho_A)-\delta_{\rho_A})}\Pi_{\rho_A}$, where $ \delta_{\rho_A} \searrow 0 $ as $ \delta \searrow 0 $.
This implies, 
\begin{align}
\CutOffBarA&\sqrt{\rho_{A}^{\tensor n}}^{-1}\left(\sum_{u^n}\lambdauA\rhotilduA\right)\sqrt{\rho_{A}^{\tensor n}}^{-1}\CutOffBarA \nonumber \\ &\leq 2^{-n(S(\rho_A)-\delta_{\rho_A})}\CutOffBarA\sqrt{\rho_{A}^{\tensor n}}^{-1}\PiA\sqrt{\rho_{A}^{\tensor n}}^{-1}\CutOffBarA \nonumber \\ &\leq  2^{2n\delta_{\rho_A}}\CutOffBarA\Pi_{\rho_A}\CutOffBarA \leq 2^{2n\delta_{\rho_A}}\CutOffBarA, \label{eq:ineqrhoA}
\end{align}
where the second inequality appeals to  the fact that $ \sqrt{\rho_{A}^{\tensor n}}^{-1} \Pi_{\rho_{A}} \sqrt{\rho_{A}^{\tensor n}}^{-1} \leq 2^{n(S(\rho_A)+\delta_{\rho_A})}\Pi_{\rho_A}$. 
%Further, using the above inequality and the fact that $ \PihatA  $ commutes with $ \rho_{A}^{\tensor n} $, we have
%\begin{align} 
%\sqrt{\rho_{A}^{\tensor n}}^{-1}\left(\sum_{u^n}\lambdauA\LambdauA\right)\sqrt{\rho_{A}^{\tensor n}}^{-1} \leq  \sqrt{\rho_{A}^{\tensor n}}^{-1} \PihatA\rho_{A}^{\tensor n}\PihatA  \sqrt{\rho_{A}^{\tensor n}}^{-1} = \PihatA \label{eq:ineqrhoA}
%\end{align}
Similarly, using the same arguments above for the operators acting on $ \mathcal{H}_B $, we have
\begin{equation} 
\CutOffBarB\sqrt{\rho_{B}^{\tensor n}}^{-1}\left(\sum_{v^n}\lambdavB\rhotildvB\right)\sqrt{\rho_{B}^{\tensor n}}^{-1}\CutOffBarB \leq 2^{2n\delta_{\rho_B}}\CutOffBarB, \label{eq:ineqrhoB}
\end{equation}
where $ \delta_{\rho_B} \searrow 0 $ as $ \delta \searrow 0 $.
Using (i) the simplifications in \eqref{eq:ineqrhoA} and \eqref{eq:ineqrhoB}, and (ii) the fact that for $A_1 \geq B_1 \geq 0 $ and $A_2 \geq B_2 \geq 0 $, $(A_1\tensor A_2)\geq (B_1\tensor B_2)$  in \eqref{eq:omegalambdaUlambdaV}, gives 
\begin{align}
\sum_{u^n,v^n}&\Omega_{u^n,v^n}\lambdauA\lambdavB \nonumber \\
& \leq 2^{2n(\delta_{\rho_A}+\delta_{\rho_B})}\Tr{\left (\CutOffBarA\tensor\CutOffBarB\right ) \rho^{\tensor n}_{AB}} \nonumber \\
& \leq 2^{2n(\delta_{\rho_A}+\delta_{\rho_B})}\Tr{ \rho^{\tensor n}_{{AB}}} = 2^{2n(\delta_{\rho_A}+\delta_{\rho_B})}. \nonumber
\end{align}
Substituting $ \delta_{\rho_{AB}} = 2(\delta_{\rho_A}+\delta_{\rho_B}) $ gives the result.

\end{proof}

%%%%-------------------------------------------%%%

\subsection{Proof of Proposition \ref{prop:Lemma for Tilde_S}}\label{appx:proof of Tilde_S} 
We bound $\widetilde{S}$ as $\widetilde{S}\leq \widetilde{S}_2 +\widetilde{S}_3 +\widetilde{S}_4$, where
% \begin{widetext}
\begin{align*}
\widetilde{S}_2 \deq &\bigg\|\frac{1}{N_1N_2}\sum_{\mu_1,\mu_2}\sum_{i>0} \sqrt{\rho_{AB}^{\tensor n}}\left (\Gamma^{A, ( \mu_1)}_i\tensor \Gamma^{B, (\mu_2)}_0\right )  \\ & \hspace{1.5in}\times\sqrt{\rho_{AB}^{\tensor n}} P^n_{Z|U+V}(z^n|w^{n}_{0})\bigg\|_1, \nonumber \\
\widetilde{S}_3 \deq  &\bigg\|\frac{1}{N_1N_2}\sum_{\mu_1,\mu_2}\sum_{j>0} \sqrt{\rho_{AB}^{\tensor n}}\left (\Gamma^{A, ( \mu_1)}_0\tensor \Gamma^{B, (\mu_2)}_j\right )\\ & \hspace{1.5in}\times\sqrt{\rho_{AB}^{\tensor n}} P^n_{Z|U+V}(z^n|w^{n}_{0})\bigg\|_1, \nonumber \\
\widetilde{S}_4 \deq  &\bigg\|\frac{1}{N_1N_2}\sum_{\mu_1,\mu_2} \sqrt{\rho_{AB}^{\tensor n}}\left (\Gamma^{A, ( \mu_1)}_0\tensor \Gamma^{B, (\mu_2)}_0\right )\\ & \hspace{1.5in}\times \sqrt{\rho_{AB}^{\tensor n}} P^n_{Z|U+V}(z^n|w^{n}_{0})\bigg\|_1. \nonumber  
\end{align*}
% \end{widetext}
% \newpage
\noindent{\bf Analysis of $\widetilde{ S}_2 $:} We have 
\begin{widetext}
\begin{align}
%S_{2} &= \sum_{z^n}\left\|\frac{1}{N_1N_2}\sum_{\mu_1,\mu_2}\sum_{i>0} \sqrt{\rho_{AB}^{\tensor n}}\left (\Gamma^{A, ( \mu_1)}_i\tensor \Gamma^{B, (\mu_2)}_0\right )\sqrt{\rho_{AB}^{\tensor n}} P^n_{Z|U+V}(z^n|w^{n}_{0})\right\|_1 \nonumber \\
\widetilde{S}_{2} & \leq  \frac{1}{N_1N_2}\sum_{\mu_1,\mu_2}\sum_{i>0}\sum_{z^n}P^n_{Z|U+V}(z^n|w^{n}_{0})\left\| \sqrt{\rho_{AB}^{\tensor n}}\left (\Gamma^{A, ( \mu_1)}_i\tensor \Gamma^{B, (\mu_2)}_0\right )\sqrt{\rho_{AB}^{\tensor n}} \right\|_1 \nonumber \\
%& \leq  \frac{1}{N_1N_2}\sum_{\mu_1,\mu_2}\sum_{i>0}\left\| \sqrt{\rho_{AB}^{\tensor n}}\left (\Gamma^{A, ( \mu_1)}_i\tensor \Gamma^{B, (\mu_2)}_0\right )\sqrt{\rho_{AB}^{\tensor n}} \right\|_1 \nonumber \\
& \leq  \frac{1}{N_1N_2}\sum_{\mu_1,\mu_2}\left\| \sqrt{\rho_{B}^{\tensor n}} \Gamma^{B, (\mu_2)}_0\sqrt{\rho_{B}^{\tensor n}} \right\|_1 \nonumber \\
& \leq \frac{1}{N_2}\sum_{\mu_2}\left\| \sum_{v^n}\lambdavB\rhohatvB - \sum_{v^n}\sqrt{\rho_{B}^{\tensor n}} \zetaCoeff \bar{B}^{(\mu_2)}_{v^n}\sqrt{\rho_{B}^{\tensor n}} \right\|_1  + \frac{1}{N_2}\sum_{\mu_2}\left\| \sum_{v^n}\sqrt{\rho_{B}^{\tensor n}} \zetaCoeff \left(\bar{B}^{(\mu_2)}_{v^n} - B^{(\mu_2)}_{v^n}\right)\sqrt{\rho_{B}^{\tensor n}} \right\|_1 \nonumber \\
& \leq \underbrace{\frac{1}{N_2}\sum_{\mu_2}\left\| \sum_{v^n}\lambdavB\rhohatvB - \cfrac{1}{(1+\eta)}\cfrac{p^n}{p^{k+l_2}}\sum_{v^n}\sum_{a_2,j}\lambdavB\rhohatvB\11_{\{\VCodeword=v^n\} } \right\|_1}_{\widetilde{S}_{21}} +\underbrace{ \frac{1}{N_2}\sum_{\mu_2}\sum_{v^n}\beta_{v^n}\zeta_{v^n}^{(\mu_2)}\left\| \rhohatvB -  \rhotildvB \right\|_1}_{\widetilde{S}_{22}} \nonumber \\
& \hspace{100pt} + \underbrace{\frac{1}{N_2}\sum_{\mu_2}\sum_{v^n}\left\| \sqrt{\rho_{B}^{\tensor n}} \zetaCoeff \left(\bar{B}^{(\mu_2)}_{v^n} - B^{(\mu_2)}_{v^n}\right)\sqrt{\rho_{B}^{\tensor n}} \right\|_1}_{\widetilde{S}_{23}}, \label{eq:S2simplification}
\end{align}
\end{widetext}
where the first inequality uses triangle inequality.
%, the second uses the fact that $ \sum_{i>0}\Gamma^{A, ( \mu_1)}_i = \sum_{u^n}A_{u^n}^{(\mu_1)} $. 
The next inequality follows by using Lemma \ref{lem:Separate} where we use the fact that  $\sum_{i>0}\Gamma^{A, ( \mu_1)}_i\leq I.$
% , given that $ S_1 \geq I(U;RB)_{\sigma_{2}} - S(U)_{\sigma_3} + \log{|\mathcal{U}|}. $
Finally, the last two inequalities follows again from triangle inequality. 

%\noindent For the first term in \eqref{eq:S2simplification} we use the Pruned Covering Lemma ( Lemma \eqref{lem:softCoveringPrunedDistribution}) and conclude that for any given  $ \epsilon_{\scriptscriptstyle S_2}, \delta_{\scriptscriptstyle S_{2}} \in (0,1)$ if $ \tilde{R}_2 \geq I(V;RA)_{\sigma_{2}} $ then $ \PP() $ where $ \sigma_{2} $ is as defined in the statement of the theorem. 

\noindent Regarding the first term in \eqref{eq:S2simplification}, 
% is very similar to what we have for $ J_{1} $ in \eqref{eq:J1} and
using Lemma \ref{lem:nf_SoftCovering} we claim that for all 
$\epsilon>0$, and $\eta,\delta \in (0,1)$ sufficiently small, and any $n$ sufficiently large, 
$\EE[\tilde{S}_{21}]<\epsilon$, if  $\frac{k+l_2}{n}\log{p} \geq I(V;RA)_{\sigma_{2}} - S(V)_{\sigma_{3}} + \log{p}$, where $ \sigma_{2},\sigma_3 $ are as defined in the statement of the theorem. 
% The formal proof of a very similar result is provided in Appendix \ref{appx:proof of p2p:S_11}. 
%Note that the  requirements we obtain on $ S_1$ and $ S_2$ here were already imposed earlier in Section \ref{subsec:stochastic_dist_POVM}.
% by Lemma \ref{lem:M_XM_Y POVM}.
As for the second term, we use the gentle measurement lemma (as in \eqref{eq:gentleMeasurementVn}) and bound its expected value as
\begin{align}
\EE[\tilde{S}_{22}] \!&= \EE\left[\frac{1}{N_2}\sum_{\mu_2}\sum_{v^n}\beta_{v^n}\zeta_{v^n}^{(\mu_2)}\left\| \rhohatvB -  \rhotildvB \right\|_1\right] \nonumber \\ 
& = \hspace{-16pt}\sum_{v^n \in \TDeltaN(V)}\hspace{-6pt}\frac{\lambdavB}{(1+\eta)}\!\!\left\| \rhohatvB -  \rhotildvB \right\|_1 \hspace{-3pt} +\hspace{-16pt}\sum_{v^n \notin \TDeltaN(V)}\hspace{-6pt}\frac{\lambdavB}{(1+\eta)}\!\!\left\| \rhohatvB \right\|_1 \nonumber \\
&\leq \epsilon_{\scriptscriptstyle \widetilde{S}_{21}},\nonumber
\end{align}
where the inequality is based on the repeated usage of the Average Gentle Measurement Lemma and $ \epsilon_{\scriptscriptstyle\widetilde{S}_{21}} \searrow 0$ as $\delta \searrow 0$
% = \frac{(1-\varepsilon_B)}{(1+\eta)} (2\sqrt{\varepsilon'_B} + 2\sqrt{\varepsilon''_B}) $ with  $ \epsilon_{\scriptscriptstyle \widetilde{S}_{2}} \searrow 0 $ as $ \delta \searrow 0 $ and $ \varepsilon'_B = \varepsilon_B + 2\sqrt{\varepsilon_B} $ and $ \varepsilon''_B = 2\varepsilon' + 2\sqrt{\varepsilon'} $ 
(see (35) in \cite{wilde_e} for more details).
Finally, consider the last term.
% , using an exactly symmetric analysis as in $ J_3 $ term, but on the operators $ \bar{B}^{(\mu_2)}_{v^n} $ and $ {B}^{(\mu_2)}_{v^n}, $ we can make $ \EE[\widetilde{S}_{22}] $ arbitrarily small, for sufficiently large $ n, $ if again $ S_2 \geq I(V;RA)_{\sigma_{2}} - S(V)_{\sigma_{3}} + \log{|\mathcal{V}|}.$
To simplify this term, we appeal to Lemma \ref{lem:LemAandABar} in Section \ref{sec:p2pProof}. This gives us
\begin{align}
    \tilde{S}_{23} \leq \frac{2\; {2^{3n\delta}}}{N_2}\hspace{-3pt} \sum_{\mu_2=1}^{N_2}\!\!
  \left(\!\!H^B_0 \!+\! \frac{\sqrt{(1-\varepsilon_B)}}{(1+\eta)}\sqrt{H_1^B+H_2^B+H_3^B}\right)   ,
 \end{align}
 where 
\begin{align}
	H_0^B &\deq \left|\Delta^{(\mu_2)}_B - \EE[\Delta^{(\mu_2)}_B]\right|, \nonumber \\
	H_1^B &\deq \Tr{(\PiB-\CutOffB )\sum_{v^n}\lambdavB \rhotildvB}, \nonumber \\
	 H_2^B &\deq \left\|\sum_{v^n}\lambdavB \rhotildvB - (1-\varepsilon_B)\sum_{v^n}\frac{\beta_{v^n}\zetaCoeff}{\EE[\Delta^{(\mu)}_B]} \rhotildvB\right\|_1, \nonumber \\
% 	 \quad & \text{ and  }  \quad 
	H_3^B & \deq (1-\varepsilon_B)\left\|\sum_{v^n}\frac{\beta_{v^n}\zetaCoeff}{\Delta^{(\mu_2)}_B} \rhotildvB -  \sum_{v^n}\frac{\beta_{v^n}\zetaCoeff}{\EE[\Delta^{(\mu_2)}_B]} \rhotildvB\right\|_1,
	\end{align} 
	and $ \Delta^{(\mu)}_B \deq \sum_{v^n \in \TDeltan(V)}\beta_{v^n}\zetaCoeff $ and $\varepsilon_B = \sum_{v^n \notin \TDeltan(V)}\lambdavB$.

Further, using the simplification performed in \eqref{eq:Simplification_H1}, \eqref{eq:Simplification_H2}, and \eqref{eq:Simplification_H3}, and the concavity of the square-root function, we obtain,
\begin{widetext}
\begin{align}
    \EE[\tilde{S}_{23}] &\leq \frac{2}{N_2}2^{3n\delta_{\rho_B}}\sum_{\mu_2=1}^{N_2}\left(\EE[H_0^B] + {\frac{(1-\varepsilon_B)}{(1+\eta)}}\sqrt{\left(\frac{2^{2n\delta_{\rho_B}}}{\eta}+1\right)\EE[\widetilde{H}^B]} + \sqrt{\frac{(1-\varepsilon_B)}{(1+\eta)}}\sqrt{\EE[H_0^B]}\right),\nonumber\\
    \text{where }\; \widetilde{H}^B &\deq \left\|\frac{1}{(1-\varepsilon_B)}\sum_{v^n}\lambdavB \rhotildvB - \frac{p^n}{p^{k+l_2}} \sum_{v^n}\sum_{a_2,j>0} \frac{\lambdavB\rhotildvB}{(1-\varepsilon_B)} \11_{\{\VCodeword=v^n\}}\right\|_1.
\end{align}
\end{widetext}
% where
% \begin{align}
%     \widetilde{H}^B \deq \left\|\frac{1}{(1-\varepsilon_B)}\sum_{v^n}\lambdavB \rhotildvB - \frac{p^n}{p^{k+l_2}} \sum_{v^n}\sum_{a_2,j>0} \frac{\lambdavB\rhotildvB}{(1-\varepsilon_B)} \11_{\{\VCodeword=v^n\}}\right\|_1.
% \end{align}

Using Proposition \ref{prop:Lemma for p2p:S_13}, for any $\epsilon \in (0,1)$, any $\eta, \delta \in (0,1)$ sufficiently small, and any $n$ sufficiently large, we have
% for  all sufficiently small $\delta$ and sufficiently large $n$, we have 
$\EE\left[\tilde{S}_{23}\right] \leq \epsilon$ if $\frac{k+l_2}{n}\log p > I(V;RA)_{\sigma_2} +\log{p} - S(V)_{\sigma_3}   $, where $\sigma_2$, $\sigma_3$ are the auxiliary state defined in the statement of the theorem.
% and $\epsilon_{\tilde{S}_{22}}, \delta_{\tilde{S}_{22}} \searrow 0$ as $\delta \searrow 0$.

%Now, by using Markov inequality $\PP(E_3) \leq \sqrt{\epsilon_{\widetilde{S}_2}}$, where $E_3 \deq \{\widetilde{S}_{22} \geq \sqrt{\epsilon_{\widetilde{S}_2}}\}$. Hence, using union bound on the three events $E_1, E_2$ and $E_3$, $\widetilde{S}_2$ can be made arbitrarily small, for sufficiently large $n$, with high probability.

\noindent{\bf Analysis of $\widetilde{S}_3 $:} Due to the symmetry in $\widetilde{S}_2 $ and $ \widetilde{S}_3 $, the analysis of $ \widetilde{S}_{3} $ follows very similar arguments as that of $\widetilde{S}_2 $ and hence we obtain the following, for any $\epsilon \in (0,1)$, any $\eta, \delta \in (0,1)$ sufficiently small, and any $n$ sufficiently large, we have $\EE\left[\tilde{S}_{3}\right] \leq \epsilon$ if $S_1 > I(U;RB)_{\sigma_1} +\log{p} - S(U)_{\sigma_3} $, where $\sigma_1$, $\sigma_3$ are the auxiliary state defined in the statement of the theorem.

\noindent{\bf Analysis of $ \widetilde{S}_4$:} We have 
\begin{align}
%S_{4} & = \sum_{z^n}\left\|\frac{1}{N_1N_2}\sum_{\mu_1,\mu_2} \sqrt{\rho_{AB}^{\tensor n}}\left (\Gamma^{A, ( \mu_1)}_0\tensor \Gamma^{B, (\mu_2)}_0\right )\sqrt{\rho_{AB}^{\tensor n}} P^n_{Z|U+V}(z^n|w^{n}_{0})\right\|_1 \nonumber \\
\widetilde{S}_{4}\! & \leq  \frac{1}{N_1N_2}\sum_{\mu_1,\mu_2}\sum_{z^n}P^n_{Z|U+V}(z^n|w^{n}_{0})\nonumber \\
& \hspace{1in}\left\| \sqrt{\rho_{AB}^{\tensor n}}\left (\Gamma^{A, ( \mu_1)}_0\tensor \Gamma^{B, (\mu_2)}_0\right )\sqrt{\rho_{AB}^{\tensor n}} \right\|_1  \nonumber \\
& \leq  \frac{1}{N_1N_2}\sum_{\mu_1,\mu_2}\left\| \sqrt{\rho_{AB}^{\tensor n}}\left (\Gamma^{A, ( \mu_1)}_0\tensor I\right )\sqrt{\rho_{AB}^{\tensor n}} \right\|_1 \nonumber \\
& \hspace{5pt} + \! \frac{1}{N_1N_2}\sum_{\mu_1,\mu_2}\sum_{v^n}\left\| \sqrt{\rho_{AB}^{\tensor n}}\left (\Gamma^{A, ( \mu_1)}_0\tensor B_{v^n}^{(\mu_{2})}\right )\sqrt{\rho_{AB}^{\tensor n}} \right\|_1\!\!, \label{eq:S4inequalities}
%& \leq \frac{2}{N_1}\sum_{\mu_1}\left\| \sqrt{\rho_{A}^{\tensor n}}\left (\Gamma^{A, ( \mu_1)}_0\right )\sqrt{\rho_{A}^{\tensor n}} \right\|_1
\end{align}
where the inequalities above are obtained by a straight forward substitution and use of triangle inequality.
Further, since $ 0 \leq \Gamma^{A, ( \mu_1)}_0 \leq I $  and  $ 0 \leq \Gamma^{B, (\mu_2)}_0 \leq I $, this simplifies the first term in \eqref{eq:S4inequalities} as 
\begin{align}
\frac{1}{N_1N_2}\sum_{\mu_1,\mu_2}&\left\| \sqrt{\rho_{AB}^{\tensor n}}\left (\Gamma^{A, ( \mu_1)}_0\tensor I\right )\sqrt{\rho_{AB}^{\tensor n}} \right\|_1 \nonumber \\
& = \frac{1}{N_1}\sum_{\mu_1}\left\| \sqrt{\rho_{A}^{\tensor n}}\left (\Gamma^{A, ( \mu_1)}_0\right )\sqrt{\rho_{A}^{\tensor n}} \right\|_1. \nonumber
\end{align}
Similarly, the second term in \eqref{eq:S4inequalities} simplifies using Lemma \ref{lem:Separate} as
\begin{align}
\frac{1}{N_1N_2}\sum_{\mu_1,\mu_2}\sum_{v^n}&\left\| \sqrt{\rho_{AB}^{\tensor n}}\left (\Gamma^{A, ( \mu_1)}_0\tensor B_{v^n}^{(\mu_{2})}\right )\sqrt{\rho_{AB}^{\tensor n}} \right\|_1 \nonumber \\
&\leq \frac{1}{N_1}\sum_{\mu_1}\left\| \sqrt{\rho_{A}^{\tensor n}}\left (\Gamma^{A, ( \mu_1)}_0\right )\sqrt{\rho_{A}^{\tensor n}} \right\|_1.\nonumber
\end{align}
Using these simplifications, we have
\begin{align}
\widetilde{S}_4 & \leq \frac{2}{N_1}\sum_{\mu_1}\left\| \sqrt{\rho_{A}^{\tensor n}}\left (\Gamma^{A, ( \mu_1)}_0\right )\sqrt{\rho_{A}^{\tensor n}} \right\|_1. \nonumber
%  = \frac{1}{N_2}\sum_{\mu_2}\left\| \sum_{v^n}\lambdavB\rhohatvB - \sum_{v^n}\sqrt{\rho_{B}^{\tensor n}} B^{(\mu_2)}_{v^n}\sqrt{\rho_{B}^{\tensor n}} \right\|_1. \nonumber
\end{align}
The above expression is similar to the one obtained in the simplification of $\widetilde{S}_2 $ and hence we can bound $\widetilde{S}_4 $ using similar constraints as $ \widetilde{S}_2 $, for sufficiently large $ n $.

%%-----------------------------------------xxxxxxxxxx------------------------------------%%

\subsection{Proof of Proposition \ref{prop:Lemma for p2p:J_1}}
\label{appx:proof of J_1}

We start by applying triangle inequality to obtain $J_1 \leq J_{11} + J_{12}$, where
\begin{widetext}
\begin{align}
J_{11} & \deq  \sum_{z^n,v^n}\left\|\sum_{u^n}\sqrt{\rho_{AB}^{\tensor n}}  \left (\bar{\Lambda}^A_{u^n}\tensor \bar{\Lambda}^B_{v^n} -	\frac{1}{N_1}\sum_{\mu_1 =1}^{N_1}\frac{\alpha_{u^n}\gammaCoeff}{\lambdauA}\bar{\Lambda}^A_{u^n} \tensor  \bar{\Lambda}^B_{v^n} \right )\sqrt{\rho_{AB}^{\tensor n}} P^n_{Z|W}(z^n|u^n+ v^n)\right\|_1, \nonumber \\
J_{12} &\deq \sum_{z^n,v^n}\left\|\frac{1}{N_1}\sum_{\mu_1=1}^{N_1}\sum_{u^n}\sqrt{\rho_{AB}^{\tensor n}}  \left ( \frac{\alpha_{u^n}\gammaCoeff}{\lambdauA}\bar{\Lambda}^A_{u^n}\tensor \bar{\Lambda}^B_{v^n} -	\gammaCoeff  \bar{A}_{u^n}^{(\mu_1)} \tensor  \bar{\Lambda}^B_{v^n}\right )\sqrt{\rho_{AB}^{\tensor n}} P^n_{Z|W}(z^n|u^n+ v^n)\right\|_1, \nonumber 
\end{align}
\end{widetext}
Now with the intention of employing Lemma \ref{lem:nf_SoftCovering}, we express $J_{11}$ as
\begin{align}
J_{11} & = \left\|\sum_{u^n,v^n,z^n}\lambdaAB\rhohatuAvB \tensor \phi_{u^n,v^n,z^n} \right . \nonumber \\ 
& \hspace{25pt}\left.- \frac{1}{(1+\eta)}\frac{p^n}{p^{k+l_1}N_1}\sum_{\mu_1}\sum_{u^n,v^n,z^n}\sum_{a_1,i> 0}\lambda_{u^n}^A\right. \nonumber \\ 
& \hspace{20pt}\times\left.\mathbbm{1}_{\{\UCodeword = u^n\}}\frac{\lambdaAB}{\lambdauA}\rhohatuAvB \tensor \phi_{u^n,v^n,z^n}\right\|_1\!\!\!, \nonumber 
\end{align}
where the equality above is obtained by defining $ \phi_{u^n,v^n,z^n} = P^n_{Z|W}(z^n|u^n+ v^n)\ketbra{v^n}\tensor\ketbra{z^n} $ and using the definitions of $ \alpha_{u^n}, \gamma_{u^n}^{(\mu_1)} $ and $ \rhohatuAvB $, followed by using the triangle inequality for the block diagonal operators. Note that the triangle inequality in this case becomes an equality.

\noindent Let us define $ \mathcal{T}_{u^n} $ as
\begin{align}
\mathcal{T}_{u^n} \deq \sum_{v^n,z^n}\frac{\lambdaAB}{\lambdauA}\rhohatuAvB \tensor \phi_{u^n,v^n,z^n}. \nonumber
\end{align} 
Note that in the above definition of $ \mathcal{T}_{u^n} $ we have $ \mathcal{T}_{u^n} \geq 0$ and $\Tr{\mathcal{T}_{u^n}} =1 $ for all $u^n \in \FF_p^n$. Further, it contains all the elements in product form, and thus can be written as $ \mathcal{T}_{u^n} = \bigtensor_{i=1}^{n}\mathcal{T}_{u_i}. $
This simplifies $ J_{11} $ as
\begin{align}
%  & = \left\|\sum_{u^n}\lambdauA\mathcal{T}_{u^n} -  \frac{(1-\varepsilon)}{(1+\eta)}\frac{1}{2^{n(S_1+C_1)}}\sum_{\mu_1,l}\sum_{u^n}\IndiU1 \mathcal{T}_{u^n} \right\|_1 \nonumber \\
J_{11} &= \bigg\|\sum_{u^n}\lambdauA\mathcal{T}_{u^n} - \frac{1}{(1+\eta)}\frac{p^n}{p^{k+l_1}}\frac{1}{N_1}\sum_{\mu_1} \nonumber \\ &  \hspace{35pt}\sum_{u^n}\sum_{a_1,i>0} \lambda^A_{u^n} \mathcal{T}_{u^n}\11_{\{\UCodeword=u^n\}} \bigg\|_1. \nonumber%\label{eq:J1} 
\end{align}
Using Lemma \ref{lem:nf_SoftCovering}, we claim the following: for any $\epsilon \in (0,1)$, any $\eta, \delta \in (0,1)$ sufficiently small, and any $n$ sufficiently large, we have 
% for any $\epsilon_2 \ >0$ the term in \eqref{eq:ensemble} 
$\EE[J_{11}] \leq \epsilon$, if 
$\frac{k+l_1}{n}\log{p} + \frac{1}{n}\log{N_1} > I(U;RZV)_{\sigma_{3}} - S(U)_{\sigma_{3}} + \log{p}$,  where $\sigma_3$ is the auxiliary state defined in the statement of the theorem.

Now we consider the term corresponding to $ J_{12} $ and prove that its expectation with respect to the Alice's codebook is small. Recalling $ J_{12} $, we get
\begin{widetext}
\begin{align}
%J_{12} 
%&= \sum_{z^n,v^n}\left\|\sum_{u^n}\sqrt{\rho_{AB}^{\tensor n}}  \left (\frac{1}{N_1}\sum_{\mu_1=1}^{N_1} \frac{\gamma_{u^n}^{(\mu_1)}}{\lambdauA}\bar{\Lambda}^A_{u^n}\tensor \bar{\Lambda}^B_{v^n} -\frac{1}{N_1} \sum_{\mu_1=1}^{N_1} 	  A_{u^n}^{(\mu_1)} \tensor  \bar{\Lambda}^B_{v^n}\right )\sqrt{\rho_{AB}^{\tensor n}} P^n_{Z|W}(z^n|u^n+ v^n)\right\|_1 \nonumber \\
J_{12}  & \leq \frac{1}{N_1}\sum_{\mu_1=1}^{N_1}\sum_{u^n,v^n}\sum_{z^n} P^n_{Z|W}(z^n|u^n+ v^n) \left\|\sqrt{\rho_{AB}^{\tensor n}}  \left ( \frac{\alpha_{u^n}\gamma_{u^n}^{(\mu_1)}}{\lambdauA}\bar{\Lambda}^A_{u^n}\tensor \bar{\Lambda}^B_{v^n} -	 \gammaCoeff \bar{A}_{u^n}^{(\mu_1)} \tensor  \bar{\Lambda}^B_{v^n}\right )\sqrt{\rho_{AB}^{\tensor n}} \right\|_1, \nonumber \\
& = \frac{1}{N_1}\sum_{\mu_1=1}^{N_1}\sum_{u^n,v^n}\alpha_{u^n}\gammaCoeff \left\|\sqrt{\rho_{AB}^{\tensor n}}  \left ( \left ( \frac{1}{\lambdauA}\bar{\Lambda}^A_{u^n}-\sqrt{\rho_{A}^{\tensor n}}^{-1}\rhotilduA\sqrt{\rho_{A}^{\tensor n}}^{-1}\right )\tensor \bar{\Lambda}^B_{v^n} \right )\sqrt{\rho_{AB}^{\tensor n}} \right\|_1, \nonumber
\end{align}
\end{widetext}
where the inequality is obtained by using triangle and the next equality follows from the fact that $ \sum_{z^n} P^n_{Z|W}(z^n|u^n+ v^n) =1  $ for all $ u^n \in \mathcal{U}^n $ and $ v^n \in \mathcal{V}^n $ and using the definition of $ A_{u^n}^{(\mu_1)} $.
By applying expectation of $ J_{12}  $ over  the Alice's codebook, we get
% \begin{widetext}
\begin{align}
%\EE{\left[J_{12} \right]} & \leq \frac{1}{N_1}\sum_{\mu_1=1}^{N_1}\sum_{u^n,v^n}\EE\left[\gamma_{u^n}^{(\mu_1)}\right] \left\|\sqrt{\rho_{AB}^{\tensor n}}  \left ( \left ( \frac{1}{\lambdauA}\bar{\Lambda}^A_{u^n}-\sqrt{\rho_{A}^{\tensor n}}^{-1}\LambdauA\sqrt{\rho_{A}^{\tensor n}}^{-1}\right )\tensor \bar{\Lambda}^B_{v^n} \right )\sqrt{\rho_{AB}^{\tensor n}} \right\|_1 \nonumber \\
% \EE{\left[J_{12} \right]} &\leq \frac{(1-\varepsilon)}{(1+\eta)}\!\!\!\sum_{\substack{u^n \in \TDeltaN(A) ,\\v^n}}\!\!\frac{\lambdauA}{1-\varepsilon} \left\|\sqrt{\rho_{AB}^{\tensor n}}  \left ( \left ( \frac{1}{\lambdauA}\bar{\Lambda}^A_{u^n}-\sqrt{\rho_{A}^{\tensor n}}^{-1}\LambdauA\sqrt{\rho_{A}^{\tensor n}}^{-1}\right )\tensor \bar{\Lambda}^B_{v^n} \right )\sqrt{\rho_{AB}^{\tensor n}} \right\|_1 \nonumber \\
 \EE{\left[J_{12} \right]} & \leq \frac{1}{(1+\eta)}\sum_{\substack{u^n }}\lambdauA\sum_{v^n} \left\|\sqrt{\rho_{AB}^{\tensor n}}  \left ( \left ( \frac{1}{\lambdauA}\bar{\Lambda}^A_{u^n}- \right.\right.\right. \nonumber \\ 
 & \hspace{20pt} \left.\left.\left.\sqrt{\rho_{A}^{\tensor n}}^{-1}\rhotilduA\sqrt{\rho_{A}^{\tensor n}}^{-1}\right )\tensor \bar{\Lambda}^B_{v^n} \right )\sqrt{\rho_{AB}^{\tensor n}} \right\|_1, \nonumber
%  & = \frac{2^{-nS_1}}{(1+\eta)}\!\!\!\sum_{\substack{u^n \in \TDeltaN(A) ,\\v^n}}\!\!\lambdauA \left\|\Tr_{AB}\left\{\left ( \text{id}_{R'^n}\tensor\left ( \frac{1}{\lambdauA}\bar{\Lambda}^A_{u^n}-\sqrt{\rho_{A}^{\tensor n}}^{-1}\LambdauA\sqrt{\rho_{A}^{\tensor n}}^{-1}\right )\tensor \bar{\Lambda}^B_{v^n} \right )\Psi_{R'^nA^nB^n}\right\} \right\|_1 \nonumber 
\end{align}
% \end{widetext}
where we have used the fact that $ \EE{[\alpha_{u^n} \gamma_{u^n}^{(\mu_1)}]} = \frac{\lambdauA}{(1+\eta)} $. 
To simplify the above equation, we employ Lemma \ref{lem:Separate} which completely discards the effect of Bob's measurement. 
 Since $ \sum_{v^n}\bar{\Lambda}^B_{v^n} = I$, from Lemma \ref{lem:Separate} we have for every $ u^n$,
\begin{align}
	& \sum_{v^n}\left\|\sqrt{\rho_{AB}^{\tensor n}}  \left ( \left ( \frac{1}{\lambdauA}\bar{\Lambda}^A_{u^n}- \right.\right.\right.\nonumber \\
	& \hspace{20pt} \left.\left.\left. \sqrt{\rho_{A}^{\tensor n}}^{-1}\rhotilduA\sqrt{\rho_{A}^{\tensor n}}^{-1}\right )\tensor \bar{\Lambda}^B_{v^n} \right )\sqrt{\rho_{AB}^{\tensor n}} \right\|_1 \nonumber \\ 
	& \hspace{4pt} = \left\|\sqrt{\rho_{A}^{\tensor n}}  \left ( \frac{1}{\lambdauA}\bar{\Lambda}^A_{u^n}-\sqrt{\rho_{A}^{\tensor n}}^{-1}\rhotilduA\sqrt{\rho_{A}^{\tensor n}}^{-1} \right )\sqrt{\rho_{A}^{\tensor n}} \right\|_1.\nonumber
\end{align}
This simplifies $ \EE{\left[J_{12} \right]} $ as 
\begin{align}
\EE{\left[J_{12} \right]}  & \leq  \frac{1}{(1+\eta)}\sum_{\substack{u^n  }}\lambdauA\left\|\sqrt{\rho_{A}^{\tensor n}}  \left ( \frac{1}{\lambdauA}\bar{\Lambda}^A_{u^n} - \right. \right. \nonumber\\
&\hspace{1in}\left.\left.\sqrt{\rho_{A}^{\tensor n}}^{-1}\rhotilduA\sqrt{\rho_{A}^{\tensor n}}^{-1} \right )\sqrt{\rho_{A}^{\tensor n}} \right\|_1 \nonumber \\
& \leq  \frac{1}{(1+\eta)}\!\!\!\sum_{\substack{u^n \notin \TDeltaN(U) }}\!\!\!\!\lambdauA\left\|\rhohatuA\right\|_1 + \nonumber \\
& \hspace{0.5in} \frac{1}{(1+\eta)}\!\sum_{\substack{u^n \in \TDeltaN(U) }}\!\!\lambdauA\left\| \left (\rhohatuA -\rhotilduA \right ) \right\|_1 \nonumber \\
& \leq \varepsilon_{A}+ \epsilon_{\scriptscriptstyle J_{12}}' %= \epsilon_{J_{12}}\nonumber,
\end{align}
where the last inequality is obtained by repeated usage of the Average Gentle Measurement Lemma 
and $ \epsilon_{ J_{12}}' \searrow 0$ as $\delta \searrow 0$ (see (35) in \cite{wilde_e} for details). This completes the proof.
% = \frac{(1-\varepsilon)}{(1+\eta)} (2\sqrt{\varepsilon'_A} + 2\sqrt{\varepsilon''_A}) $ with  $ \epsilon_{\scriptscriptstyle J_2} \searrow 0 $ as $ \delta \searrow 0 $ and $ \varepsilon'_A = \varepsilon + 2\sqrt{\varepsilon} $ and $ \varepsilon''_A = 2\varepsilon + 2\sqrt{\varepsilon} $ 

\subsection{Proof of Proposition \ref{prop:Lemma for J_2}}\label{appx:proof of J_2}
Noting the similarity between $J_2$ and the term $\tilde{S}_2$ defined in the proof of Theorem \ref{thm:p2pTheorem} (see Section \ref{sec:p2pProof}), we begin by further simplifying $J_2$ using Lemma \ref{lem:LemAandABar}. This gives us
\begin{align}
    J_2 \leq \frac{2{2^{3n\delta_{\rho_A}}}}{N_1} \!\! \sum_{\mu_1=1}^{N_1}\!\!
 \!\left(\!H_0^A + \frac{\sqrt{(1-\varepsilon_A)}}{(1+\eta)}\sqrt{H_1^A+H_2^A+H_3^A}\right)  \!\!,
 \end{align}
 where 
\begin{align}
	H_0^A &\deq \left|\Delta^{(\mu_1)}_A - \EE[\Delta^{(\mu_1)}_A]\right|, \nonumber \\
	H_1^A &\deq \Tr{(\PiA-\CutOffA )\sum_{w^n}\lambdauA \rhotilduA}, \nonumber \\
	 H_2^A &\deq \|\sum_{u^n}\lambdauA \rhotilduA - (1-\varepsilon_A) \sum_{u^n}\frac{\alpha_{u^n}\gammaCoeff}{\EE[\Delta^{(\mu_1)}_A]} \rhotilduA\|_1,\nonumber \\
% 	 \quad & \text{ and  }  \quad 
	H_3^A & \deq (1-\varepsilon_A)\|\sum_{u^n}\frac{\alpha_{u^n}\gammaCoeff}{\Delta^{(\mu_1)}_A} \rhotilduA -  \sum_{u^n}\frac{\alpha_{u^n}\gammaCoeff}{\EE[\Delta^{(\mu_1)}_A]} \rhotilduA\|_1, \nonumber
\end{align}
	and $ \Delta^{(\mu_1)}_A \deq \sum_{u^n\in\TDeltan(U)}\alpha_{u^n}\gammaCoeff, $ $\varepsilon_A \deq \sum_{u^n \notin \TDeltan(U)}\lambdauA$, and $\zeroWithDelta{\delta_{\rho_A}}$.
Further, using the simplification performed in \eqref{eq:Simplification_H1}, \eqref{eq:Simplification_H2}, and \eqref{eq:Simplification_H3}, and the concavity of the square-root function, we obtain,
\begin{align}
    \EE[J_2] &\leq \frac{2}{N_1}2^{3n\delta_{\rho_A}}\sum_{\mu_1=1}^{N_1}\Bigg(\EE[H_0^A] + \frac{(1-\varepsilon_A)}{(1+\eta)}\nonumber \\ &\times\sqrt{\left(\frac{2^{2n\delta_{\rho_A}}}{\eta}+1\right)\EE[\widetilde{H}^A]} + \sqrt{\frac{(1-\varepsilon_A)}{(1+\eta)}}\sqrt{\EE[H_0^A]}\Bigg),\nonumber
\end{align}
where
\begin{align}
    \widetilde{H}^A &\deq \bigg\|\sum_{u^n}\frac{\lambdauA}{(1-\varepsilon_A)} \rhotilduA - \nonumber \\& \hspace{10pt}\frac{p^n}{2^{nS_1}} \sum_{u^n}\sum_{a_1,i>0} \frac{\lambdauA}{(1-\varepsilon_A)}\rhotilduA \11_{\{\UCodeword=u^n\}}\bigg\|_1.\nonumber
\end{align}
The proof from here follows from Proposition \ref{prop:Lemma for p2p:S_13}.
\subsection{Proof of Proposition \ref{prop:Lemma for Q2}}\label{appx:proof of Q2} 
We start by adding and subtracting the following terms within $ Q_2 $ 
\begin{align}
(i) &\sum_{u^n,v^n}\sqrt{\rho_{AB}^{\tensor n}}  \left(\bar{\Lambda}^A_{u^n}\tensor\bar{\Lambda}^B_{v^n} \right )\sqrt{\rho_{AB}^{\tensor n}} P^n_{Z|W}(z^n|u^n+ v^n), \nonumber \\
(ii)&\sum_{u^n,v^n}\frac{1}{N_2}\sum_{\mu_2=1}^{N_2}\sqrt{\rho_{AB}^{\tensor n}}  \left( \bar{\Lambda}^A_{u^n} \tensor \cfrac{\beta_{v^n}\zeta^{(\mu_2)}_{v^n}}{\lambdavB}\bar{\Lambda}^B_{v^n} \right )\nonumber \\ 
& \hspace{1.35in}\times\sqrt{\rho_{AB}^{\tensor n}} P^n_{Z|W}(z^n|u^n+ v^n), \nonumber \\
(iii)&\sum_{u^n,v^n}  \frac{1}{N_1N_2} \sum_{\mu_1,\mu_2}\sqrt{\rho_{AB}^{\tensor n}}\left( \!\!\gammaCoeff A_{u^n}^{(\mu_1)} \tensor \cfrac{\beta_{v^n}\zeta^{(\mu_2)}_{v^n}}{\lambdavB}\bar{\Lambda}^B_{v^n}\!\! \right )\nonumber \\ 
& \hspace{1.35in}\times\sqrt{\rho_{AB}^{\tensor n}} P^n_{Z|W}(z^n|u^n+ v^n),  \nonumber \\
(iv)&\sum_{u^n,v^n}  \frac{1}{N_1N_2} \sum_{\mu_1,\mu_2}\sqrt{\rho_{AB}^{\tensor n}}\left(  \gammaCoeff A_{u^n}^{(\mu_1)} \tensor {\zeta^{(\mu_2)}_{v^n}} \bar{B}_{v^n}^{(\mu_2)} \right )\nonumber \\ 
& \hspace{1.35in}\times\sqrt{\rho_{AB}^{\tensor n}} P^n_{Z|W}(z^n|u^n+ v^n). \nonumber
\end{align}
This gives us $ Q_2 \leq Q_{21} + Q_{22} + Q_{23} + Q_{24} + Q_{25} $, where
\begin{widetext}
\begin{align}
Q_{21} & \deq \sum_{z^n} \left\|\sum_{u^n,v^n}\sqrt{\rho_{AB}^{\tensor n}}  \left ( \left (\frac{1}{N_1} \sum_{\mu_1=1}^{N_1}\gammaCoeff A_{u^n}^{(\mu_1)}\right ) \tensor \bar{\Lambda}^B_{v^n} -\bar{\Lambda}^A_{u^n}\tensor\bar{\Lambda}^B_{v^n} \right )\sqrt{\rho_{AB}^{\tensor n}} P^n_{Z|W}(z^n|u^n+ v^n)\right\|_1, \nonumber\\
Q_{22} & \deq \sum_{z^n} \left\| \sum_{u^n,v^n}\sqrt{\rho_{AB}^{\tensor n}}  \left (\bar{\Lambda}^A_{u^n}\tensor\bar{\Lambda}^B_{v^n}  - \bar{\Lambda}^A_{u^n} \tensor\left (\frac{1}{N_2}\sum_{\mu_2=1}^{N_2} \cfrac{\beta_{v^n}\zeta^{(\mu_2)}_{v^n}}{\lambdavB}\bar{\Lambda}^B_{v^n}\right )\right )\sqrt{\rho_{AB}^{\tensor n}} P^n_{Z|W}(z^n|u^n+ v^n)\right\|_1, \nonumber	\\
Q_{23} & \deq \sum_{z^n} \left\| \sum_{u^n,v^n}\!\sqrt{\rho_{AB}^{\tensor n}}  \left( \left(\bar{\Lambda}^A_{u^n}- \frac{1}{N_1} \sum_{\mu_1}\gammaCoeff A_{u^n}^{(\mu_1)} \right) \tensor\left (\frac{1}{N_2}\!\sum_{\mu_2} \cfrac{\beta_{v^n}\zeta^{(\mu_2)}_{v^n}}{\lambdavB}\bar{\Lambda}^B_{v^n}\right ) \right)\sqrt{\rho_{AB}^{\tensor n}} P^n_{Z|W}(z^n|u^n+ v^n)\right\|_1, \nonumber\\
Q_{24} & \deq \sum_{z^n} \left\| \sum_{u^n,v^n}\frac{1}{N_1N_2} \sum_{\mu_1,\mu_2}\sqrt{\rho_{AB}^{\tensor n}}  \left( \gammaCoeff A_{u^n}^{(\mu_1)} \tensor \left(\cfrac{\beta_{v^n}\zeta^{(\mu_2)}_{v^n}}{\lambdavB}\bar{\Lambda}^B_{v^n} - \zetaCoeff\bar{B}_{v^n}^{(\mu_2)}\right) \right)\sqrt{\rho_{AB}^{\tensor n}} P^n_{Z|W}(z^n|u^n+ v^n)\right\|_1, \nonumber \\ 
Q_{25} & \deq \sum_{z^n} \left\| \sum_{u^n,v^n}\frac{1}{N_1N_2} \sum_{\mu_1,\mu_2}\sqrt{\rho_{AB}^{\tensor n}}  \left( \gammaCoeff A_{u^n}^{(\mu_1)} \tensor  \left(\zetaCoeff\bar{B}_{v^n}^{(\mu_2)} -  \zetaCoeff B_{v^n}^{(\mu_2)}\right) \right)\sqrt{\rho_{AB}^{\tensor n}} P^n_{Z|W}(z^n|u^n+ v^n)\right\|_1. \nonumber
\end{align}
\end{widetext}
We start by analyzing $ Q_{21} $. Note that $ Q_{21} $ is exactly same as $ Q_{1} $ and hence using the same rate constraints as $ Q_{1} $, this term can be bounded. Next, consider $ Q_{22} $. Substitution of $ \zeta^{(\mu_2)}_{v^n} $ gives
% \begin{widetext}
\begin{align}
%Q_{22}  & = \sum_{z^n} \left\| \sum_{u^n,v^n}\sqrt{\rho_{AB}^{\tensor n}}  \left (\bar{\Lambda}^A_{u^n}\tensor\bar{\Lambda}^B_{v^n} \right)\sqrt{\rho_{AB}^{\tensor n}}P^n_{Z|W}(z^n|u^n+ v^n) \right .\nonumber \\ 
%& \hspace{5pt} \left . - \cfrac{(1-\varepsilon')}{(1+\eta)}\sum_{u^n,v^n} \sqrt{\rho_{AB}^{\tensor n}}\left( \bar{\Lambda}^A_{u^n} \tensor\left (\frac{2^{-n\tilde{R}_2}}{N_2}\sum_{\mu_2=1}^{N_2} \cfrac{|\{k: V^{n,(\mu_2)}(k)=v^n\}|}{\lambdavB}\bar{\Lambda}^B_{v^n}\right )\right )\sqrt{\rho_{AB}^{\tensor n}} P^n_{Z|W}(z^n|u^n+ v^n)\right\|_1\ \nonumber \\
%& = \sum_{z^n} \left\| \sum_{u^n,v^n}\sqrt{\rho_{AB}^{\tensor n}}  \left (\bar{\Lambda}^A_{u^n}\tensor\bar{\Lambda}^B_{v^n} \right)\sqrt{\rho_{AB}^{\tensor n}}P^n_{Z|W}(z^n|u^n+ v^n) \right .\nonumber \\ 
%& \hspace{20pt} \left . - \cfrac{(1-\varepsilon')}{(1+\eta)}\cfrac{1}{2^{n(\tilde{R}_2 + C_2)}}\sum_{\mu_2,k} \sum_{u^n,v^n}\cfrac{\IndiV1}{\lambdavB} \sqrt{\rho_{AB}^{\tensor n}}\left(\bar{\Lambda}^A_{u^n} \tensor \bar{\Lambda}^B_{v^n}\right )\sqrt{\rho_{AB}^{\tensor n}} P^n_{Z|W}(z^n|u^n+ v^n)\right\|_1\ \nonumber \\
Q_{22} & = \bigg\| \sum_{u^n,v^n,z^n}\lambdaAB\rhohatuAvB \tensor \psi_{u^n,v^n,z^n}  \nonumber \\
& \hspace{20pt} - \frac{1}{N_2}\sum_{\mu_2}\! \sum_{u^n,v^n,z^n}\!\!\!\!\!\beta_{v^n}\!\!\!\sum_{a_2,j>0}\!\!\! \mathbbm{1}_{\{\VCodeword = v^n\}}\nonumber \\
& \hspace{1.15in}\times \frac{\lambdaAB}{\lambdavB} \rhohatuAvB \tensor\psi_{u^n,v^n,z^n}\bigg\|_1, \nonumber
\end{align}
% \end{widetext}
where $ \psi_{u^n,v^n,z^n} $ is defined as $ \psi_{u^n,v^n,z^n} =  P^n_{Z|W}(z^n|u^n+ v^n)\ketbra{z^n}, $ and the equality uses the triangle inequality for block operators.
Now we use Lemma \ref{lem:nf_SoftCovering} to bound $Q_{22}$. Let
% we use similar set of arguments as used to bound $ J_1 $. For this let us define $ \mathcal{T}_{v^n} $ as
\begin{align}
 \mathcal{T}_{v^n} \deq \sum_{u^n,z^n}\cfrac{\lambdaAB}{\lambdavB} \rhohatuAvB \tensor \psi_{u^n,v^n,z^n}. \nonumber
\end{align}
Note that $ \mathcal{T}_{v^n} $ can be written in tensor product form as
$ \mathcal{T}_{v^n} = \bigtensor_{i=1}^{n}\mathcal{T}_{v_{i}} $. This simplifies $ Q_{22} $ as
\begin{align}
Q_{22} & =  \bigg\| \sum_{v^n}\lambdavB\mathcal{T}_{v^n}  - \cfrac{1}{(1+\eta)}\cfrac{p^n}{2^{nS_2}N_2}\sum_{\mu_2}\sum_{v^n} \nonumber \\
& \hspace{1in}\sum_{a_2,j>0}\lambda_{v^n}^B \mathcal{T}_{v^n}\11_{\{\VCodeword = v^n\}}\bigg\|_1.\nonumber
\end{align}
% The above equation is very similar to what we have for $ J_{1} $ in \eqref{eq:J1} and hence using the same arguments we claim that 
Application of Lemma \ref{lem:nf_SoftCovering} gives the following: for any $\epsilon \in (0,1)$, any $\eta, \delta \in (0,1)$ sufficiently small, and any $n$ sufficiently large, we have $\EE[Q_{22}] \leq \epsilon$ if
% us functions, $ \epsilon_{\scriptscriptstyle Q_{22}}(\delta), \delta_{\scriptscriptstyle Q_{22}}(\delta) \in (0,1)$, such that if  
\begin{align}
\frac{k+l_2}{n}\log{p}+\frac{1}{n}\log{N_2} > I(V;RZ)_{\sigma_3} - S(V)_{\sigma_{3}} + \log{p}.\nonumber
\end{align}
% then $ \EE[Q_{22}] \leq \epsilon_{\scriptscriptstyle Q_{22}} $ where $ \epsilon_{\scriptscriptstyle Q_{22}}, \delta_{\scriptscriptstyle Q_{22}} \searrow 0 $ as $ \delta \searrow 0. $
%, where $ \sigma_4 = \sum_{u,v,z}\lambda_{u,v}\hat{\rho}_{u,v}^{AB} \tensor P_{Z|U,V}(z|u,v)\ketbra{z}\tensor \ketbra{v},$ (as defined in the theorem statement).

Now, we move on to consider $ Q_{23} $. Taking expectation with respect $G, h_1^{(\mu_1)}, h_2^{(\mu_2)}$  gives
%\begin{align}
%Q_{23} &\leq \sum_{z^n, v^n}\left|\frac{1}{N_2}\sum_{\mu_2=1}^{N_2} \cfrac{\zeta^{(\mu_2)}_{v^n}}{\lambdavB}\right| \left\| \sum_{u^n}\sqrt{\rho_{AB}^{\tensor n}} \left (\bar{\Lambda}^A_{u^n} \tensor\bar{\Lambda}^B_{v^n}\right)\sqrt{\rho_{AB}^{\tensor n}}P^n_{Z|W}(z^n|u^n+ v^n)   \right . \nonumber \\ 
%& \hspace{100pt} \left. - \sum_{u^n}\sqrt{\rho_{AB}^{\tensor n}} \left(\frac{1}{N_1} \sum_{\mu_1}A_{u^n}^{(\mu_1)} \tensor \bar{\Lambda}^B_{v^n} \right)\sqrt{\rho_{AB}^{\tensor n}} P^n_{Z|W}(z^n|u^n+ v^n)\right\|_1 \nonumber\\
%\end{align}
%
%%Further 
%\begin{align}
%\EE[Q_{23}] & \leq \EE_{\mathbbm{C}}\left[ \sum_{z^n,v^n}\frac{1}{N_2}\!\sum_{\mu_2} \cfrac{\beta_{v^n}\zeta^{(\mu_2)}_{v^n}}{\lambdavB}\left\| \sum_{u^n}\!\sqrt{\rho_{AB}^{\tensor n}}  \left( \left(\bar{\Lambda}^A_{u^n}- \frac{1}{N_1} \sum_{\mu_1}\!\!\gammaCoeff A_{u^n}^{(\mu_1)} \right) \tensor
%\bar{\Lambda}^B_{v^n} \right)\sqrt{\rho_{AB}^{\tensor n}} P^n_{Z|W}(z^n|u^n+ v^n)\right\|_1\right], \nonumber\\
%\end{align}
\begin{widetext}
\begin{align}
\EE\left[Q_{23}\right ] &\leq \EE\left[\sum_{z^n, v^n}\frac{1}{N_2}\sum_{\mu_2=1}^{N_2} \cfrac{\beta_{v^n}\zeta^{(\mu_2)}_{v^n}}{\lambdavB} \left\| \sum_{u^n}\sqrt{\rho_{AB}^{\tensor n}} \left (\bar{\Lambda}^A_{u^n} \tensor\bar{\Lambda}^B_{v^n}\right)\sqrt{\rho_{AB}^{\tensor n}}P^n_{Z|W}(z^n|u^n+ v^n)   \right .\right . \nonumber \\ 
& \hspace{130pt} \left.\left . - \sum_{u^n}\sqrt{\rho_{AB}^{\tensor n}} \left(\frac{1}{N_1} \sum_{\mu_1}\gammaCoeff A_{u^n}^{(\mu_1)} \tensor \bar{\Lambda}^B_{v^n} \right)\sqrt{\rho_{AB}^{\tensor n}} P^n_{Z|W}(z^n|u^n+ v^n)\right\|_1\right] \nonumber\\
&= \EE_{G,h_1}\left[\sum_{z^n, v^n}\frac{1}{N_2}\sum_{\mu_2=1}^{N_2} \cfrac{\EE_{h_2|G}\left[\beta_{v^n}\zeta^{(\mu_2)}_{v^n}|G\right ]}{\lambdavB} \left\| \sum_{u^n}\sqrt{\rho_{AB}^{\tensor n}} \left (\bar{\Lambda}^A_{u^n} \tensor\bar{\Lambda}^B_{v^n}\right)\sqrt{\rho_{AB}^{\tensor n}}P^n_{Z|W}(z^n|u^n+ v^n)   \right . \right .\nonumber \\ 
& \hspace{130pt} \left .\left. - \sum_{u^n}\sqrt{\rho_{AB}^{\tensor n}} \left(\frac{1}{N_1} \sum_{\mu_1}\gammaCoeff A_{u^n}^{(\mu_1)} \tensor \bar{\Lambda}^B_{v^n} \right)\sqrt{\rho_{AB}^{\tensor n}} P^n_{Z|W}(z^n|u^n+ v^n)\right\|_1\right] \nonumber\\
&= \EE_{G,h_1}\left[\sum_{z^n, v^n}\cfrac{1}{(1+\eta)} \left\| \sum_{u^n}\sqrt{\rho_{AB}^{\tensor n}} \left (\bar{\Lambda}^A_{u^n} \tensor\bar{\Lambda}^B_{v^n}\right)\sqrt{\rho_{AB}^{\tensor n}}P^n_{Z|W}(z^n|u^n+ v^n)   \right . \right .\nonumber \\ 
& \hspace{130pt} \left .\left. - \sum_{u^n}\sqrt{\rho_{AB}^{\tensor n}} \left(\frac{1}{N_1} \sum_{\mu_1}\gammaCoeff A_{u^n}^{(\mu_1)} \tensor \bar{\Lambda}^B_{v^n} \right)\sqrt{\rho_{AB}^{\tensor n}} P^n_{Z|W}(z^n|u^n+ v^n)\right\|_1\right] \nonumber\\
&= \EE\left[\cfrac{J}{(1+\eta)}\right], \nonumber
\end{align}
\end{widetext}
where the inequality above is obtained by using the triangle inequality, and the first equality follows from $h_1^{(\mu_1)}$ and $h_2^{(\mu_2)}$ being generated independently. 
%The second   then using the fact that $ \EE_{\mathbbm{C}_2}\left[\zeta^{(\mu_2)}_{v^n}\right] = \cfrac{\lambdavB}{(1-\varepsilon')} $ and 
The last equality follows from the definition of $ J $ as in \eqref{def:J}. Hence, we use the result obtained in bounding $ \EE[J]. $
Next, we consider $ Q_{24} $.
\begin{widetext}
\begin{align}
Q_{24} & \leq \sum_{u^n,v^n}\sum_{z^n} P^n_{Z|W}(z^n|u^n+ v^n)\left\| \frac{1}{N_1N_2} \sum_{\mu_1,\mu_2}\sqrt{\rho_{AB}^{\tensor n}}  \left(\gammaCoeff  A_{u^n}^{(\mu_1)} \tensor \cfrac{\beta_{v^n} \zeta^{(\mu_2)}_{v^n}}{\lambdavB}\bar{\Lambda}^B_{v^n}\right )\sqrt{\rho_{AB}^{\tensor n}} \right .\nonumber \\
& \hspace{70pt} - \left .\frac{1}{N_1N_2} \sum_{\mu_1,\mu_2}\sqrt{\rho_{AB}^{\tensor n}}\left( \gammaCoeff A_{u^n}^{(\mu_1)} \tensor \beta_{v^n} \zeta^{(\mu_2)}_{v^n} \left (\sqrt{\rho_{B}}^{-1}\rhotildvB\sqrt{\rho_{B}}^{-1}\right ) \right)\sqrt{\rho_{AB}^{\tensor n}} \right\|_1 \nonumber \\
& \leq \frac{1}{N_2} \sum_{\mu_2}\sum_{u^n,v^n}\beta_{v^n} \zeta^{(\mu_2)}_{v^n}\left\| \sqrt{\rho_{AB}^{\tensor n}}  \left( \frac{1}{N_1} \sum_{\mu_1}\gammaCoeff A_{u^n}^{(\mu_1)} \tensor \cfrac{1}{\lambdavB}\bar{\Lambda}^B_{v^n}\right )\sqrt{\rho_{AB}^{\tensor n}} \right .\nonumber \\
& \hspace{70pt} - \left .\sqrt{\rho_{AB}^{\tensor n}}\left( \frac{1}{N_1} \sum_{\mu_1}\gammaCoeff A_{u^n}^{(\mu_1)} \tensor  \left (\sqrt{\rho_{B}}^{-1}\rhotildvB\sqrt{\rho_{B}}^{-1}\right ) \right)\sqrt{\rho_{AB}^{\tensor n}} \right\|_1, \nonumber
\end{align}
\end{widetext}
where the inequalities follow from the definition of $ \bar{B}_{v^n}^{(\mu_2)} $  and using multiple triangle inequalities.
% For ease of notation, let us define  $ \bar{A}_{u^n} $  as 
%\begin{align}
% \bar{A}_{u^n} = \cfrac{1}{N_1} \sum_{\mu_1}A_{u^n}^{(\mu_1)}  \nonumber
%\end{align}
Taking expectation of $ Q_{24} $ with respect to $h_2^{(\mu_2)}$, we get
% and using the fact that $ \EE_{\mathbbm{C}_2}{\left[\zeta^{(\mu_2)}_{v^n}\right]} = \cfrac{\lambdavB}{(1+\eta)}$ gives 
\begin{widetext}
\begin{align}
\EE\left[Q_{24}\right]& 
\leq \EE_{G,h_1}\left[\sum_{\substack{u^n,v^n}}\cfrac{\lambdavB}{(1+\eta)}\Bigg\| \sqrt{\rho_{AB}^{\tensor n}}  \left( \frac{1}{N_1} \sum_{\mu_1}\gammaCoeff A_{u^n}^{(\mu_1)} \tensor \left(\cfrac{1}{\lambdavB}\bar{\Lambda}^B_{v^n} - \sqrt{\rho_{B}}^{-1}\rhotildvB\sqrt{\rho_{B}}^{-1} \right)\right)\sqrt{\rho_{AB}^{\tensor n}}\right] \nonumber \\
%\EE_{\mathbbm{C}_2}{\left[Q_{24}\right]} & 
%\leq \sum_{u^n}\sum_{\substack{v^n\in \TDeltaN(V)}}\cfrac{\lambdavB}{(1+\eta)}\Bigg\| \sqrt{\rho_{AB}^{\tensor n}}  \Big( \frac{1}{N_1} \sum_{\mu_1}\gammaCoeff A_{u^n}^{(\mu_1)} \tensor \cfrac{1}{\lambdavB}\bar{\Lambda}^B_{v^n}\Big)\sqrt{\rho_{AB}^{\tensor n}} \nonumber \\
%& \hspace{100pt} -  \sqrt{\rho_{AB}^{\tensor n}}\Big( \frac{1}{N_1} \sum_{\mu_1}\gammaCoeff A_{u^n}^{(\mu_1)} \tensor  \left (\CutOffB\sqrt{\rho_{B}}^{-1}\rhotildvB\sqrt{\rho_{B}}^{-1}\CutOffB\Big ) \right)\sqrt{\rho_{AB}^{\tensor n}} \Bigg\|_1 \nonumber \\
& \leq  \EE_{G,h_1}\left[\sum_{v^n}\cfrac{\lambdavB}{(1+\eta)}  \left \| \sqrt{\rho_{B}^{\tensor n}}  \left (\cfrac{1}{\lambdavB}\bar{\Lambda}^B_{v^n} - \sqrt{\rho_{B}}^{-1}\rhotildvB\sqrt{\rho_{B}}^{-1} \right)\sqrt{\rho_{B}^{\tensor n}} \right \|_1\right] \nonumber \\
& = \sum_{v^n\notin \TDeltaN(V)}\cfrac{\lambdavB}{(1+\eta)}  \left \|\rhohatvB \right \|_1 + \sum_{v^n\in \TDeltaN(V)}\cfrac{\lambdavB}{(1+\eta)}  \left \|\rhohatvB - \rhotildvB \right \|_1\leq \varepsilon_B + \epsilon_{Q_{24}}', \label{eq:gentleMeasurementVn}
\end{align}
\end{widetext}
where the second inequality above follows by using Lemma \ref{lem:Separate} and the fact that $\frac{1}{N_1} \sum_{\mu_1}\sum_{ u^n}\gammaCoeff A_{u^n}^{(\mu_1)} \leq I,$ 
% we use Lemma \ref{lem:Separate} and perform the following simplification 
% \begin{align}
% \sum_{ u^n}&  \left\| \sqrt{\rho_{AB}^{\tensor n}}  \left( \frac{1}{N_1} \sum_{\mu_1}A_{u^n}^{(\mu_1)} \tensor \left (\cfrac{1}{\lambdavB}\bar{\Lambda}^B_{v^n} - \sqrt{\rho_{B}}^{-1}\LambdavB\sqrt{\rho_{B}}^{-1} \right )\right )\sqrt{\rho_{AB}^{\tensor n}} \right \| \nonumber \\
% & \hspace{50pt} = \left \| \sqrt{\rho_{B}^{\tensor n}}  \left (\cfrac{1}{\lambdavB}\bar{\Lambda}^B_{v^n} - \sqrt{\rho_{B}}^{-1}\LambdavB\sqrt{\rho_{B}}^{-1} \right)\sqrt{\rho_{B}^{\tensor n}} \right \| \nonumber
% \end{align}
% This gives us
% \begin{align}
% \EE_{\mathbbm{C}_2}{\left[Q_{24}\right]} & \leq \sum_{v^n\in \TDeltaN(V)}\cfrac{\lambdavB}{(1+\eta)}  \left \| \sqrt{\rho_{B}^{\tensor n}}  \left (\cfrac{1}{\lambdavB}\bar{\Lambda}^B_{v^n} - \sqrt{\rho_{B}}^{-1}\LambdavB\sqrt{\rho_{B}}^{-1} \right)\sqrt{\rho_{B}^{\tensor n}} \right \| \nonumber \\
% & = \sum_{v^n\in \TDeltaN(V)}\cfrac{\lambdavB}{(1+\eta)}  \left \|\rhohatvB - \LambdavB \right \|\leq \cfrac{(1-\varepsilon')}{(1+\eta)} (2\sqrt{\varepsilon'_B} + 2\sqrt{\varepsilon''_B}) = \epsilon_{\scriptscriptstyle Q_{24}} \label{eq:gentleMeasurementVn}
% \end{align}
and the last inequality follows by applying the Average Gentle Measurement Lemma repeated
and $  \epsilon_{ Q_{24}}' \searrow 0 $ as $\delta \searrow 0$ 
(see (35) in \cite{wilde_e} for more details). This completes the proof for the term $ Q_{24} $.
% = \frac{(1-\varepsilon')}{(1+\eta)} (2\sqrt{\varepsilon'_B} + 2\sqrt{\varepsilon''_B}) $ with  $ \epsilon_{\scriptscriptstyle Q_{24}} \searrow 0 $ as $ \delta \searrow 0 $ and $ \varepsilon'_B = \varepsilon' + 2\sqrt{\varepsilon'} $ and $ \varepsilon''_B = 2\varepsilon' + 2\sqrt{\varepsilon'} $ 
Finally, we move onto considering $ Q_{25} $. Simplifying $ Q_{25} $ gives
\begin{align}
Q_{25} &\leq \frac{1}{N_1N_2} \!\!\sum_{\mu_1,\mu_2}\!\sum_{z^n}\!\sum_{u^n,v^n}\!\!\! P^n_{Z|W}(z^n|u^n+ v^n)   \bigg\| \sqrt{\rho_{AB}^{\tensor n}} \nonumber \\
& \hspace{15pt}  \left( \gammaCoeff A_{u^n}^{(\mu_1)} \tensor \left(\zetaCoeff\bar{B}_{v^n}^{(\mu_2)} -  \zetaCoeff B_{v^n}^{(\mu_2)}\right) \right)\sqrt{\rho_{AB}^{\tensor n}}\bigg\|_1 \nonumber \\
&\leq  \frac{1}{N_2}\sum_{\mu_2}  \sum_{v^n}\left\| \sqrt{\rho_{B}^{\tensor n}}   \!\left(\!\zetaCoeff\bar{B}_{v^n}^{(\mu_2)}\! -\!  \zetaCoeff B_{v^n}^{(\mu_2)}\!\right)\! \sqrt{\rho_{B}^{\tensor n}}\right\|_1 \nonumber \\ & = \tilde{S}_{23}, \nonumber 
\end{align}
where the first inequality uses traingle inequality and the second inequality uses Lemma \ref{lem:Separate} to remove the affect of approximating Alice's POVM on Bob's approximation, and $\tilde{S}_{23}$ is defined in \eqref{eq:S2simplification} in the proof of Proposition \ref{prop:Lemma for Tilde_S}. Therefore, we have the following: for any $\epsilon \in (0,1)$, any $\eta, \delta \in (0,1)$ sufficiently small, and any $n$ sufficiently large, we have $\EE[Q_{25}] \leq \epsilon,$ if $ S_2 \geq I(V;RA)_{\sigma_{2}} - S(V)_{\sigma_{3}} + \log{p}$.
% where $ \epsilon_{Q_{25}}(\delta), \delta_{Q_{25}} \searrow 0  $ as $ \delta \searrow 0.$ 
This completes the proof for $ Q_{25} $ and hence for all the terms corresponding to $ Q_{2} $.
% Now that we have shown that all the terms corresponding to $ Q_1 $ and $ Q_{2} $ can be made arbitrarily small for sufficiently large $ n $, this means that the term $ S_{11} $ can be successfully bounded. What remains now is the term $ S_{12} $, as defined in \eqref{eq:binningIsolation}.

\bibliography{references}

\end{document}